\documentclass[acmsmall,screen,nonacm]{acmart}
\AtBeginDocument{%
  }

\usepackage{mathtools}
\usepackage{amsmath,stmaryrd,mathpartir,amsthm}
\usepackage{xcolor}
\usepackage{graphicx}
\usepackage{cleveref}
\crefname{remark}{Remark}{Remarks}
\crefname{appendix}{appendix}{appendices}
\Crefname{appendix}{Appendix}{Appendices}
\usepackage{csquotes}
\usepackage[disable]{todonotes}
\usepackage{tabularx}
\usepackage{tensor}
\usepackage{mdframed}
\usepackage{multirow}
\usepackage{caption}
\usepackage{subcaption}
\usepackage{pdfpages}
\usepackage{thmtools}
\usepackage{nicefrac}
\usepackage{listings}

\lstdefinelanguage{pgcl}{morekeywords={uniform,observe,assert,int,uint,wp,wlp,ert,ett,ite,if,else,cost,skip,while,@fixpoint,@unroll,@kinduction}, sensitive=false}
\crefname{lstlisting}{listing}{listings}
\Crefname{lstlisting}{Listing}{Listings}

\usepackage{tikz}
\usetikzlibrary{positioning,arrows.meta,calc,shapes}

\tikzset{
	bdd node/.style={ellipse, draw=black!70, line width=0.45pt, fill=white, inner sep=3pt, minimum size=15pt},
	bdd leaf/.style={rectangle, rounded corners=1.4pt, draw=black!70, line width=0.45pt, fill=white, inner sep=3pt, minimum size=15pt},
	bdd edge0/.style={-{Stealth[length=1.9mm]}, dashed, line width=0.65pt, draw=black!75},
	bdd edge1/.style={-{Stealth[length=1.9mm]}, solid, line width=0.65pt, draw=black!75},
}

\usepackage[ruled, linesnumbered, nokwfunc]{algorithm2e}
\SetAlgoHangIndent{2em}
\SetKwProg{Fn}{def}{}{}
\SetKw{In}{in}
\SetFuncSty{textsc}

\SetKwFunction{apply}{Apply}
\SetKwFunction{obtain}{Obtain}
\SetKwFunction{solvesmt}{SolveSMT}
\SetKwFunction{prune}{Prune}
\SetKwFunction{substitute}{Substitute}
\SetKwFunction{minfunc}{Minimum}
\SetKwFunction{entails}{Entails}
\crefname{algocf}{Algorithm}{Algorithms}
\Crefname{algocf}{Algorithm}{Algorithms}
\theoremstyle{remark}
\newtheorem*{remark}{Remark}

\newcommand{\sj}[1]{\todo[noinline,color=green]{\footnotesize SJ\@: #1}}

\newcommand{\dbinline}[1]{\todo[inline,color=orange,caption={}]{\footnotesize DB\@: #1}}
\newcommand{\kb}[1]{\todo[noinline,color=cyan]{\footnotesize KB\@: #1}}
\newcommand{\kbinline}[1]{\todo[inline,color=cyan,caption={}]{\footnotesize KB\@: #1}}

\newcommand{\highlight}[1]{#1}
\newcommand{\setfont}[1]{\mathsf{#1}}

\newcommand{\toolfont}[1]{\textnormal{\textsc{#1}}}

\newcommand{\toolzt}{\toolfont{Z3}\xspace}
\newcommand{\toolcvc}{\toolfont{CVC5}\xspace}
\newcommand{\toolddsolve}{\toolfont{DDSolve}\xspace}
\newcommand{\smtlib}{\toolfont{SMT-LIB}\xspace}

\DeclareRobustCommand{\algname}[1]{\textnormal{\textsc{#1}}}
\DeclareRobustCommand{\ApplyAlg}{\algname{Apply}\xspace}
\DeclareRobustCommand{\ObtainAlg}{\algname{Obtain}\xspace}

\DeclareRobustCommand{\PruneAlg}{\algname{Prune}\xspace}
\DeclareRobustCommand{\SubstituteAlg}{\algname{Substitute}\xspace}
\DeclareRobustCommand{\MinimumAlg}{\algname{Minimum}\xspace}
\DeclareRobustCommand{\EntailsAlg}{\algname{Entails}\xspace}

\newcommand{\mylambda}[1]{\ensuremath{\lambda #1.\,}}

\newcommand{\true}{\mathsf{true}}
\newcommand{\false}{\mathsf{false}}

\newcommand{\Nats}{\mathbb{N}}

\newcommand{\Rats}{\mathbb{Q}}
\newcommand{\Reals}{\mathbb{R}}
\newcommand{\PosRats}{\Rats_{\geq 0}}
\newcommand{\PosReals}{\Reals_{\geq 0}}
\newcommand{\PosRealsInf}{\PosReals^{\infty}}

\newcommand{\proba}{\ensuremath{p}}

\newcommand{\reala}{\ensuremath{r}}

\newcommand{\eeq}{~{}={}~}
\newcommand{\monus}{\mathbin{\dot-}}

\newcommand{\sorts}{\mathsf{Types}}
\newcommand{\sortbool}{\ensuremath{\mathsf{Bool}}}
\newcommand{\sortnat}{\ensuremath{\mathsf{Nat}}}
\newcommand{\sortint}{\ensuremath{\mathsf{Int}}}
\newcommand{\sortrat}{\ensuremath{\mathsf{Rat}}}
\newcommand{\sortreal}{\ensuremath{\mathsf{Real}}}
\newcommand{\sorteureal}{\ensuremath{\mathsf{EUReal}}}

\newcommand{\sorta}{\ensuremath{\tau}}

\newcommand{\sortarray}{\ensuremath{\mathbf{Array}}}
\newcommand{\arraya}{\ensuremath{a}}
\newcommand{\funcstore}{\mathsf{store}}
\newcommand{\funcselect}{\mathsf{select}}

\newcommand{\foeq}{\mathbin{\dot=}}
\newcommand{\foneq}{\mathbin{\dot\neq}}

\newcommand{\vars}{\ensuremath{\mathsf{Vars}}}
\newcommand{\vara}{x}
\newcommand{\varb}{y}
\newcommand{\varc}{z}

\newcommand{\hastype}[2]{\ensuremath{#1 \colon #2}}

\newcommand{\functs}{\mathsf{Func}}
\newcommand{\funca}{\ensuremath{F}}
\newcommand{\funcb}{\ensuremath{G}}

\newcommand{\term}{\ensuremath{\mathsf{Terms}}}
\newcommand{\terma}{\ensuremath{t}}

\newcommand{\termop}{\mathcal{F}}

\newcommand{\fo}{\ensuremath{\mathsf{FO}}}
\newcommand{\at}{\ensuremath{\mathsf{AT}}}
\newcommand{\forma}{\ensuremath{\varphi}}
\newcommand{\formb}{\ensuremath{\psi}}

\newcommand{\tforall}[2]{\forall \hastype{#1}{#2}.\,}
\newcommand{\texists}[2]{\exists \hastype{#1}{#2}.\,}

\newcommand{\struct}{\mathfrak{A}}
\newcommand{\sem}[2]{\llbracket #1 \rrbracket^{#2}}
\newcommand{\propsem}[2]{\ensuremath{#2\left( #1 \right)}}
\newcommand{\vala}{\ensuremath{\sigma}}
\newcommand{\valsubst}[2]{\left[ #1 \mapsto #2 \right]}
\newcommand{\interprets}{\mathsf{Intrs}}
\newcommand{\varvals}{\mathcal{V}}
\newcommand{\inta}{\mathcal{I}}
\newcommand{\propinta}{\ensuremath{\eta}}

\newcommand{\valuea}{\ensuremath{w}}

\newcommand{\theory}{\mathcal{T}}
\newcommand{\theoryla}{\mathsf{LA}}
\newcommand{\theorylaarr}{\mathsf{LA+ARR}}

\newcommand{\sats}[1]{\ensuremath{\mathsf{SAT}_{\theory}\left( #1 \right)}}

\newcommand{\switchfun}{\mathbf{e}}
\newcommand{\switchfuncase}{\mapsfrom}
\newcommand{\switchfunequiv}{\ensuremath{\equiv}}
\newcommand{\switchfunterm}[2]{\ensuremath{#1 \left( #2 \right)}}
\newcommand{\switchITEsymbol}{\mathsf{ITE}}
\newcommand{\switchITE}[3]{\ensuremath{\switchITEsymbol \left( #1,#2,#3 \right)}}
\newcommand{\switchfunset}[1]{\mathsf{CaseExp}_{#1}}

\newcommand{\switchfunopsymbol}{\mathsf{op}}
\newcommand{\switchfunop}[1]{\ensuremath{\switchfunopsymbol\left[ #1 \right]}}
\newcommand{\switchfunsubst}[2]{\left[ #1 / #2 \right]}

\newcommand{\varord}{\ensuremath{\prec}}
\newcommand{\xadd}{\mathcal{X}}
\newcommand{\xaddb}{\mathcal{Y}}
\newcommand{\xaddc}{\mathcal{Z}}
\newcommand{\succtxadd}[1]{#1^{+}}
\newcommand{\succfxadd}[1]{#1^{-}}

\newcommand{\sizeof}[1]{\ensuremath{|{#1}|}}

\newcommand{\xnodes}{\ensuremath{V}}
\newcommand{\xnodea}{\ensuremath{v}}
\newcommand{\xroot}{\ensuremath{\xnodea_0}}
\newcommand{\xtnodes}{\ensuremath{T}}
\newcommand{\xnnodes}{\ensuremath{N}}
\newcommand{\xsuccsub}[1]{\ensuremath{\mathsf{succ}_{#1}}}
\newcommand{\xsucct}{\ensuremath{\mathsf{succ}_1}}
\newcommand{\xsuccf}{\ensuremath{\mathsf{succ}_0}}

\newcommand{\xlabn}{\ensuremath{L_N}}
\newcommand{\xlabt}{\ensuremath{L_T}}

\newcommand{\xswitchfun}[1]{\switchfun_{#1}}
\newcommand{\xswitchfuns}{\switchfun_{\xadd}}
\newcommand{\xswitchfunn}[2]{\ensuremath{\switchfun_{#1}^{#2}}}

\newcommand{\patha}{\ensuremath{\pi}}
\newcommand{\contexta}{\ensuremath{\theta}}
\newcommand{\pc}{\ensuremath{\mathsf{pc}}}
\newcommand{\pathcond}[2]{\ensuremath{\pc_{#1}(#2)}}

\newcommand{\pgcl}{\ensuremath{\setfont{pGCL}}}

\newcommand{\cc}{\ensuremath{C}}

\newcommand{\SKIP}{\ensuremath{\textnormal{\texttt{skip}}}}
\newcommand{\ABORT}{\ensuremath{\textnormal{\texttt{abort}}}}
\newcommand{\AssignSymbol}{\mathrel{\textnormal{\texttt{:=}}}}
\newcommand{\ASSIGN}[2]{\ensuremath{#1 \AssignSymbol #2}}

\newcommand{\KWTICK}{\ensuremath{\textnormal{\texttt{cost}}}}
\newcommand{\TICK}[1]{\ensuremath{\KWTICK\, \left(\, {#1} \,\right)}}
\newcommand{\tickrew}{r}

\newcommand{\COMPOSE}[2]{\ensuremath{{#1}{\,;}~ {#2}}}

\newcommand{\PCHOICE}[3]{\ensuremath{\left\{\, {#1} \,\right\}\mathrel{\left[\,#2\,\right]}\left\{\, {#3} \,\right\}}}

\newcommand{\ndchoicesymb}{\mathrel{\Box}}
\newcommand{\NDCHOICE}[2]{\ensuremath{\left\{\, {#1} \,\right\}\ndchoicesymb\left\{\, {#2} \,\right\}}}

\newcommand{\IFSYMBOL}{\ensuremath{\textnormal{\texttt{if}}}}
\newcommand{\IF}[1]{\ensuremath{\IFSYMBOL\,\left(\, {#1} \,\right)\,\{}}
\newcommand{\ELSESYMBOL}{\ensuremath{\textnormal{\texttt{else}}}}

\newcommand{\ITE}[3]{\ensuremath{\IFSYMBOL\,\left(\, {#1} \,\right)\,\left\{\, {#2} \,\right\}\,\ELSESYMBOL\,\left\{\, {#3} \,\right\}}}

\newcommand{\WHILESYMBOL}{\ensuremath{\textnormal{\texttt{while}}}}
\newcommand{\WHILE}[1]{\ensuremath{\WHILESYMBOL \left(\, {#1} \,\right)\left\{\right.}}
\newcommand{\WHILEDO}[2]{\ensuremath{\WHILESYMBOL \left(\, {#1} \,\right)\left\{\, {#2} \,\right\}}}

\newcommand{\ARRAYSTORE}[3]{\ensuremath{\ASSIGN{#1 \left[ #2 \right]}{#3}}}
\newcommand{\ARRAYREAD}[2]{#1 \left[ #2 \right]}
\newcommand{\OBSERVESYMBOL}{\ensuremath{\textnormal{\texttt{observe}}}}
\newcommand{\OBSERVE}[1]{\ensuremath{\OBSERVESYMBOL \left( #1 \right)}}

\newcommand{\lfpnospace}{\ensuremath{\mathsf{lfp}}}

\newcommand{\lfp}{\lfpnospace~}

\newcommand{\E}[1]{\mathbb{E}^{#1}} %expectations
\newcommand{\Es}{\E{\struct}} %expectations

\newcommand{\eleq}{\sqsubseteq}
\newcommand{\eeleq}{~{}\eleq{}~}

\newcommand{\einf}{\sqcap}
\newcommand{\bigesup}{\bigsqcup}
\newcommand{\bigeinf}{\bigsqcap}
\newcommand{\FF}{\ensuremath{X}}
\newcommand{\FG}{\ensuremath{Y}}
\newcommand{\FH}{\ensuremath{Z}}

\newcommand{\iverson}[1]{\ensuremath{\left[ #1 \right]}}

\newcommand{\expzero}{\mathbf{0}}

\newcommand{\wpsymbolnostruct}{\mathsf{wp}}
\newcommand{\wpsymbol}{\wpsymbolnostruct^{\struct}}
\newcommand{\wptrans}[1]{\wpsymbol\llbracket#1\rrbracket}
\newcommand{\wptransnostruct}[1]{\wpsymbolnostruct\llbracket#1\rrbracket}
\renewcommand{\wp}[2]{\wptrans{#1}\left(#2\right)}
\newcommand{\wpnostruct}[2]{\wptransnostruct{#1}\left(#2\right)}

\newcommand{\xwpsymbol}{\overline{\mathsf{wp}}}
\newcommand{\xwptrans}[1]{\xwpsymbol\llbracket#1\rrbracket}
\newcommand{\xwp}[2]{\xwptrans{#1}\left(#2\right)}

\newcommand{\wcharfun}[1]{\ensuremath{\Phi_{#1}}}
\newcommand{\kindop}[2]{\ensuremath{\Psi_{{#1}, {#2}}}}
\newcommand{\kindops}{\ensuremath{\Psi_{{\FF}, {\FG}}}}

\newcommand{\xwcharfun}[1]{\ensuremath{\overline{\Phi}_{#1}}}

\newcommand{\xkindops}{\ensuremath{\overline{\Psi}_{{\xadd}, {\xaddb}}}}
\newcommand{\wcharfuniter}[2]{\wcharfun{#1}^{#2}}

\newcommand{\kindopsiter}[1]{\ensuremath{\kindops^{#1}}}
\newcommand{\xwcharfuniter}[2]{\xwcharfun{#1}^{#2}}

\newcommand{\xkindopsiter}[1]{\ensuremath{\xkindops^{#1}}}

\newcommand{\ddsolve}{\textsc{DdSolve}}
\newcommand{\caesar}{\textsc{Caesar}}
\newcommand{\sylvan}{\textsc{Sylvan}}
\newcommand{\lace}{\textsc{Lace}}
\newcommand{\dice}{\textsc{Dice}}
\newcommand{\storm}{\textsc{Storm}}

\begin{document}

\title{Scalable Probabilistic Program Verification via~Typed~Extended~Decision~Diagrams}
%\subtitle{Scalable Verification of Expected Outcomes}

\author{Daniel Basg{\"o}ze}
\affiliation{%
    \institution{RWTH Aachen University}
    \city{Aachen}
    \country{Germany}
}
\email{daniel.basgoeze@rwth-aachen.de}

\author{Kevin Batz}
\orcid{0000-0001-8705-2564}
\affiliation{%
    \institution{Cornell University}
    \city{Ithaca}
    \state{NY}
    \country{USA}
}
\email{ksb239@cornell.edu}

\author{Sebastian Junges}
\orcid{0000-0003-0978-8466}
\affiliation{%
    \institution{Radboud University}
    \city{Nijmegen}
    \country{The Netherlands}
}
\email{sebastian.junges@ru.nl}

\author{Joost-Pieter Katoen}
\orcid{0000-0002-6143-1926}
\affiliation{%
    \institution{RWTH Aachen University}
    \city{Aachen}
    \country{Germany}
}
\email{katoen@cs.rwth-aachen.de}

\begin{abstract}
Weakest pre-expectations are the probabilistic program analogue to weakest preconditions in classical programs. 
Deductive verification approaches aim to establish bounds on these quantitative expectations. 
Their automation has been successful in analysing a variety of discrete probabilistic programs.
Key routines in that automation require reasoning about (partially unrolled) loops, however, the logical representation of weakest pre-expectations on such unrollings often explodes.
In this paper, we develop \emph{typed extended decision diagrams} (TEDDs), inspired by various extensions to binary decision diagrams. We demonstrate computing WPs represented as TEDDs, SMT-based pruning to further shrink their representation, and we lift some proof rules to operate on TEDDs. Finally, we demonstrate that TEDDs boost the scalability of deductive probabilistic program verification by orders of magnitude over the state of the art.   
\end{abstract}

\maketitle

\section{Introduction}

\kbinline{25 pages excl refs}
\kbinline{TODOs: nested loops, PC citations, emphasize log-log scale, cite PC members, continuous paper, explain why one case study performs bad, go through reviews, cite more unno so he stops being mad. EXIST, birkedal, ellora,...}
Probabilistic programs describe probabilistic models~\cite{DBLP:conf/icse/GordonHNR14} in a concise, programmatic manner. 
Applications include implementing randomized algorithms, solving complex Bayesian inference tasks, and modeling network protocols~\cite{DBLP:conf/pldi/GehrMTVWV18,DBLP:conf/esop/FosterKMR016,DBLP:conf/popl/SmolkaKFK017,DBLP:journals/corr/SmolkaKKFK017,DBLP:conf/pldi/SmolkaKKFHK019,DBLP:journals/pacmpl/GiannarakisSW21}.
\emph{Deductively} verifying  probabilistic programs means establishing properties of a program's output distributions by means of logical inference.
Approaches come in different flavours.
Relational, coupling-based, and assertion-based approaches provide expressive logics for reasoning about probabilistic programs~\cite{DBLP:journals/pacmpl/BartheHL20,DBLP:conf/popl/BartheGHS17,DBLP:conf/popl/BartheKOB12,DBLP:journals/pacmpl/BaoDF25,DBLP:journals/pacmpl/HoWR26}.
\emph{Mechanized} approaches often target feature-rich higher-order programs~\cite{DBLP:conf/cpp/MarionneauB0B26,DBLP:journals/corr/abs-2604-12713,DBLP:journals/pacmpl/HaselwarterLAGTB25,DBLP:journals/corr/abs-2511-10135,DBLP:journals/pacmpl/StassenMZAB25,DBLP:journals/pacmpl/Li0GHTB25,DBLP:journals/pacmpl/Gregersen0HTB24,DBLP:journals/pacmpl/GregersenAHTB24,DBLP:journals/pacmpl/0001HMLGTB24,DBLP:journals/pacmpl/HaselwarterLMG024,DBLP:journals/pacmpl/AguirreB23,DBLP:journals/pacmpl/0001024}, whereas \emph{automated} approaches~\cite{kind_cav,DBLP:conf/tacas/BatzCJKKM23,DBLP:journals/pacmpl/SchroerBKKM23,DBLP:conf/qest/GretzKM13,DBLP:conf/sas/KatoenMMM10,DBLP:conf/cav/BartheEFH16,DBLP:journals/pacmpl/AlbarghouthiH18,DBLP:conf/cav/AlbarghouthiH18,DBLP:conf/cav/BaoTPHR22,DBLP:journals/fmsd/BaoTPHR25,DBLP:journals/pacmpl/SusagLHR22,DBLP:journals/pacmpl/SmithHA19,DBLP:journals/pacmpl/ChatterjeeGMZ24,DBLP:conf/cav/SunFCG23,DBLP:conf/cav/ChatterjeeGMZ22,DBLP:conf/pldi/WangS0CG21,DBLP:journals/pacmpl/Huang0CG19,DBLP:conf/pldi/Wang0GCQS19,DBLP:conf/cav/ChatterjeeFG16,DBLP:journals/pacmpl/LiLHW0025,DBLP:conf/pldi/ChenH20,DBLP:journals/corr/abs-2504-04132,DBLP:journals/pacmpl/AvanziniMS23,DBLP:journals/pacmpl/AvanziniMS20,DBLP:journals/corr/abs-1908-11343,DBLP:conf/tacas/ChistikovDM15,DBLP:journals/acta/ChistikovDM17} target first-order programs popular in probabilistic modeling.

In this paper, we build upon the weakest pre-expectation ($\wpsymbolnostruct$) calculus for first-order programs, a prominent generalization of Dijkstra's weakest pre-conditions~\cite{dijkstra_discipline} to the probabilistic setting. Predicates, i.e., Boolean-valued functions, are replaced by Real-valued functions --- \enquote{quantitative assertions} --- for capturing properties like the probability of terminating in a desired program state. \caesar~\cite{DBLP:journals/pacmpl/SchroerBKKM23} is among the state of the art  tools in automated reasoning using the $\wpsymbolnostruct$ calculus. However, it suffers from scalability issues:  For various  seemingly simple programs, the representations of quantitative assertions are prohibitively large.  
%While the $\wpsymbolnostruct$ calculus is state-of-the-art in automated deductive reasoning \kb{cite}, current implementations suffer from scalability issues:  For various  seemingly simple programs, the representations of quantitative assertions are prohibitively large.  
%%
Inspired by the widespread success of decision diagrams (DDs) in formal verification~\cite{DBLP:journals/iandc/BurchCMDH92,DBLP:books/daglib/0020348}%or concisely representing Boolean (or more general) functions, 
, this paper introduces a novel notion of DDs to represent quantitative assertions and operate on them. %alongise efficient operations for concisely representing quantitative assertions of infinite-state probabilistic programs over rich domains. In fact, 
Our extensive experimental evaluation demonstrates superior scalability when compared to state-of-the-art implementations of the $\wpsymbolnostruct$ calculus. We scale, e.g., the verification of a distributed network protocol from 5 to 25 participants and the verification of array routines from 5 to hundreds of entries in the array.

\smallskip\noindent\emph{Large assertions.}
The state-of-the-art in automated $\wpsymbolnostruct$-based deductive verification is based on expressing (comparisons on) the quantitative assertions as Satisfiability Modulo Theories (SMT, \cite{DBLP:series/faia/BarrettSST21}) formulas and discharge proving the validity of these formulas with SMT solvers.  The representation of quantitative assertions tends to grow exponentially in the number of conditionals and random choices in a program. This leads to challenges when verifying \emph{possibly unbounded probabilistic loops}. Both Hoare-style invariant-based reasoning principles and fixed point iteration-based principles rely on \emph{loop unrolling}. This is problematic for the representation of the quantitative assertions, as one essentially glues ever more conditionals and random choices together. %We put emphasis on verifying infinite-state probabilistic loops via the aforementioned reasoning principles.\sj{Do we really need this sentence?}
\smallskip\noindent\emph{DD-based reasoning.}
For verifying finite-state \emph{non}-probabilistic programs, binary decision diagrams (BDDs)~\cite{DBLP:journals/tc/Bryant86} have a long tradition~\cite{DBLP:journals/iandc/BurchCMDH92,DBLP:books/daglib/0020348} as they concisely represent sets of states and allow for efficient operations on these sets. This success has motivated first-order DDs~\cite{DBLP:journals/jlp/GrooteT03} and the use of algebraic decision diagrams (ADDs, or MTBDDs) in probabilistic model checking. ADDs concisely map \emph{finite} sets of states to numbers such as reachability probabilities in Markov chains~\cite{DBLP:conf/icalp/BaierCHKR97}.
For \emph{infinite-state probabilistic programs}, however, we have to map \emph{infinite} sets of states to numbers. For that, we take inspiration from probabilistic planning, where value functions of infinite-state Markov decision processes (MDPs) --- maps from MDP states to numbers --- have to be represented concisely. Towards this end, Sanner et al.\ have introduced  \emph{extended algebraic decision diagrams} (XADDs)~\cite{DBLP:conf/uai/SannerDB11} and \emph{first-order algebraic decision diagrams} (FOADDs)~\cite{DBLP:journals/ai/SannerB09}. XADDs specifically aim to represent value functions with discrete and continuous variables. FOADDs, on the other hand, have been specifically developed to target first-order MDPs, a type of MDP where states are defined by assigning truth values to a set of first-order predicates.

\smallskip\noindent\emph{Our approach.}
We elevate the ideas of XADDs and FOADDs to the verification of infinite-state probabilistic programs over rich domains such as arrays and unbounded integer variables, proposing  \emph{typed extended decision diagrams} (TEDDs). They generalize existing notions of decision diagrams in various ways: First, for generically capturing rich domains, inner nodes are labeled by \emph{many-sorted first-order} formulae. Second, leafs are labeled by \emph{typed expressions}, i.e., while BDDs represent functions from propositional interpretations to truth values, a $\sorta$-TEDD for a given type $\sorta$ maps \emph{first-order interpretations} to \emph{values of type $\sorta$}. Third, this formal basis enables \emph{theory-aware SMT-based pruning} mechanisms for keeping TEDDs small by removing unreachable sub-branches. Fourth, 
we implement TEDDs on top of the efficient BDD library  \sylvan~\cite{DBLP:journals/sttt/DijkP17}, which supports effective and safe multithreading. This has a highly beneficial practical impact:
Our experiments show that exploiting \sylvan's multithreading capabilities regularly leads to a superlinear speedup. 

\smallskip\noindent\emph{Contributions.} In summary, we contribute the following: (1)~We introduce \emph{typed extended decision diagrams} based on many-sorted first order logic and provide various algorithms on TEDDs. (2)~We lift the $\wpsymbolnostruct$ calculus and related proof rules for infinite-state probabilistic loops to TEDDs. (3)~We provide a flexible open-source deductive probabilistic program verifier based on TEDDs, supporting verification and exact inference on infinite-state discrete probabilistic programs with conditioning or non-determinism, costs, arrays, and nonlinear expectations, which significantly beats the state-of-the-art tool \caesar{} on many benchmarks.\sj{I hate the alignment of TEDDs on the right side of the pdf...}

\section{Deductive Verification with TEDDs: A Bird's Eye View}
\label{sec:overview}

We illustrate the type of quantitative properties and probabilistic programs studied in this paper and highlight the benefits of our novel TEDD-based techniques by examples.

\begin{figure}[t]
	%\begin{tikzpicture}[node distance=20mm and 30mm]
	%	% Leaves
	%	\node[bdd leaf] (L0) {$0$};
	%	\node[bdd leaf, right=of L0] (L1) {$1$};
	%	% A decision node x
	%	\node[bdd node, above=of $(L0)!0.5!(L1)$] (x) {$x$};
	%	\draw[bdd edge0] (x) -- (L0) node[midway, left] {};
	%	\draw[bdd edge1] (x) -- (L1) node[midway, right] {};
	%\end{tikzpicture}
	\centering
	%	\begin{subfigure}[b]{.2\textwidth}
		%		\centering
		%		\begin{align*}
			%			\begin{cases}
				%				\vara &\switchfuncase \ARRAYREAD{A}{0} \leq 0 \\
				%				%
				%				%
				%				\nicefrac{1}{2}\cdot \vara &\switchfuncase  \ARRAYREAD{A}{0} > 0~.
				%			\end{cases}
			%		\end{align*}
		%		%
		%		%
		%		%
		%		\caption{BLABLA.}
		%		\label{fig:ex_xadd_overview_1}
		%	\end{subfigure}
	%	%
	%	\hfill
	%
	\begin{subfigure}[b]{.39\textwidth}
	\centering

	\begin{align*}
	&\ASSIGN{i}{0}\, ; \\
	&\WHILE{i < N} \\
	&\qquad \IF{\ARRAYREAD{A}{i} > 0} \\
	&\qquad \qquad \PCHOICE{\ASSIGN{\vara}{0}}{\nicefrac{1}{2}}{\SKIP}\, \\
	&\qquad \}\, ; \\
	&\qquad \ASSIGN{i}{i+1} \\
	%
%	&\qquad \ELSE \\
%	%
%	&\qquad \qquad \SKIP \\
%	%
	%
	&\}
\end{align*}
\caption{Probabilistic Array Program $\cc_N$.}
\label{fig:overview:ccn}
\end{subfigure}
%
%	\begin{subfigure}[b]{.39\textwidth}
%		\centering
%		%
%		%
%	%	\scalebox{0.75}{
%		\begin{align*}
%			\begin{cases}
%				\vara &\switchfuncase \ARRAYREAD{A}{0} \leq 0 \wedge \ARRAYREAD{A}{1} \leq 0  \\
%				%
%				%
%				\nicefrac{1}{2}\cdot \vara &\switchfuncase  \ARRAYREAD{A}{0} \leq 0 \wedge \ARRAYREAD{A}{1} > 0 \\ 
%				%
%				\nicefrac{1}{2}\cdot \vara &\switchfuncase  \ARRAYREAD{A}{0} > 0 \wedge \ARRAYREAD{A}{1} \leq 0 \\ 
%				%
%				\nicefrac{1}{4}\cdot \vara &\switchfuncase  \ARRAYREAD{A}{0} > 0 \wedge \ARRAYREAD{A}{1} > 0 \\ 
%				%
%			\end{cases}
%		\end{align*}
%%	}
%		%
%		%
%		%
%		\caption{Case expression for $N=2$.}
%		\label{fig:overview:case_exprs1}
%	\end{subfigure}
	%
	%
	%	\hfill
	%	%
	\begin{subfigure}[b]{.59\textwidth}
		\centering
		%
		% 
		%\scalebox{0.75}{
		\begin{align*}
			\begin{cases}
				\vara &\switchfuncase \ARRAYREAD{A}{0} \leq 0 \wedge \ARRAYREAD{A}{1} \leq 0  \wedge  \ARRAYREAD{A}{2} \leq 0 \\
				\nicefrac{1}{2}\cdot \vara &\switchfuncase \ARRAYREAD{A}{0} \leq 0 \wedge \ARRAYREAD{A}{1} \leq 0  \wedge  \ARRAYREAD{A}{2} > 0 \\
				\nicefrac{1}{2}\cdot \vara &\switchfuncase \ARRAYREAD{A}{0} \leq 0 \wedge \ARRAYREAD{A}{1} > 0  \wedge  \ARRAYREAD{A}{2} \leq 0 \\
				\nicefrac{1}{2}\cdot \vara &\switchfuncase \ARRAYREAD{A}{0} > 0 \wedge \ARRAYREAD{A}{1} \leq 0  \wedge  \ARRAYREAD{A}{2} \leq 0 \\
				\nicefrac{1}{4}\cdot \vara  &\switchfuncase \ARRAYREAD{A}{0} \leq 0 \wedge \ARRAYREAD{A}{1} > 0  \wedge  \ARRAYREAD{A}{2} > 0 \\
				\nicefrac{1}{4}\cdot \vara  &\switchfuncase \ARRAYREAD{A}{0} > 0 \wedge \ARRAYREAD{A}{1} \leq 0  \wedge  \ARRAYREAD{A}{2} > 0 \\
				\nicefrac{1}{4}\cdot \vara  &\switchfuncase \ARRAYREAD{A}{0} > 0 \wedge \ARRAYREAD{A}{1} > 0  \wedge  \ARRAYREAD{A}{2} \leq 0 \\
				\nicefrac{1}{8}\cdot \vara  &\switchfuncase \ARRAYREAD{A}{0} > 0 \wedge \ARRAYREAD{A}{1} > 0  \wedge  \ARRAYREAD{A}{2} > 0 \\
			\end{cases}
		\end{align*}
		%}
		%
		\caption{Case expression for $N=3$.}
		\label{fig:overview:case_exprs2}
	\end{subfigure}

	\caption{Program $\cc_N$ and case expression $\wpnostruct{C_N}{\vara}$ for $N=3$.}
		\label{fig:overview:case_exprs}
\end{figure}
\subsection{TEDDs: Concise Representations of Expected Outcomes}
Consider the \emph{probabilistic array programs} $\,\cc_N$ with interleaved conditionals and probabilistic choices in Fig.~\ref{fig:overview:ccn}.
%
%\begin{align*}
%	&\ASSIGN{i}{0}\, ; \\
%	%
%	&\WHILE{i < N} \\
%	%
%	&\qquad \IF{\ARRAYREAD{A}{i} > 0} \\
%	%
%	&\qquad \qquad \PCHOICE{\ASSIGN{\vara}{0}}{\nicefrac{1}{2}}{\SKIP}\, \\
%	%
%	&\qquad \}\, ; \\
%	%
%	&\qquad \ASSIGN{i}{i+1} \\
%	%
%%	&\qquad \ELSE \\
%%	%
%%	&\qquad \qquad \SKIP \\
%%	%
%	%
%	&\}
%\end{align*}
%
We use $N \in \Nats$ as a fixed constant. The variables $i$ and $x$ are of type $\Nats$ and $A$ is an $\Nats$-index array over $\Nats$. 
 We iterate over $i$ from $0$ to $N-1$ and if $A$ at index $i$ exceeds $0$, we flip a fair coin using the \emph{probabilistic choice} construct $\PCHOICE{\ldots}{\nicefrac{1}{2}}{\ldots}$. If the coin lands heads (left branch), we set the variable $\vara$ to $0$ and otherwise  we do nothing. Verifying properties about this program thus requires reasoning about interleaved conditional \emph{and} probabilistic choices. 

In this paper, we reason about \emph{expected outcomes} of probabilistic programs. For $\cc_N$ above, we are interested in the \emph{expected final value of $\vara$} obtained upon termination of $\cc_N$. 
%How does that quantity look like? 
This expected outcome depends on the initial program state, in particular the value of $A$, as that value determines the control flow. A bit more formal, we are interested in determining the so-called \emph{weakest pre-expectation} $\wpnostruct{\cc_N}{\vara}$ of $\cc_N$ w.r.t.\ $\vara$ --- a function from initial program states to the expected value of $x$ upon termination. Details are given in \Cref{sec:wp}.
For $N=1$, this weakest pre-expectation is given by
\begin{align*}
	\wpnostruct{\cc_1}{\vara}
	\eeq 
	&
	\begin{cases}
		\vara &\switchfuncase \ARRAYREAD{A}{0} \leq 0 \\
		\nicefrac{1}{2}\cdot \vara &\switchfuncase  \ARRAYREAD{A}{0} > 0~.
	\end{cases}
\end{align*}
This pre-expectation states that if (initially) $\ARRAYREAD{A}{0} \leq 0$, the expected final value of $\vara$ is its initial value, and if $ \ARRAYREAD{A}{0} > 0$, the expected final value of $\vara$ is $\nicefrac{1}{2}$ times its initial value. 

The state-of-the-art implementations of weakest pre-expectations represent $\wpnostruct{\cc_N}{\vara}$ essentially via \emph{case expressions}, i.e., syntactic constructs denoting complete case distinctions on the initial state, just like the example above.
%$ \varvals \to \PosRealsInf$ of $\cc_N$ w.r.t.\ $\vara$},
%
%a function $\FG \colon \varvals \to \PosRealsInf$, 
%where $\varvals$ is the set of all program states and $\PosRealsInf$ is the set of non-negative reals extended by $\infty$. This weakest pre-expectation maps each initial state to the sought-after expected value in the sense that
%
%\[
 %\FG(\vala) 
 %\wpnostruct{\cc_N}{\vara}(\vala)
 %\qquad\quad\eeq\quad\quad 
 %\substack{
% 	\text{\normalsize expected final value of $\vara$ obtained from} \\
% 	\text{\normalsize executing $\cc$ on the initial program state $\vala$~.}}
%\]
%
%State-of-the-art implementations \kb{cite many} of weakest pre-expectations represent $\wpnostruct{\cc}{\vara}$ essentially via \emph{case expressions}, i.e., syntactic constructs denoting complete case distinctions on the initial state~$\vala$. For $N=1$, this case expression is
%%
%\begin{align*}
%	\wpnostruct{\cc}{\vara}
%	%
%	\eeq 
%	&
%	\begin{cases}
%		\vara &\switchfuncase \ARRAYREAD{A}{0} \leq 0 \\
%		%
%		%
%		\nicefrac{1}{2}\cdot \vara &\switchfuncase  \ARRAYREAD{A}{0} > 0~.
%	\end{cases}
%\end{align*}
%%
%That is, if (initially) $\ARRAYREAD{A}{0} \leq 0$, the expected final value of $\vara$ is its initial value, and if $ \ARRAYREAD{A}{0} > 0$, the expected final vlaue of $\vara$ is $\nicefrac{1}{2} \cdot x + \nicefrac{1}{2} \cdot 0$ times its initial value.
 %due to the probabilistic choice in $\cc$. 
 Case expressions often grow \emph{exponentially} in, e.g., the number of conditionals encountered during program execution. Consider $\cc_N$ for increasing $N$: the expected final value of $\vara$ depends on ever more entries of the array $A$, which results in case expressions consisting of $2^N$ cases. \Cref{fig:overview:case_exprs2} depicts the case expressions for $N=3$. This is intractable.

\begin{figure}[t]
%\begin{tikzpicture}[node distance=20mm and 30mm]
%	% Leaves
%	\node[bdd leaf] (L0) {$0$};
%	\node[bdd leaf, right=of L0] (L1) {$1$};
%	% A decision node x
%	\node[bdd node, above=of $(L0)!0.5!(L1)$] (x) {$x$};
%	\draw[bdd edge0] (x) -- (L0) node[midway, left] {};
%	\draw[bdd edge1] (x) -- (L1) node[midway, right] {};
%\end{tikzpicture}
\centering
%\begin{subfigure}[t]{.22\textwidth}
%	\centering
%	\scalebox{0.65}{
%\begin{tikzpicture}[node distance=6mm and 3mm]
%	% Leaves
%	
%	\node[bdd node] (x0) at (0,0) {$\ARRAYREAD{A}{0} > 0$};
%	\node[bdd leaf] (y) at (-1.5,-1.5) {$\vara$};
%	\draw[bdd edge0] (x0) -- (y)  {};
%	
%	\node[bdd leaf] (x1) at (1.5,-1.5) {$\nicefrac{1}{2}\cdot \vara$};
%	\draw[bdd edge1] (x0) -- (x1)  {};
%	
%%	\node[bdd leaf, below left = of x1] (y1) {$\varb+1$};
%%	\draw[bdd edge1] (x1) -- (y1)  {};
%%	
%%	\node[bdd node, below right = of x1] (x2) {$\vara=2$};
%%	\draw[bdd edge0] (x1) -- (x2)  {};
%%	\node[bdd leaf, below left = of x2] (y2) {$\varb+2$};
%%	\node[bdd leaf, below right = of x2] (y3) {$\varb+3$};
%%	\draw[bdd edge1] (x2) -- (y2)  {};
%%	\draw[bdd edge0] (x2) -- (y3)  {};
%	
%	%\node[bdd leaf, below left=of x0] (y) {$\varb$};
%	
%\end{tikzpicture}
%}
%\caption{$N=1$.}
%\label{fig:ex_xadd_overview_1}
%\end{subfigure}
%
%\hfill
%%
\begin{subfigure}[t]{.39\textwidth}
	\centering
	\scalebox{0.75}{
	\begin{tikzpicture}[node distance=0mm and 0mm]
		 
		 	\node[bdd node] (x0) at (0,0) {$\ARRAYREAD{A}{0} > 0$};
		 	
		 	\node[bdd node] (x10) at (-1.5,-1.5) {$\ARRAYREAD{A}{1} > 0$};
		 	\node[bdd node] (x11) at (1.5,-1.5) {$\ARRAYREAD{A}{1} > 0$};
		 	
		 	\node[bdd leaf] (x20) at (-3.0,-3.0) {$\vara$};
		 	\node[bdd leaf] (x21) at (0,-3.0) {$\nicefrac{1}{2}\cdot\vara$};
		 	\node[bdd leaf] (x22) at (3.0,-3.0) {$\nicefrac{1}{4}\cdot\vara$};
		 	
		 	%\node[bdd node] (y1) {$\ARRAYREAD{A}{1} > 0$};
		 	\draw[bdd edge0] (x0) -- (x10)  {};
		 	\draw[bdd edge1] (x0) -- (x11)  {};
		 	
		 	\draw[bdd edge0] (x10) -- (x20)  {};
		 	\draw[bdd edge1] (x10) -- (x21)  {};
		 	
		 	\draw[bdd edge0] (x11) -- (x21)  {};
		 	\draw[bdd edge1] (x11) -- (x22)  {};
		 	
		 	%\node[bdd node] (y2) {$\ARRAYREAD{A}{1} > 0$};
		 	%\draw[bdd edge1] (x0) -- (y2)  {};
		 	
		 	%\node[bdd leaf] (y) {$\vara$};
		 	%\node[bdd leaf] (y) {$\nicefrac{1}{2}\cdot \vara$};

	\end{tikzpicture}
}
	\caption{TEDD for $N=2$.}
	\label{fig:overview:tedds_1}
\end{subfigure}
\hfill
\begin{subfigure}[t]{.59\textwidth}
	\centering
\scalebox{0.75}{
\begin{tikzpicture}[node distance=0mm and 0mm]
	
	\node[bdd node] (x0) at (0,0) {$\ARRAYREAD{A}{0} > 0$};
	
	\node[bdd node] (x10) at (-1.5,-1.5) {$\ARRAYREAD{A}{1} > 0$};
	\node[bdd node] (x11) at (1.5,-1.5) {$\ARRAYREAD{A}{1} > 0$};
	
	\node[bdd node] (x20) at (-3.0,-3.0) {$\ARRAYREAD{A}{2} > 0$};
	\node[bdd node] (x21) at (0,-3.0) {$\ARRAYREAD{A}{2} > 0$};
	\node[bdd node] (x22) at (3.0,-3.0) {$\ARRAYREAD{A}{2} > 0$};
	
	\node[bdd leaf] (x30) at (-4.5,-4.5) {$\vara$};
	\node[bdd leaf] (x31) at (-1.5,-4.5) {$\nicefrac{1}{2}\cdot\vara$};
	\node[bdd leaf] (x32) at (1.5,-4.5) {$\nicefrac{1}{4}\cdot\vara$};
	\node[bdd leaf] (x33) at (4.5,-4.5) {$\nicefrac{1}{8}\cdot\vara$};
	
	\draw[bdd edge0] (x0) -- (x10)  {};
	\draw[bdd edge1] (x0) -- (x11)  {};
	
	\draw[bdd edge0] (x10) -- (x20)  {};
	\draw[bdd edge1] (x10) -- (x21)  {};
	
	\draw[bdd edge0] (x11) -- (x21)  {};
	\draw[bdd edge1] (x11) -- (x22)  {};

	\draw[bdd edge0] (x20) -- (x30)  {};
	\draw[bdd edge1] (x20) -- (x31)  {};
	
	\draw[bdd edge0] (x21) -- (x31)  {};
	\draw[bdd edge1] (x21) -- (x32)  {};
	
	\draw[bdd edge0] (x22) -- (x32)  {};
	\draw[bdd edge1] (x22) -- (x33)  {};
	
	%\node[bdd node] (y1) {$\ARRAYREAD{A}{1} > 0$};
	%\draw[bdd edge0] (x0) -- (y1)  {};
	
	%\node[bdd node] (y2) {$\ARRAYREAD{A}{1} > 0$};
	%\draw[bdd edge1] (x0) -- (y2)  {};
	
	%\node[bdd leaf] (y) {$\vara$};
	%\node[bdd leaf] (y) {$\nicefrac{1}{2}\cdot \vara$};
\end{tikzpicture}
}
\caption{TEDD for $N=3$.}
\label{fig:overview:tedds_2}
\end{subfigure}
\caption{TEDDs representing $\wpnostruct{C_N}{\vara}$ for the program $\cc_N$ in \Cref{fig:overview:ccn}.}
\label{fig:overview:tedds}
\end{figure}
This paper introduces \emph{typed extended decision diagrams} (\emph{TEDDs}) to mitigate the growth of case expressions, thereby rendering the verification of probabilistic programs more scalable. 
TEDDs extend classical BDDs~\cite{DBLP:journals/tc/Bryant86}, where now inner nodes are labeled by many-sorted first order logic formulae and terminal nodes are labelled by some typed expressions. In our case, we use expressions of type $\PosRealsInf$ and logical formulae over a signature suitable for reasoning about $\Nats$-indexed arrays. Such TEDDs express our case expressions, in particular, 
 the TEDD  in \Cref{fig:overview:tedds_1} represents $\wpnostruct{\cc_2}{\vara}$, 
 %--- in our case --- labeled by expressions of type $\PosRealsInf$. The semantics of a TEDD is a case expression. 
 %--- in our case over a signature suitable for reasoning about $\Nats$-indexed arrays. 
 while the TEDD from \Cref{fig:overview:tedds_2} concisely represents the case expression from \Cref{fig:overview:case_exprs2}: Every path from the root to a terminal node represents one case obtained from conjoining the (negated, if a $\dashrightarrow$ successor is taken) formulae encountered along a path from the root to that terminal node.
 
Indeed, while the case expressions grow exponentially, the TEDDs grow quadratically. 
%Let us now compare the case expressions from \Cref{fig:overview:case_exprs} with the corresponding TEDDs from \Cref{fig:overview:tedds}. Whereas the size of the case expressions grows \emph{exponentially} in $N$, the TEDDs grow much more slowly: The number of inner nodes is in $\mathcal{O}(N^2)$ and the number of terminale nodes is $N+1$.
 The reason is that the decision diagram-based structure \emph{naturally captures symmetries} in the function $\wpnostruct{\cc_N}{\vara}$. The TEDD in \Cref{fig:overview:tedds_1} (for $N=2$) makes it explicit that the expected final values of $\vara$ coincide whenever (i) $\ARRAYREAD{A}{0} \leq 0 \wedge \ARRAYREAD{A}{1} > 0$ and (ii) $\ARRAYREAD{A}{0} > 0 \wedge \ARRAYREAD{A}{1} \leq 0$, which follows as for both cases, the program $\cc_2$ performs \emph{exactly one} coin flip. As $N$ increases, more of these symmetries appear (cf.\ \Cref{fig:overview:tedds_2}), resulting in concise TEDDs for huge case expressions. 

To fully exploit these effects, we develop a weakest pre-expectation calculus \emph{operating entirely on TEDDs}. We implement this calculus in a novel software tool and demonstrate its superior performance over state-of-the-art implementations by means of an extensive experimental evaluation.

\subsection{First Order Theory-Aware Pruning: Keeping TEDDs Small}
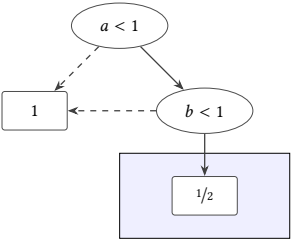
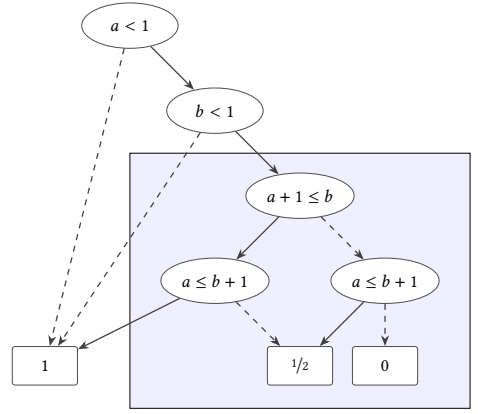
\begin{figure}[t]
%\begin{tikzpicture}[node distance=20mm and 30mm]
%	% Leaves
%	\node[bdd leaf] (L0) {$0$};
%	\node[bdd leaf, right=of L0] (L1) {$1$};
%	% A decision node x
%	\node[bdd node, above=of $(L0)!0.5!(L1)$] (x) {$x$};
%	\draw[bdd edge0] (x) -- (L0) node[midway, left] {};
%	\draw[bdd edge1] (x) -- (L1) node[midway, right] {};
%\end{tikzpicture}
\centering
%\begin{subfigure}[t]{.22\textwidth}
%	\centering
%	\scalebox{0.65}{
%\begin{tikzpicture}[node distance=6mm and 3mm]
%	% Leaves
%	
%	\node[bdd node] (x0) at (0,0) {$\ARRAYREAD{A}{0} > 0$};
%	\node[bdd leaf] (y) at (-1.5,-1.5) {$\vara$};
%	\draw[bdd edge0] (x0) -- (y)  {};
%	
%	\node[bdd leaf] (x1) at (1.5,-1.5) {$\nicefrac{1}{2}\cdot \vara$};
%	\draw[bdd edge1] (x0) -- (x1)  {};
%	
%%	\node[bdd leaf, below left = of x1] (y1) {$\varb+1$};
%%	\draw[bdd edge1] (x1) -- (y1)  {};
%%	
%%	\node[bdd node, below right = of x1] (x2) {$\vara=2$};
%%	\draw[bdd edge0] (x1) -- (x2)  {};
%%	\node[bdd leaf, below left = of x2] (y2) {$\varb+2$};
%%	\node[bdd leaf, below right = of x2] (y3) {$\varb+3$};
%%	\draw[bdd edge1] (x2) -- (y2)  {};
%%	\draw[bdd edge0] (x2) -- (y3)  {};
%	
%	%\node[bdd leaf, below left=of x0] (y) {$\varb$};
%	
%\end{tikzpicture}
%}
%\caption{$N=1$.}
%\label{fig:ex_xadd_overview_1}
%\end{subfigure}
%
%\hfill
%%
\begin{subfigure}[b]{.49\textwidth}
	\centering
	\begin{align*}
		&\WHILE{a < N \wedge b <N} \\
		&\qquad \PCHOICE{\ASSIGN{a}{a+1}}{\nicefrac{1}{2}}{\ASSIGN{b}{b+1}}\, ; \\
		&\qquad \OBSERVE{a \leq b} \\
		&\}
	\end{align*}

	\scalebox{0.75}{
	\begin{tikzpicture}[
		node distance=0mm and 0mm,
		bdd node/.style={ellipse, draw=black!70, line width=0.45pt, fill=white, inner sep=2.5pt, minimum width=17mm, minimum height=8mm, font=\small},
		bdd leaf/.style={rectangle, rounded corners=1.4pt, draw=black!70, line width=0.45pt, fill=white, inner sep=1.7pt, minimum width=11.5mm, minimum height=6.8mm, font=\small}
	]	 
		 	\draw[solid, fill=blue!20, fill opacity=0.3] (0,-2.25) rectangle (3.0,-3.75);
		 
		 	\node[bdd node] (x0) at (0,0) {$a<1$};
		 	
		 	\node[bdd leaf] (x10) at (-1.5,-1.5) {$1$};
		 	\node[bdd node] (x11) at (1.5,-1.5) {$b<1$};

		 	\node[bdd leaf] (x31) at (1.5,-3.0) {$\nicefrac{1}{2}$};
		 	
		 	\draw[bdd edge0] (x0) -- (x10)  {};
		 	\draw[bdd edge1] (x0) -- (x11)  {};
		 	\draw[bdd edge0] (x11) -- (x10)  {};
		 	\draw[bdd edge1] (x11) -- (x31)  {};
	\end{tikzpicture}
}
	\caption{A program with conditioning and TEDD $\xadd_1$.}
	\label{fig:overview:tedds:observe1}
\end{subfigure}
\hfill
\begin{subfigure}[b]{.49\textwidth}
	\centering
\scalebox{0.75}{
\begin{tikzpicture}[
	node distance=0mm and 0mm,
	bdd node/.style={ellipse, draw=black!70, line width=0.45pt, fill=white, inner sep=2.5pt, minimum width=17mm, minimum height=8mm, font=\small},
	bdd leaf/.style={rectangle, rounded corners=1.4pt, draw=black!70, line width=0.45pt, fill=white, inner sep=1.7pt, minimum width=11.5mm, minimum height=6.8mm, font=\small}
]
	\draw[solid, fill=blue!20, fill opacity=0.3] (0,-2.25) rectangle (6.0,-6.75);
	
			 	\node[bdd node] (x0) at (0,0) {$a<1$};
	
	\node[bdd leaf] (x10) at (-1.5,-6.0) {$1$};
	\node[bdd node] (x11) at (1.5,-1.5) {$b<1$};
	
	\node[bdd node] (x31) at (3.0,-3.0) {$a+1 \leq b$};
	
	\node[bdd node] (x41) at (1.5,-4.5) {$a \leq b+1$};
	\node[bdd node] (x42) at (4.5,-4.5) {$a \leq b+1$};

	\node[bdd leaf] (x51) at (3.0,-6.0) {$\nicefrac{1}{2}$};
	\node[bdd leaf] (x53) at (4.5,-6.0) {$0$};

	\draw[bdd edge0] (x0) -- (x10)  {};
	\draw[bdd edge1] (x0) -- (x11)  {};
	\draw[bdd edge0] (x11) -- (x10)  {};
	\draw[bdd edge1] (x11) -- (x31)  {};
	
	\draw[bdd edge1] (x31) -- (x41)  {};
	\draw[bdd edge0] (x31) -- (x42)  {};
	
	\draw[bdd edge1] (x41) -- (x10)  {};
	\draw[bdd edge0] (x41) -- (x51)  {};
	
	\draw[bdd edge1] (x42) -- (x51)  {};
	\draw[bdd edge0] (x42) -- (x53)  {};
\end{tikzpicture}
}
\caption{This TEDD $\xadd_2$ is equivalent to $\xadd_1$.}
\label{fig:overview:tedds:observe2}
\end{subfigure}
\caption{A program with conditioning and two equivalent TEDDs, representing the function mapping initial program states to the probability of all observations encountered during execution being satisfied (for $N=1$). The \textcolor{blue!40}{marked} area in $\xadd_2$ can be pruned to the area \textcolor{blue!40}{marked} in $\xadd_1$. Variables $a$ and $b$ are of type $\Nats$.}
\label{fig:overview:tedds:observe}
\end{figure}
While TEDDs have the potential to be concise, they may not be. For any fixed theory, a path in the TEDD may not be satisfiable. Using an SMT solver, we can prune these parts of the TEDD to limit their size. We illustrate this by example:
%When reasoning with TEDDs, it is essential to keep them small. The fact that inner nodes are labeled by many-sorted first order logic formulae enables a tight integration of SMT-based, theory-aware pruning techniques. Let us illustrate this by means of an example involving conditioning.
%
%
The loop $\cc$ shown in \Cref{fig:overview:tedds:observe1} tracks the number of heads and tails by a coin flip along with intermediate observations.
The $\OBSERVESYMBOL$-statement \emph{conditions} program execution on $a \leq b$ \emph{after each iteration}. With weakest pre-expectations, we can reason about \emph{conditional} expected outcomes such as the expected final value of variable $a$ \emph{given that all observations have been satisfied} \cite{DBLP:journals/toplas/OlmedoGJKKM18}. Determining conditional probabilities requires computing the probability of all observations being satisfied during program execution, which  is a weakest pre-expectation. % of type $\varvals \to \PosRealsInf$. 
The TEDD in \Cref{fig:overview:tedds:observe1} represents this pre-expectation for $N=1$: satisfying all observations has probability $1$ if $a \geq 1 \vee b \geq 1$ (no iterations) and $\nicefrac{1}{2}$ otherwise (one iteration).

The TEDD $\xadd_1$ in \Cref{fig:overview:tedds:observe1} and the TEDD $\xadd_2$ in \Cref{fig:overview:tedds:observe2} represent the same function. TEDD $\xadd_2$ is the one actually obtained when syntactically calculating the weakest pre-expectation of $\cc$: It essentially enumerates the possible ways of (un)satisfying the loop guard and the $\OBSERVESYMBOL$-statement. If, however, one takes into account that the formulae are interpreted over the \emph{theory of linear arithmetic}, one can \emph{collapse} (or \emph{prune}) the \textcolor{blue!40}{marked} area in $\xadd_2$ to the one in $\xadd_1$. This is because paths like $a<1 \wedge b<1 \wedge a+1 \leq b$ are \emph{unsatisfiable modulo linear arithmetic}. We can therefore employ an SMT-based pruning procedure to keep TEDDs small, and this can be essential: For $N=40$, the un-pruned TEDD consists of $1\,923\,612 $ nodes whereas the pruned TEDD has  $1\,723$ nodes. Our experiments in \Cref{sec:experiments:ablation} show that pruning is particularly important when using  the k-induction principle (\Cref{sec:proof_rules_via_tedds}), accelerating TEDD-based solving by orders of magnitude.

\section{Many-Sorted First Order Logic}
\label{sec:fo}
We briefly introduce a variant of many-sorted first order logic as popularized by the \smtlib standard \cite{Barrett2010TheSS}.
The formalization is tailored to our application, i.e., to deductive (probabilistic) program verification with built-in theories from SMT solvers such as \toolzt~\cite{z3} and \toolcvc~\cite{cvc5}. %Syntax and semantics are presented in \Cref{sec:fo:syntax} and \Cref{sec:fo:semantics}, respectively.

\subsection{Syntax}
\label{sec:fo:syntax}
Let $\sorts = \{\sorta_1,\sorta_2,\ldots\}$ be a countable set of \emph{type (symbols)}, containing the standard numeric types $\sortbool$, $\sortnat$, $\sortint$, $\sortrat$, $\sortreal$, $\sorteureal$, where the latter will be interpreted by the \textbf{E}xtended \textbf{U}nsigned reals, that is, the set $\PosRealsInf = \{ r \in \Reals \mid r \geq 0 \} \cup \{ \infty \}$. 
Let $\vars = \{\vara,\varb,\varc,\ldots\}$ be a countably infinite set of \emph{typed variables}. We write $\hastype{\vara}{\sorta}$ to indicate that $\vara$ is of type $\sorta$.  Let $\functs = \{\funca_1,\funca_2,\ldots\}$ be a countable set of \emph{(typed) function (symbols)}.
We write $\hastype{\funca}{\sorta_1\times\ldots \times\sorta_n \to \sorta}$ to indicate that $\funca$ is of type $\sorta_1\times\ldots\times\sorta_n \to \sorta$, and we call $n$ the \emph{arity} of $\funca$.  If $n=0$, then $\funca$ is a \emph{constant (symbol)}. We assume that $\functs$ contains symbols for the following standard symbols:
%
%\begin{enumerate}
	%\item 
	(1)~the constants $\hastype{\true,\false}{\sortbool}$, all rationals in $\Rats$, and $\infty$,
	%
	%\item 
	(2)~symbols for the logical  connectives, i.e.,  $\hastype{\wedge,\vee}{\sortbool \times \sortbool \to \sortbool}$ and $\hastype{\neg}{\sortbool\to\sortbool}$,
	%
	%\item all rationals in $\Rats$ as well as $\infty$,
	%
	%\item
	 (3)~symbols for the arithmetic connectives\footnote{We do not distinguish between operations for the different numeric types for the sake of simplicity.} $+,\cdot,-$, and the $\monus$ connective for the subtraction truncated at $0$,
	%
	%\item 
	(4)~symbols for the arithmetic relations $<,\leq,\geq,>$.
	%
%\end{enumerate}
%
We use standard infix- and prefix notations and use parentheses to resolve ambiguities. 

%Let $\rels = \{\rela_1,\rela_2,\ldots\}$ be a countable set of \emph{(typed) predicate (symbols)}, containing the standard predicates $<,\leq,\geq,>$ for the numeric types. We write $\hastype{\rela}{\sorta_1\times\ldots \times\sorta_n}$ to indicate that $\real$ has type $\sorta_1\times\ldots \times\sorta_n$. 
%\kb{we might not need this anymore}
%
\begin{definition}
	\label{def:fo}
	 The set $\term$ of \emph{terms}  is given by the grammar
	\begin{align*}
		\terma, \terma_1,\ldots,\terma_n \quad\longrightarrow\quad \vara ~\mid~ \funca ~\mid~ \terma_1 \foeq \terma_2  ~\mid~ \funcb(\terma_1,\ldots,\terma_n)~\mid~\tforall{\varb}{\sorta} \terma ~\mid~\texists{\varb}{\sorta} \terma~,
	\end{align*}
	where $\vara\in\vars$, $\funca\in\functs$ is a constant symbol, and  $\funcb\in\functs$ is a function symbol (with $n\geq 1$). \hfill $\triangle$
\end{definition}
We assume that $\term$ contains only well-typed terms and write $\hastype{\terma}{\sorta}$ to indicate that term $\terma$ has type $\sorta$, which is defined in the standard way. In particular, the term $\terma$ appearing in \Cref{def:fo} is of type $\sortbool$ so that both $\tforall{\varb}{\sorta} \terma$ and $\texists{\varb}{\sorta} \terma$ are also of type $\sortbool$. Similarly,\footnote{$\terma_1 \foeq \terma_2$ is to be understood as an $\fo$ formula whereas $\terma_1 = \terma_2$ denotes \emph{syntactic} equality of $\terma_1$ and $\terma_2$.} terms $ \terma_1 \foeq \terma_2$ denoting equalities are of type $\sortbool$ and well-typed iff both $\terma_1$ and $\terma_2$ are of the same type. Terms of type $\sortbool$ are called \emph{formulae} and the countable set of formulae is denoted by $\fo$.
Treating formulae as terms will allow for a uniform treatment in Section \ref{sec:caseexpr}ff.
For all other types $\sorta$, we denote the countable set of all $\sorta$-typed terms by $\term_\sorta$. Similarly to how SMT solvers treat numeric types and to simplify the presentation, we assume implicit sub-typing for the numeric types. Thus, every term of type $\sortnat$ is also of type $\sortreal$ and $\sorteureal$.

%In particular, term $\terma$ from above is of type $\sortbool$ so that both $\tforall{\varb}{\sorta} \terma$ and $\texists{\varb}{\sorta} \terma$ are also of type $\sortbool$. Similarly,\footnote{$\terma_1 \foeq \terma_2$ is to be understood as an $\fo$ formula whereas $\terma_1 = \terma_2$ denotes \emph{syntactic} equality of $\terma_1$ and $\terma_2$.} terms $ \terma_1 \foeq \terma_2$ denoting equalities are of type $\sortbool$ and well-typed iff both $\terma_1$ and $\terma_2$ are of the same type. We denote the countable set of all $\sorta$-typed terms by $\term_\sorta$. Terms $\terma$ of type $\sortbool$ are called \emph{formulae} and are denoted by $\forma,\formb$, and variations thereof. Bound and free variables are defined as usual. If $\forma$ contains no free variable, then we call $\forma$ a \emph{sentence}.  We usually denote set countable set of formulae by $\fo$ instead of $\term_\sortbool$. 
%

%\emph{Formulae} $\forma$ in the set $\fo$ adhere to the grammar
%%
%\begin{align*}
%	\forma \quad\rightarrow\quad  \true~\mid~\false ~\mid~ \terma = \terma ~\mid~ \rela(\terma_1,\ldots,\terma_n)~\mid~\forma\wedge\forma~\mid~\forma\vee\forma
%	~\mid~\neg\forma~\mid~\tforall{\varb}{\sorta} \forma~\mid~\texists{\varb}{\sorta} \forma~.
%\end{align*}
%%
%We consider only well-typed formulae. \hfill $\triangle$

%Bound and free variabls are defined as usual. If $\forma \in \fo$ contains no free variable, then we call $\forma$ a \emph{sentence}.

\subsection{Semantics}
\label{sec:fo:semantics}
 We use \emph{structures} to assign meanings to type- and function symbols: To each $\sorta \in \sorts$, $\struct$ assigns a set $\sem{\sorta}{\struct}$ of values. 
 Moreover, $\struct$ assigns to each function symbol $\hastype{\funca}{\sorta_1\times\ldots\times\sorta_n \to \sorta}$ a corresponding function  $\sem{\funca}{\struct}\colon \sem{\sorta_1}{\struct}\times \ldots\times\sem{\sorta_n}{\struct}\to\sem{\sorta}{\struct}$.
  In this paper, structures  assign the standard meaning to standard symbols. 
  For the extended reals,  we define $\infty \cdot 0 = 0 \cdot \infty = 0$ and $\infty \cdot \reala = \reala\cdot \infty = \infty$ as well as $\reala \leq \infty$ for all $\reala\in\PosRealsInf$, as is standard  in measure theory \cite{Tao2011MeasureTheory} and  wp-calculi~\cite{kaminski_diss}.

A \emph{(variable) valuation} $\vala$ (for $\struct$) assigns to each variable $\hastype{\vara}{\sorta}$ a value $\vala(\vara)\in\sem{\sorta}{\struct}$. The set of all valuations (for $\struct$) is $\varvals$, where $\struct$ will always be clear from the context. Pairs $(\struct,\vala)$ are called \emph{interpretations} and we denote the class of all interpretations by $\interprets$.
The semantics $\sem{\terma}{(\struct,\vala)}\in\sem{\sorta}{\struct}$ of terms $\hastype{\terma}{\sorta}$ is defined in the standard Tarski-style \cite{DBLP:books/daglib/0068003}. For formulae $\forma$, we write $(\struct,\vala) \models \forma$ instead of $\sem{\forma}{(\struct,\vala)} = \true$. If $\forma$ is a sentence (i.e., $\forma$ does not contain free variables), we write $\struct \models \forma$. 
%\sj{we move to table 1}Given a valuation $\vala$, a variable $\hastype{\vara}{\sorta}$, and a value $\valuea \in \sem{\sorta}{\struct}$, 
%%
%\[
%\vala\valsubst{\vara}{\valuea} \eeq \mylambda{\varb}
%\begin{cases}
%	\valuea & \text{if}~ \vara = \varb \\
%	%
%	\vala(\varb) & \text{otherwise}
%\end{cases}
%\]
%%
%denotes the valuation obtained from updating the value of $\vara$ in $\vala$ by $\valuea$.

Given a set of sentences $\theory$, we write $\struct \models \theory$ to indicate that $\struct \models \forma$ for all $\forma \in \theory$, in which case we call $\struct$ a \emph{$\theory$-structure}.
The set $\theory$ is a \emph{theory}, if $\theory$ is (i) satisfiable and (ii) $\theory$ is closed under logical entailment, i.e., if $\theory \models \formb$, then $\formb\in\theory$ for all sentences $\formb\in\fo$.  Moreover, we say that $(\struct,\vala)$ is a $\theory$-interpretation, if $\struct$ is a $\theory$-structure (which is independent of $\vala$). We say that $\forma\in\fo$ is \emph{satisfiable modulo $\theory$}, denoted $\sats{\forma}$, if there is a $\theory$-interpretation $(\struct,\vala)$ with  $(\struct,\vala) \models \forma$. 

\paragraph{Standard Theories}
We exemplify our concepts by means of two concrete theories:
We refer to $\theoryla$ as the theory of (mixed real- and integer) linear arithmetic, which, via our built-in type-and functions, corresponds to reasoning with $\fo$ formulae involving only linear inequalities and Boolean connectives. Satisfiability modulo $\theoryla$ for quantifier-free formulae is decidable~\cite{DBLP:conf/fmcad/0001BT14}. Moreover, we refer to $\theorylaarr$ as the extension of $\theoryla$ by $\Nats$-indexed arrays with values in $\Nats$. This theory can be obtained by introducing the type symbol $\sortarray$, the function symbols $\hastype{\funcstore}{\sortarray \times \sortnat \times \sortnat \to \sortarray}$ and $\hastype{\funcselect}{\sortarray \times \sortnat\to \sortnat}$ with the axioms
\begin{align*}
&\forall \hastype{\arraya}{\sortarray}\colon\forall \hastype{i,j}{\sortnat}\colon \funcselect(\funcstore(\arraya,j,i), j) \foeq i \\
{}\wedge{}~&\forall \hastype{\arraya}{\sortarray}\colon\forall \hastype{i,j,k}{\sortnat}\colon i\foeq j \vee
\funcselect(\funcstore(\arraya,i,k) ,j) \foeq \funcselect(\arraya,j)~.
\end{align*}
This theory is the standard theory of arrays over $\Nats$ (without extensionality) \cite{Barrett2010TheSS,BarFT-SMTLIB}. We emphasize that our results are compatible with arbitrary first-order theories, including standard theories such as fixed-sized bitvectors and user-defined theories for, e.g., exponential functions~\cite[Section 5]{DBLP:journals/pacmpl/SchroerBKKM23}. % The computational efficiency of our approach  performance of our approach in the presence of such theories does, of course, highly dependent on the capabilities of SMT solvers for dealing with them. For $\theoryla$ and $\theorylaarr$, we evaluate the performance in \Cref{sec:implementation}.
%\kbinline{decidability? }
%We remark that our approach can handle a variety of other theory extensions, including (fixed-size) bitvectors \cite{} or user-defined domains \cite{} for reasoning about, e.g., exponential functions.
%
%
%\kbinline{maybe give example on, e.g, array theory induced by select and store axioms}
%\kbinline{maybe give standard theories considered in this paper. Still have to think about $\infty$.}
%\kbinline{We have built-in LIRA, NRA, give array as example of additional theory?}

%Similarly, we say that $\forma$ \emph{entails} $\formb$ \emph{modulo} $\theory$, denoted $\forma \models_\theory \formb$, if, whenever $(\struct,\vala) \models \forma$, also $(\struct,\vala) \models \forma$

\section{Non-Deterministic Discrete Probabilistic Programs}
\label{sec:pgcl}
We consider an imperative probabilistic programming language à la McIver \& Morgan \cite{mciver_morgan} featuring (i) discrete, (ii) (pure) non-deterministic choices, (iii) modeling costs such as run-time or memory consumption via a designated statement as in \cite{aert}, and (iv) conditioning \cite{DBLP:journals/toplas/OlmedoGJKKM18}. 
\begin{definition}[The Probabilistic Guarded Command Language]
	\label{def:prelim:wp:pgcl}
	The set $\pgcl$ of programs is
%	in the \highlight{\emph{probabilistic guarded command language}} is given \mbox{by the grammar}
	%
	\begin{align*}
		\cc ~ \longrightarrow ~& \phantom{{}\mid{}~} \SKIP  
		%\tag{effectless program}\\
		%
		%&
		{}\mid{}~ \ASSIGN{\vara}{\terma}  
		%\tag{assignment}\\
		%
		%&
		{}\mid{}~ \TICK{\tickrew} 
		 %\tag{incur cost}\\
		%
		%&
		{}\mid{}~\OBSERVE{\forma} \\
		 %\tag{conditioning} \\
		%
		%&
		&{}\mid{}~ \COMPOSE{\cc}{\cc}
		 % \tag{sequential composition}\\
		%
		%&
		{}\mid{}~ \PCHOICE{\cc}{\proba}{\cc} 
		%\tag{probabilistic choice}\\
		%
		%&
		{}\mid{}~ \NDCHOICE{\cc}{\cc}
		 %\tag{non-deterministic choice}\\
		%
		%&
		{}\mid{}~ \ITE{\forma}{\cc}{\cc}
		 %\tag{conditional choice}\\
		%
		%&
		{}\mid{}~
		 \WHILEDO{\forma}{\cc}%~, %\tag{while loop}
	\end{align*}
	where all terms are well-typed and where:
	%
	%\begin{enumerate}
		%\item
		 (1)~$\vara$ is variable of some type $\sorta$ and $\hastype{\terma}{\sorta}$ is a term of type $\sorta$,
		%
		%\item
		 (2)~$\tickrew\in\PosRats$ is a non-negative rational \emph{cost} (often a constant of type $\sortrat$), 
		(3)~$\highlight{\proba} \in [0,1] \cap \Rats$ is a \emph{\highlight{rational probability}}, and
%		%
		(4)~$\highlight{\forma} \in \fo$ is a quantifier-free formula also referred to as a \emph{\highlight{guard}}. \hfill $\triangle$ 
	%\end{enumerate}
	%
	%\hfill $\triangle$
\end{definition}
%
  % by when accessing array elements in formulae.\sj{often? Why not always?}
 % $\hastype{\vara}{\sortarray}$ and $\hastype{\terma,\terma'}{\sortnat}$, 
Let us now briefly go over each statement. $\SKIP$ does nothing. The assignment $\ASSIGN{\vara}{\terma}$ updates the value of $\vara$ by the value of $\terma$ under the current variable valuation. $\TICK{\tickrew}$ has no effect on the variable valuation but models the incurrence of $\tickrew$ costs (e.g., $\tickrew$ units of run-time). $\OBSERVE{\forma}$ conditions program execution on $\forma$ being true. $\COMPOSE{\cc_1}{\cc_2}$ first executes $\cc_1$ and then $\cc_2$. The probabilistic choice $\PCHOICE{\cc_1}{\proba}{\cc_2}$ executes $\cc_1$ with probability $\proba$, and $\cc_2$ with probability $1-\proba$. $\NDCHOICE{\cc_1}{\cc_2}$ models a \emph{purely non-deterministic} choice between $\cc_1$ and $\cc_2$. Such a choice may model an \emph{adversary} aiming at, e.g., minimizing certain success probabilities, or as \emph{controllable} for, e.g., minimizing the overall incurred costs. The conditional $\ITE{\forma}{\cc_1}{\cc_2}$ executes either $\cc_1$ or $\cc_2$, depending on whether $\forma$ holds in the current variable valuation or not. Finally, the loop $\WHILEDO{\forma}{\cc}$ keeps executing $\cc$ as long as the guard $\forma$ holds.

Furthermore, for arrays,  we write $\ARRAYSTORE{\vara}{\terma}{\terma'}$ instead of $\ASSIGN{\vara}{\funcstore(\vara, \terma,\terma')}$ and  $\ARRAYREAD{\vara}{\terma}$ instead of $\funcselect(\vara,\terma)$ to access array elements.
If $\cc$ contains no loops, then we call $\cc$ \emph{loop-free}.
 If $\cc$ contains no non-deterministic choice, then we call $\cc$ \emph{fully probabilistic}.

\section{The Weakest Pre-Expectation Calculus}
\label{sec:wp}
In this section, we introduce the \emph{weakest pre-expectation ($\wpsymbolnostruct$) calculus}.
%Reasoning about costs such as run-time was introduced by \cite{ert_journal}.
%
%To align the formalization of the $\wpsymbolnostruct$-calculus with our many-sorted first order logic setting from \Cref{sec:fo}, 
We parameterize the calculus by a structure $\struct$ (see \Cref{sec:fo}), which we fix throughout this section.  \emph{Program states} for $\pgcl$ programs then correspond to variable valuations $\vala$ in $\varvals$. 

\subsection{Expectations}

The central objects the $\wpsymbolnostruct$-calculus operates on are so-called \emph{expectations} --- random variables on a program's state space. 
The set of expectations (for $\struct$) is 
\(
	\Es = \{ \FF ~\mid~ \FF \colon \varvals \to \PosRealsInf \}
\).
Expectations are denoted by $\FF,\FG,\FH$ and variations thereof. Expectations are quantitative generalizations of predicates from classic program verification: While classical predicates map program states to Booleans, expectations map program states to quantities in $\PosRealsInf$. The tuple $(\Es, \, \eleq)$ is a complete lattice with the partial order $\eleq$, defined as 
\(
	\FF \eleq \FG ~\text{iff}~ \forall \vala \in \varvals \colon \FF(\vala) \leq \FG(\vala)
\),
and state-wise lifted suprema and infima over $\PosRealsInf$.
%, i.e., 
%%
%\[
%	\bigesup \{ \FF_1,\FF_2,\ldots \} \eeq \mylambda{\vala} \sup \{ \FF_1(\vala), \FF_2(\vala), \ldots\}
%	\quad \text{and}\quad
%	\bigeinf \{ \FF_1,\FF_2,\ldots \} \eeq \mylambda{\vala} \inf \{ \FF_1(\vala), \FF_2(\vala), \ldots\}~.
%\]
%
Finally, we usually denote pairwise minima $\bigeinf \{ \FF, \FG \}$ by $\FF \einf \FG$.
%We usually write $\FF \einf \FG$ instead of $\bigeinf \{ \FF, \FG \}$ to denote minima.

\subsection{The Calculus}
\begin{table}[t]
	\begin{center}
		\begin{tabularx}{\textwidth}{X@{\quad}l@{\quad~~}lX}
			\toprule
			\toprule
			$\boldsymbol{\cc}$ & $\mathsf{\mathbf{wp}^{\mathsf{\struct}}}\boldsymbol{\llbracket\cc \rrbracket(\FF)}$  \\[0.5ex]
			\hline 
			%\midrule
			$\SKIP$ & $\FF$  \rule{0pt}{3.5ex}\\[1.8ex]
			%		%
			$\ASSIGN{\vara}{\terma}$ & $\mylambda{\vala} \FF\big( \vala\valsubst{\vara}{\sem{\terma}{(\struct, \vala)}}\big)$   \\[1.5ex]
			%		%
			$\TICK{\tickrew}$ & $\mylambda{\vala} \tickrew + \FF(\vala)$   \\[1.5ex]
			$\OBSERVE{\forma}$ & $\mylambda{\vala} 
			\begin{cases}
				\FF(\vala) & \text{if}~ (\struct, \vala) \models \forma \\
				0& \text{if}~ (\struct, \vala) \models \neg\forma
			\end{cases}$ \\[3.5ex]
			%		%
			$\COMPOSE{\cc_1}{\cc_2}$ & $\wp{\cc_1}{\wp{\cc_2}{\FF}}$  \\[1.5ex]
			%		%
			$\NDCHOICE{\cc_1}{\cc_2}$ & $\wp{\cc_1}{\FF} \einf \wp{\cc_2}{\FF}$ \\[1.5ex]
			$\PCHOICE{\cc_1}{\proba}{\cc_2}$ & $\mylambda{\vala} \proba \cdot \wp{\cc_1}{\FF}(\vala) + (1-\proba )\cdot \wp{\cc_2}{\FF}(\vala)$ \\[1.5ex]
			$\ITE{\forma}{\cc_1}{\cc_2}$ & 
				$\mylambda{\vala} 
				\begin{cases}
						\wp{\cc_1}{\FF}(\vala) & \text{if}~ (\struct, \vala) \models \forma \\
						\wp{\cc_2}{\FF}(\vala) & \text{if}~ (\struct, \vala) \models \neg\forma
				\end{cases}$    \\[3.5ex]
			%		%$\iverson{\forma}\cdot \wp{\cc_1}{\FF} + \iverson{\neg\forma}\cdot \wp{\cc_2}{\FF} $
			%
			$\WHILEDO{\forma}{\cc'}$ & $\lfp \FG.\,\mylambda{\vala}
			\begin{cases}
				   \wp{\cc'}{\FG}(\vala) &  \text{if}~ (\struct, \vala) \models \forma \\ 
				   \FF(\vala) & \text{if}~ (\struct, \vala) \models \neg\forma
			\end{cases}$
			%$\lfp \FG.\, \iverson{\forma}\cdot \wp{\cc'}{\FG} + \iverson{\neg\forma}\cdot\FF$ 
			 \\
			\bottomrule
			\bottomrule
		\end{tabularx}
	\end{center}%
	\caption{Recursive definition of the weakest pre-expectation of program $\cc$ w.r.t.\ post-expectation $\FF$. Given $\hastype{\vara}{\sorta}$ and $\valuea\in \sem{\sorta}{\struct}$, we denote by $\vala\valsubst{\vara}{\valuea}$ the state obtained from updating the value of $\vara$ in  $\vala$ by $\valuea$.}
	\label{tab:wp}%
\end{table}%
We start with the central definition of weakest pre-expectations, lifting weakest pre-conditions.
\begin{definition}
	Let $\cc \in \pgcl$ and $\FF \in \Es$. 
	%
%	\begin{enumerate}
%	\item 
	The \emph{weakest pre-expectation} $\wp{\cc}{\FF} \in \Es$ of $\cc$ w.r.t.\ (post-expectation) $\FF$ is defined recursively on the structure of $\cc$ as shown in \Cref{tab:wp}. \hfill $\triangle$
	%
%	\item The \emph{weakest liberal pre-expectation} $\wlp{\cc}{\FF}$ of $\cc$ w.r.t.\ $\FF$ is defined for $\KWTICK$-free $\cc$ and $1$-bounded $\FF$ and obtained from \Cref{tab:wp} by replacing  $\wpsymbol$ by $\wlpsymbol$ and $\lfpnospace$ by $\gfpnospace$. \hfill $\triangle$
%	\end{enumerate}
\end{definition}

Let us gain some intuition on the quantities weakest pre-expectations determine. We focus on programs $\cc$ \emph{not} containing $\OBSERVESYMBOL$. Conditioning is treated at the end of this section. Given a post-expectation $\FF$, the weakest pre-expectation $\wp{\cc}{\FF}$ is yet another expectation --- a function mapping program states to numbers. The intuition is:
\[
\wp{\cc}{\FF}(\vala) \eeq
\substack{
	\text{\normalsize \emph{minimal} --- under all possible resolutions of the non-determinism ---} \\
	\text{\normalsize \emph{expected} cost incurred when executing $\cc$ on $\vala$} \\
	 \text{\normalsize and collecting a final cost of $\FF(\vala')$ upon termination in state $\vala'$}~.
	}
\]
Hence, in this paper, $\wpsymbol$ treats the pure non-determinism in a \emph{minimizing} manner. If $\cc$ is fully probabilistic, we can drop the topmost line from the above explanation since there is no non-determinism to be resolved.
%If $\cc$ is fully probabilistic, then $\wp{\cc}{\FF}(\vala)$ is the expected value of $\FF$ w.r.t.\ the distribution of final states reached upon executing $\cc$ on $\vala$. If $\cc$ contains non-deterministic choices, this final distribution is not necessarily unique and thus $\wp{\cc}{\FF}(\vala)$ minimizes the expected final value of $\FF$ under \emph{all} final distributions induced by the non-determinism. This can be made precise by establishing tight connections to Markov decision processes \cite{}.
The post-expectation $\FF$ can be thought of as a \emph{continuation} that is taken into account upon termination of $\cc$. More concretely, the weakest pre-expectation calculus enables us to reason, amongst others, about the following \emph{expected outcomes} of $\pgcl$ programs:
\begin{enumerate}
	\item If $\cc$ is $\KWTICK$-free and $\iverson{\forma}$ is the indicator function
	%\footnote{i.e., $\iverson{\forma}(\vala)\in\{0,1\}$ evaluating to $1$ iff $(\struct,\vala)\models \forma$.} 
	of $\forma \in \fo$, then $\wp{\cc}{\iverson{\forma}}(\vala)$ is the (minimal) probability of terminating in a state satisfying $\forma$ when executing $\cc$ on $\vala$.
	\item If $\cc$ is $\KWTICK$-free and $\vara$ is a variable of type $\Nats$, then $\wp{\cc}{\mylambda{\vala} \vala(\vara)}(\vala)$ is the (minimal) expected final value of $\vara$ when executing $\cc$ on the initial state $\vala$.
	\item For $\cc$ possibly containing $\KWTICK$, $\wp{\cc}{0}(\vala)$ is the (minimal) expected cost incurred  (as determined by the $\KWTICK$ statements in $\cc$) when executing $\cc$ on the initial state $\vala$.
\end{enumerate}
\begin{example}
	\label{ex:wp:array}
	Let $\struct$  be a structure with $\sem{\sortarray}{\struct} = \Nats \to \Nats$, assigning to $\funcstore$ and $\funcselect$ the obvious functions for manipulating and accessing these functions. Now consider the program $\cc$
	\begin{align*}
		\ITE{
			\ARRAYREAD{A}{\vara} \leq \varb
		}{\quad 
			\PCHOICE{\ARRAYSTORE{A}{\vara}{\varb +2}}{\nicefrac{1}{2}}{\ARRAYSTORE{A}{\vara}{\varb +4}}\quad
		}{
			\SKIP
		}~,
	\end{align*}
	where $\hastype{A}{\sortarray}$ and $\hastype{\vara,\varb}{\sortnat}$, and fix the post-expectation $\FF = \mylambda{\vala} \sem{A}{(\struct, \vala)} (\vala(\vara))$, evaluating to the value of $A$ at index $\vala(\vara)$ under every state $\vala$. Using the rules in \Cref{tab:wp}, we calculate:
	\begin{align*}
		\wp{\cc}{\FF} \eeq
		\mylambda{\vala} 
		\begin{cases}
			\vala(\varb) + 3&\text{if}~(\struct,\vala) \models  \ARRAYREAD{A}{\vara} \leq \varb \\
			\sem{A}{(\struct, \vala)} (\vala(\vara))&\text{if}~ (\struct,\vala) \not\models  \ARRAYREAD{A}{\vara} \leq \varb~,
		\end{cases}
	\end{align*}
	i.e., if initially $\ARRAYREAD{A}{\vara} \leq \varb$ then the expected final value of $\ARRAYREAD{A}{\vara}$ is $\varb$'s initial value plus $3$. Otherwise, the expected final value of $\ARRAYREAD{A}{\vara}$ is its initial value. \hfill $\triangle$
\end{example}

Let us now go over the rules from \Cref{tab:wp}. Since $\SKIP$ does nothing, $\wp{\SKIP}{\FF}$ is just $\FF$. For assignments, $\wp{\ASSIGN{\vara}{\terma}}{\FF}$, intuitively, simply replaces $\vara$ by $\terma$ in $\FF$. Since, for now, we treat expectations as semantic rather than syntactic objects, this substitution process is formalized semantically. $\wp{\TICK{\tickrew}}{\FF}$ adds the cost $\tickrew$ state-wise to the post-expectation. The rule for sequential composition suggests that weakest pre-expectations are --- analogously to Dijkstra's weakest pre-\emph{conditions} \cite{dijkstra_discipline} --- determined in a \emph{backward-moving} fashion:  We first determine $\wp{\cc_2}{\FF}$ as an intermediate expectation, which is then plugged into $\wp{\cc_1}{\cdot}$, yielding the weakest pre-expectation of $\COMPOSE{\cc_1}{\cc_2}$ w.r.t.\ $\FF$. The weakest pre-expectation of a non-deterministic choice $\NDCHOICE{\cc_1}{\cc_2}$ is the state-wise minimum of the weakest pre-expectations of the two branches, which reflects that non-determinism is treated in a minimizing fashion. For conciseness, we omit the straightforward adaptions for resolving the non-determinism in a maximizing manner. The weakest pre-expectation of a probabilistic choice $\PCHOICE{\cc_1}{\proba}{\cc_2}$ is the $\proba$-weighted convex sum of the pre-expectations of the respective branches. The $\wpsymbolnostruct$ of a conditional evaluates to either of the two branches, depending on whether the current state satisfies the guard or not. 

It is evident that, for the loop-free constructs, determining $\wpsymbolnostruct$ boils down to recursing on the program structure and is hence a rather syntactic process. Reasoning about $\wpsymbolnostruct$'s of loops is more involved as it requires a least fixed-point construction, which we detail next.
%
%

%\paragraph{Loops}
%Let us now detail $\wpsymbolnostruct$'s of loops.
%We now detail weakest pre-expectations of loops. 
Given a loop $\cc = \WHILEDO{\forma}{\cc'}$ and a post-expectation $\FF$, we call\footnote{We omit $\struct$ from $\wcharfun{}$ to avoid clutter as it will always be clear from the context.}
\[
	\wcharfun{\FF} \colon \Es \to \Es, \qquad \wcharfun{\FF}(\FG) \eeq \mylambda{\vala}
	\begin{cases}
		\wp{\cc'}{\FG}(\vala) &  \text{if}~ (\struct, \vala) \models \forma \\ 
		\FF(\vala) & \text{if}~ (\struct, \vala) \models \neg\forma
	\end{cases}
\]
the \emph{characteristic function of $\cc$ w.r.t.\ $\FF$}. By Kleene's fixed point theorem, we have
%Since $\wpsymbol$ satisfies the necessary continuity conditions \cite{kaminski_diss}, Kleene's fixed-point theorem applies so that we have
%
\begin{align*}
	\wp{\WHILEDO{\forma}{\cc'}}{\FF} \eeq \lfp \wcharfun{\FF} \eeq \bigesup_{n\in\Nats} \wcharfuniter{\FF}{n}(\expzero)~,
	\tag{see \cite{kaminski_diss}}
\end{align*}
where $\expzero$ denotes the constantly-$0$-expectation and $\wcharfuniter{\FF}{n}(\expzero)$ is the $n$-fold application of $\wcharfun{\FF}$ to $0$. Intuitively, $\wcharfuniter{\FF}{n}(\expzero)$ is the weakest pre-expectation of $\cc$ w.r.t.\ $\FF$ when aborting $\cc$ after at most $n$ iterations. Thus, iterating $\wcharfun{\FF}$ --- so to speak --- ad infinitum on $\FF$ on $\expzero$ yields the precise $\wpsymbolnostruct$. %This fixed-point iteration is one of the principles we will implement via TEDDs.

\begin{remark}[Conditioning]
	\label{rem:conditioning}
	This presentation of $\wpsymbolnostruct$-calculus  treats conditioning as in~\cite{DBLP:journals/toplas/OlmedoGJKKM18}. That is, for fully probabilistic programs $\cc$ not containing $\KWTICK$-statements but possibly containing $\OBSERVESYMBOL$, $\wp{\cc}{\FF}$ maps each initial program state to the expected final of $\FF$ when \emph{excluding} all runs not satisfying all encountered observations. The \emph{conditional} expected final value of $\FF$ \emph{given} that all observations are satisfied is then given by the function\footnote{This characterization assumes $\cc$ to be almost-surely terminating, which can be avoided by taking weakest \emph{liberal} pre-expectations in the denominator \cite{DBLP:journals/toplas/OlmedoGJKKM18}, as adopted by our implementation.}
	\begin{align*}
		\mylambda{\vala}
		\begin{cases}
			\frac{\wp{\cc}{\FF}(\vala)}{\wp{\cc}{\mylambda{\vala'} 1}(\vala)} & \text{if}~ \wp{\cc}{\mylambda{\vala'} 1}(\vala) > 0, \text{ and}\\
			\text{undefined} & \text{otherwise}~.
		\end{cases}
		\tag*{$\triangle$}
	\end{align*}
\end{remark}

\section{Case Expressions}
\label{sec:caseexpr}
We now introduce (typed) \emph{case expressions} --- the semantic basis of TEDDs. \Cref{sec:caseexpr:syntax_semantics} introduces their syntax and semantics of case expressions, and \Cref{sec:caseexprs:ops} treats important operations.
\subsection{Syntax and Semantics}
\label{sec:caseexpr:syntax_semantics}
Towards defining the syntax of case expressions, we introduce the following notions.
Let $\at \subseteq \fo$ be the set of \emph{Boolean atoms}, i.e., formulae whose outermost constructor is none of the Boolean connectives $\wedge,\vee,\neg$. Notice that formulae starting with a quantifier such as $\exists \hastype{\vara}{\sortnat}\colon \vara \neq 0$ are considered Boolean atoms here. A \emph{propositional interpretation} is a valuation $\propinta \colon \at \to \{0,1\}$ with $\propinta(\true) =1$ and $\propinta(\false)=0$. We extend $\propinta$ homomorphically to Boolean combinations of atoms by
\begin{align*}
\propsem{\formb_1 \wedge \formb_2}{\propinta}
\eeq &\;
\min\{\propsem{\formb_1}{\propinta},\propsem{\formb_2}{\propinta}\}
\\[0.3ex]
\propsem{\formb_1 \vee \formb_2}{\propinta}
\eeq &\;
\max\{\propsem{\formb_1}{\propinta},\propsem{\formb_2}{\propinta}\}
\\[0.3ex]
\propsem{\neg\formb}{\propinta}
\eeq &\;
1-\propsem{\formb}{\propinta}~.
\end{align*}
We call $\propsem{\forma}{\propinta}$ the \emph{propositional truth value of $\forma$ w.r.t.\ $\propinta$}.
Formulae of the form $\forma = \bigwedge_i \formb_i$ with $\formb_i \in \fo$ are called \emph{cubes}. We call $n\geq 1$ cubes\sj{why restrict to cubes here?}\kb{we could generalize but we don't need to}\sj{I think i would define this for a formula.} $\forma_1,\ldots,\forma_n$ \emph{propositionally partitioning}, if for all propositional interpretations $\propinta$ there is exactly one $i\in\{1,\ldots,n\}$ with $\propsem{\forma_i}{\propinta}  = 1$.
%\begin{definition}
%	Let $\sorta\in\sorts$ be a type, let $\hastype{\terma_1,\ldots,\terma_n}{\sorta}$ be terms, and let $\forma_1,\ldots,\forma_n\in\fo$ be partitioning formulae. We call functions $\switchfun \colon \interprets \to \term_\sorta$ of the form
%	%
%	\[
%		\switchfun(\inta) \eeq \begin{cases}
	%			\terma_1 &\text{if}~\inta\models\forma_1 \\
	%			%
	%			\vdots \\
	%			%
	%			\terma_n &\text{if}~\inta\models\forma_n
	%		\end{cases}
%	\]
%	\emph{$\sorta$-switching functions}. Moreover, we define the \emph{value function}\footnote{This function is dependently typed since its codomain depends on the structure $\struct$ determining the semantics of the type $\sorta$.}  $\valfun{\switchfun}\colon ((\struct, \vala) \colon \interprets) \to \sem{\sorta}{\struct}$ as
%	%
%	\[
%		\valfun{\switchfun}(\inta) \eeq \sem{\switchfun(\inta)}{\inta}~.
%	\]
%	%
%	and call $\valfun{\switchfun}(\inta)$ the \emph{value of $\switchfun$ at $\inta$}.
%	\hfill$\triangle$
%\end{definition}
%
\begin{definition}
	Let $\sorta\in\sorts$. A \emph{$\sorta$-case expression $\switchfun$} is a syntactic construct of the form \sj{spacing below is a bit weird?}\kb{i don't see that}
	\[
	\switchfun \eeq \begin{cases}
		\terma_1 &\switchfuncase \forma_1 \\
		\vdots  &  \\
		\terma_n &\switchfuncase\forma_n~,
	\end{cases}
	\]
	where $n\geq 1$, % :
	%
	%\begin{enumerate}
	%	\item 
		$\hastype{\terma_1,\ldots,\terma_n}{\sorta}$ are terms, % of type $\sorta$, 
		and
		%
		%\item 
		$\forma_1,\ldots,\forma_n$
		%\in\fo$
		 are propositionally partitioning cubes. \hfill $\triangle$
	%\end{enumerate}
	%
\end{definition}
We denote the set of all $\sorta$-case expressions by $\switchfunset{\sorta}$. For $\sortbool$-case expressions in the set $\switchfunset{\sortbool}$, we impose the additional restriction that all left-hand side's terms are one of the Boolean constants $\true$ or $\false$, and we define the \emph{canonical $\sortbool$-case expression} of $\forma\in\fo$ as
\[
	\switchfun_\forma \eeq \begin{cases}
		\true &\switchfuncase \forma \\
		\false &\switchfuncase \neg \forma~.
	\end{cases}
\]

A $\tau$-case expression $\switchfun$ represents functions from interpretations $\inta\in\interprets$ to terms of type $\tau$. The \emph{term selected by $\switchfun$ under $\inta$} is
\begin{align*}
	\switchfunterm{\switchfun}{\inta}
	\eeq 
	\begin{cases}
		{\terma_1}&\text{if}~ \inta\models\forma_1 \\
		\vdots &  \\
		{\terma_n} & \text{if}~\inta\models\forma_n~. 
	\end{cases}
\end{align*}
The \emph{value of $\switchfun$ under $\inta$} is $\sem{\switchfunterm{\switchfun}{\inta}}{\inta}$, written $\sem{\switchfun}{\inta}$. Well-definedness  is ensured as $\forma_1,\ldots,\forma_n$ are propositionally partitioning: For every $\inta\in\interprets$, there is a unique $i$ with $\inta\models\forma_i$. The  property is generally undecidable, thus, we impose this property by its decidable \emph{propositional} variant.
\begin{example}
	\label{ex:eurealswitchfun}
	$\sorteureal$-case expressions provide us with a syntactic means to specify (classes of) expectations. Let $\hastype{\vara,\varb}{\sortnat}$ and consider the $\sorteureal$-case expression
	\[
		\switchfun \eeq 
		\begin{cases}
			\varb &\switchfuncase \vara \foeq 0 \\
			\varb+1 &\switchfuncase \vara \foneq 0 \wedge \vara\foeq 1 \\
			\varb+2 &\switchfuncase \vara \foneq0 \wedge \vara \foneq 1 \wedge \vara\foeq 2 \\
			\varb+3 &\switchfuncase \vara \foneq0 \wedge \vara \foneq 1 \wedge \vara\foneq 2 ~.
		\end{cases}
	\]
	We have $\sem{\switchfun}{(\struct,\vala)}\in\PosRealsInf$ for all interpretations $(\struct,\vala)$. Consequently, we have $\mylambda{\vala}\sem{\switchfun}{(\struct,\vala)}\in~\Es$ for every structure $\struct$. Here, $\struct$ is irrelevant since $\switchfun$ only involves standard arithmetic.
	\hfill $\triangle$
\end{example}

\subsection{Operations on Case Expressions}
\label{sec:caseexprs:ops}
Next, we introduce two operations on case expressions central to our approach: Operations induced by operators on terms and substitutions of variables by terms.
%
%
%
%\paragraph{Lifting Operators on Terms to Case Expressions}
\begin{definition}
	Let $n\geq 1$, let $\sorta_1,\ldots,\sorta_n,\sorta\in\sorts$, and let $\termop\colon\term_{\sorta_1}\times\ldots\times\term_{\sorta_n} \to \term_\sorta$ be an operator on terms. 
	We lift $\termop$ to an  \emph{operator} on case expressions:
	 \begin{align*}
&	\switchfunop{\termop} \colon \switchfunset{\sorta_1} \times \ldots \times \switchfunset{\sorta_n} \to \switchfunset{\sorta}, \\
	& 	\switchfunop{\termop} \Big(
	 	\begin{cases}
	 		\terma_{1,1} & \switchfuncase \forma_{1,1} \\
	 		\vdots & \\
	 		\terma_{1,m_1} & \switchfuncase \forma_{1,m_1}
	 	\end{cases}
	 	,\,
	 	\ldots,\,
	 	\begin{cases}
	 		\terma_{n,1} & \switchfuncase \forma_{n,1} \\
	 		\vdots & \\
	 		\terma_{n,m_n} & \switchfuncase \forma_{n,m_n}
	 	\end{cases}
	 	\Big)
	 	\quad\eeq\quad
	 	\underset{\mathclap{\text{for each}~(i_1,\ldots,i_n) \in \{1,\ldots,m_1 \}\times\ldots\times\{ 1,\ldots,m_n \} }}{
	 	\left\{
	 	\begin{array}{ll}
	 		\termop(\terma_{1,i_1},\ldots,\terma_{n,i_n}) & \switchfuncase \bigwedge_{j=1}^n \forma_{j,i_j}\\
	 	\end{array}
	 	\right.
	 }~.
	 \tag*{$\triangle$}
	 \end{align*}
	 %
	 %\hfill $\triangle$
\end{definition}
%
%\sj{Maybe introduce notation for $\{ 1, \dots, m_n \}$}
%
This operator satisfies 
$
	\sem{\switchfunop{\termop}(\switchfun_1,\ldots,\switchfun_n)}{\inta}
	 = 
	 \sem{\termop(\switchfunterm{\switchfun_1}{\inta},\ldots,\switchfunterm{\switchfun_n}{\inta})}{\inta}~,
$
i.e., applying $\switchfunop{\termop}$ to $\switchfun_1,\ldots,\switchfun_n$ is equivalent to applying $\termop$ to the terms selected by $\switchfun_1,\ldots,\switchfun_n$ for every $\inta$.
If $\termop$ simply applies some function symbol $\funca$ to its arguments, we often write $\switchfunop{\funca}$ \mbox{instead of $\switchfunop{\termop} $.}

Our formalization via term operators enables a uniform treatment of various operations on case expressions, which will also generalize to TEDDs. For instance, case distinctions can be implemented via the operator $\switchITEsymbol\colon \fo \times \term_{\sorta} \times \term_{\sorta} \to \term_\sorta$, defined for every type $\sorta$ as
\[
	\switchITE{\forma}{\terma_1}{\terma_2} \eeq
	\begin{cases}
		 \terma_1 & \text{if}~\forma = \true \\
		 \terma_2 & \text{otherwise}~,
	\end{cases}
	\quad\text{so that}\quad
	\sem{\switchfunop{\switchITEsymbol}(\switchfun_\forma, \switchfun_1,\switchfun_2)}{\inta}
	\eeq
	\begin{cases}
		\sem{\switchfun_1}{\inta} & \text{if}~\inta \models \forma \\
		\sem{\switchfun_2}{\inta} & \text{if}~\inta \not\models \forma~.
	\end{cases}
\]
%
%so that 
%%
%\[
%	\sem{\switchfunop{\switchITEsymbol}(\switchfun_\forma, \switchfun_1,\switchfun_2)}{\inta}
%	\eeq
%	\begin{cases}
%		\sem{\switchfun_1}{\inta} & \text{if}~\inta \models \forma \\
%		%
%		\sem{\switchfun_2}{\inta} & \text{if}~\inta \not\models \forma~.
%	\end{cases}
%\]

\begin{example}
	Let $\switchfun_1,\switchfun_2$ be $\sorteureal$-case expressions. Using $\switchfunopsymbol$, we obtain the switching function representing the interpretation-wise addition of $\switchfun_1$ and $\switchfun_2$ by $\switchfunop{+}(\switchfun_1,\switchfun_2)$, i.e., 
	\begin{align*}
	  \forall{\inta\in\interprets}\colon\quad
		\sem{\switchfunop{+}(\switchfun_1,\switchfun_2)}{\inta}
		\eeq
		\sem{\switchfun_1}{\inta} + \sem{\switchfun_2}{\inta} ~.
		\tag*{$\triangle$}
	\end{align*}
\end{example}
Next, we define substitution of variables by terms in case expressions:
\begin{definition}
	Let $\sorta,\sorta'$ be a types, $\switchfun$ a $\sorta$-case expression and let $\hastype{\vara}{\sorta'}, \hastype{\terma}{\sorta'}$ be a variable and term of type $\sorta'$, respectively. We define $\switchfun\switchfunsubst{\vara}{\terma}$ to be the case expression obtained from $\switchfun$ by replacing every occurrence of $\vara$ by $\terma$ in all terms and formulae in $\switchfun$ (in a capture-avoiding manner). \hfill $\triangle$
\end{definition}

Naturally, this operator satisfies 
\[
	\sem{\switchfun}{(\struct, \vala\valsubst{x}{\sem{\terma}{(\struct,\vala)}})} = \sem{\switchfun\switchfunsubst{\vara}{\terma}}{(\struct,\vala)}~,
\]
i.e., syntactically replacing $\vara$ by $\sorta$ in $\switchfun$ corresponds to semantical substitution.

\section{Typed Extended Decision Diagrams}
\label{sec:xadds}
%Let $\pvars = \{\pvara_1,\pvara_2,\ldots\}$ be a countably infinite set of (propositional) variables, let $\prec$ be a linear order on $\pvars$,  and let $\adom$ be a (possibly infinite) set.
%%
%\begin{definition}
%	An \emph{algebraic decision diagram} (over $\adom$) is a structure
%	%
%	\[
%		\add \eeq
%		(\anodes,\, \annodes, \, \atnodes, \, \asuccf, \, \asucct, \, \alab, \, \aroot)~,
%	\]
%	%
%	where: 
%	%
%	\begin{enumerate}
%		\item $V$ non-empty partitioned intro $\annodes$,$\atnodes$
%	\end{enumerate}
%\end{definition}
%
%
In this section, we introduce our typed extended decision diagrams (TEDDs). We introduce syntax and semantics in \Cref{sec:xadds:syntax_semantics} and discuss reducedness and canonicity of TEDDs in \Cref{sec:tedds:reducedness}. 
% and present algorithms in \Cref{sec:xadds:operations}.\sj{Why discuss what is in 7.1 and in 8, but not what is in 7.2, 7.3?}\kb{because that section is new and I forgot, thanks} In \Cref{sec:xadds:operations}, we treat operations on TEDDs.\sj{Why mention that here?}\kb{because this is a new section now and I forgot to adapt }

\subsection{Syntax and Semantics}
\label{sec:xadds:syntax_semantics}
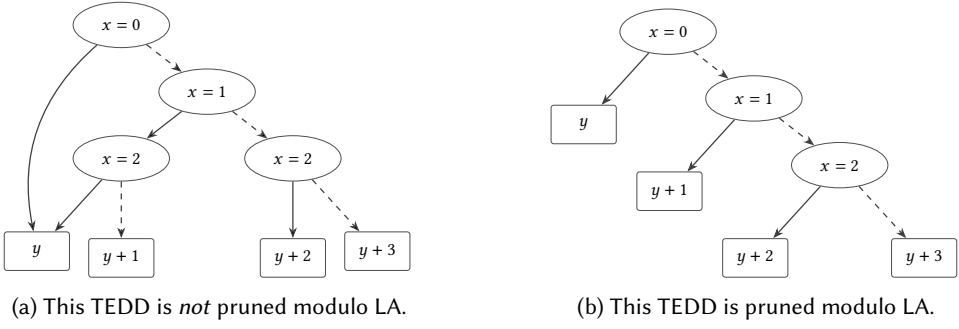
\begin{figure}[t]
%\begin{tikzpicture}[node distance=20mm and 30mm]
%	% Leaves
%	\node[bdd leaf] (L0) {$0$};
%	\node[bdd leaf, right=of L0] (L1) {$1$};
%	% A decision node x
%	\node[bdd node, above=of $(L0)!0.5!(L1)$] (x) {$x$};
%	\draw[bdd edge0] (x) -- (L0) node[midway, left] {};
%	\draw[bdd edge1] (x) -- (L1) node[midway, right] {};
%\end{tikzpicture}
\centering
\begin{subfigure}[t]{0.48\textwidth}
	\centering
	\scalebox{0.75}{%
	\begin{tikzpicture}[
		node distance=6mm and 3mm,
		decision/.style={ellipse, draw=black!70, line width=0.45pt, fill=white, inner sep=2.5pt, minimum width=17mm, minimum height=8mm, font=\small},
		leaf/.style={rectangle, rounded corners=1.4pt, draw=black!70, line width=0.45pt, fill=white, inner sep=1.7pt, minimum width=11.5mm, minimum height=6.8mm, font=\small},
		edgefalse/.style={-{Stealth[length=1.9mm]}, dashed, line width=0.65pt, draw=black!75},
		edgetrue/.style={-{Stealth[length=1.9mm]}, line width=0.65pt, draw=black!75}
	]
		 
		 	\node[decision] (x0) {$\vara=0$};
		 	\node[decision, below right=of x0] (x1) {$\vara=1$};
		 	\node[decision, below left=of x1] (x21) {$\vara=2$};
		 	\node[decision, below right=of x1] (x22) {$\vara=2$};
		 	
		 	\node[leaf, below left=10mm and 3mm of x21] (y) {$\varb$};
		 	\node[leaf, below=10mm of x21] (y1) {$\varb+1$};
		 	
		 	\node[leaf, below=10mm of x22] (y2) {$\varb+2$};
		 	\node[leaf, below right=10mm and 3mm of x22] (y3) {$\varb+3$};
		 	
		 	\draw[edgetrue] (x0) edge[bend right] (y)  {};
		 	\draw[edgefalse] (x0) edge (x1)  {};
		 	\draw[edgetrue] (x1) edge (x21)  {};
		 	\draw[edgefalse] (x1) edge (x22)  {};
		 	\draw[edgetrue] (x21) edge (y)  {};
		 	\draw[edgefalse] (x21) edge (y1)  {};
		 	\draw[edgetrue] (x22) edge (y2)  {};
		 	\draw[edgefalse] (x22) edge (y3)  {};
		
	\end{tikzpicture}
}
	\caption{This TEDD is \emph{not} pruned modulo $\theoryla$.}
	\label{fig:ex_xadd2}
\end{subfigure}\hfill
\begin{subfigure}[t]{0.48\textwidth}
	\centering
	\scalebox{0.75}{%
\begin{tikzpicture}[
	node distance=6mm and 3mm,
	decision/.style={ellipse, draw=black!70, line width=0.45pt, fill=white, inner sep=2.5pt, minimum width=17mm, minimum height=8mm, font=\small},
	leaf/.style={rectangle, rounded corners=1.4pt, draw=black!70, line width=0.45pt, fill=white, inner sep=1.7pt, minimum width=11.5mm, minimum height=6.8mm, font=\small},
	edgefalse/.style={-{Stealth[length=1.9mm]}, dashed, line width=0.65pt, draw=black!75},
	edgetrue/.style={-{Stealth[length=1.9mm]}, line width=0.65pt, draw=black!75}
]
	% Leaves
	
	\node[decision] (x0) {$\vara=0$};
	\node[leaf, below left=10mm and 3mm of x0] (y) {$\varb$};
	\draw[edgetrue] (x0) -- (y)  {};
	
	\node[decision, below right=of x0] (x1) {$\vara=1$};
	\draw[edgefalse] (x0) -- (x1)  {};
	\node[leaf, below left=10mm and 3mm of x1] (y1) {$\varb+1$};
	\draw[edgetrue] (x1) -- (y1)  {};
	
	\node[decision, below right=of x1] (x2) {$\vara=2$};
	\draw[edgefalse] (x1) -- (x2)  {};
	\node[leaf, below left=10mm and 3mm of x2] (y2) {$\varb+2$};
	\node[leaf, below right=10mm and 3mm of x2] (y3) {$\varb+3$};
	\draw[edgetrue] (x2) -- (y2)  {};
	\draw[edgefalse] (x2) -- (y3)  {};
	
	%\node[bdd leaf, below left=of x0] (y) {$\varb$};
	
\end{tikzpicture}
}
\caption{This TEDD is pruned modulo $\theoryla$.}
\label{fig:ex_xadd1}
\end{subfigure}
\caption{Two $\sorteureal$-TEDDs, which are equivalent modulo $\theoryla$. All variables are of type $\sortnat$.}
\label{fig:ex_xadd}
\end{figure}
Fix a linear order
%\footnote{Our implementation generates this order on-the-fly by ordering formulae chronologically by the time they are encountered.}\sj{This footnote feels out of place, we are talking syntax en semantics} 
$\varord$ on the countably infinite set of atomic formulae in $\at$ (cf.\ \Cref{sec:caseexpr:syntax_semantics}).
\begin{definition}\kb{labels now map tp $\at$!! Check consistent}
	Let $\sorta\in\sorts$ be a type. A \emph{$\sorta$-typed extended decision diagram ($\sorta$-TEDD)} is a tuple
	\[
		\xadd \eeq (\xnodes, \, \xnnodes, \, \xtnodes, \, \xroot, \, \xsuccf, \, \xsucct, \, \xlabn, \xlabt)~, %\quad \text{ where }
	\]
	%\sj{I would consider having explicitly labels for inner nodes and terminals... i personally find that way clearer}
	%
		%
	%\begin{enumerate}
		%\item 
		where $\xnodes$ is a finite set of \emph{nodes}, partitioned into \emph{non-terminal nodes} $\xnnodes$ and \emph{terminal nodes} $\xtnodes$, $\xroot\in\xnodes$ is the \emph{root node},
		%
		%\item
		 $\xsuccf,\xsucct\colon \xnnodes \to \xnodes$ are \emph{successor functions} such that all nodes $\xnodea \in \xnodes\setminus\{\xroot\}$ have at least one predecessor. We call $\xsuccf(\xnodea)$ and $\xsucct(\xnodea)$, the \emph{$\false$- \ and $\true$-, successors of $\xnodea$}, respectively.
		%
		%\item 
		The functions
		$\xlabn\colon \xnnodes \to \at$ and $\xlabt \colon \xtnodes \to\term_\sorta $ are \emph{labelings}, where
		%\begin{enumerate}
			%\item $\forall \xnodea\in\xnnodes\colon \xlab(\xnodea) \in \fo$ and $\forall \xnodea\in\xtnodes\colon \xlab(\xnodea) \in \term_\sorta$,
			%
			 $\xlabn$ is compatible with
			$\varord$, i.e., where for any non-terminal  successor $\xnodea'$ of $\xnodea$, $\xlabn(\xnodea) \varord\xlabn(\xnodea')$.
			%
%			\begin{align*}
%				\forall \xnodea,\xnodea'\in\xnnodes\colon \qquad \xsuccf(\xnodea)=\xnodea' ~\text{or}~ \xsucct(\xnodea)=\xnodea' 
%					\qquad\text{implies}\qquad\xlabn(\xnodea) \varord\xlabn(\xnodea')~.
%					\tag*{$\triangle$}
%			\end{align*}
	%	\end{enumerate}
	%\end{enumerate}
\end{definition}
%
%We write $\xlab\colon \xnodes \to \fo \cup \term_\sorta$ for the union of $\xlabn$ and $\xlabt$.
A $\sorta$-TEDD is thus a labeled acyclic directed graph, where the topology and the labeling of the non-terminal nodes respects the order $\varord$ on formulae.  We denote by $\succtxadd{\xadd}$ and $\succfxadd{\xadd}$ the sub-TEDDs of $\xadd$ with root $\xsucct(\xroot)$ and $\xsuccf(\xroot)$, resp. $\xadd$ is called \emph{terminal} if its root is.
% We call a TEDD \emph{atomic}, if all of its non-terminal nodes are labeled by atomic formulae, i.e., by formulae not containing Boolean connectives or quantifiers. \sj{where do we use this?}\kb{TODO, checkmaybe only once} 
We require that terminal nodes of a $\sortbool$-TEDD are labeled by $\true$ or $\false$. The \emph{size} of $\xadd$, written $\sizeof{\xadd}$, is the number of nodes of $\xadd$, i.e., $\sizeof{\xadd}=\sizeof{\xnodes}$.

A $\tau$-TEDD gives rise to a $\tau$-case expression, where each path from the root to a terminal node determines one case by conjoining the (negated) labels encountered along the path: 
%Put formally:
\begin{definition}
	\label{def:switchfun_of_xadd}
	Let $\xadd$ be a $\sorta$-TEDD and let $\xnodea\in\xnodes$ be a node. We define the \emph{$\sorta$-case expression $\xswitchfunn{\xadd}{\xnodea}$ of $\xadd$ and $\xnodea$} recursively on the acyclic topology of $\xadd$:
	%
%	\begin{align*}
%		\xswitchfunn{\xadd}{\xnodea} \eeq 
%		\begin{cases}
%			  \begin{cases}
%			  	 \xlab(\xnodea) & \switchfuncase \true
%			  \end{cases}
%			  & \text{if}~\xnodea \in \xtnodes \\[1.3em]
%			  %
%			  \switchfunop{\switchITEsymbol} \left(\switchfun_{\xlab(\xnodea)}, \xswitchfunn{\xadd}{\xsucct(\xnodea)},\xswitchfunn{\xadd}{\xsuccf(\xnodea)}  \right)
%			  % \switchITE{\xlab(\xnodea)}{\xswitchfunn{\xadd}{\xsucct(\xnodea)}}{\xswitchfunn{\xadd}{\xsuccf(\xnodea)}}
%			  & \text{if}~\xnodea \in \xnnodes ~.
%			  % \tag*{$\triangle$}
%		\end{cases}
%	\end{align*}
	%
	\begin{enumerate}
		\item If $\xnodea \in \xtnodes$, then $\xswitchfunn{\xadd}{\xnodea}= \begin{cases}
			\xlabt(\xnodea) & \switchfuncase \true
		\end{cases}$.
		\item If $\xnodea \in \xnnodes$, then $\xswitchfunn{\xadd}{\xnodea}= \switchfunop{\switchITEsymbol} \Big((\switchfun_{\xlabn(\xnodea)}, \xswitchfunn{\xadd}{\xsucct(\xnodea)},\xswitchfunn{\xadd}{\xsuccf(\xnodea)}  \Big)$. \hfill $\triangle$
	\end{enumerate}
	%\hfill $\triangle$
\end{definition}
We often write $\xswitchfuns$ instead of $\xswitchfunn{\xadd}{\xroot}$ and $\sem{\xadd}{\inta}$ instead of $\sem{\xswitchfuns}{\inta}$.

\begin{example}
	The DAG in \Cref{fig:ex_xadd1} depicts an $\sorteureal$-TEDD $\xadd$ graphically. Solid (resp.\ dashed) arrows point to $\true$- (resp.\ $\false$-) successors. Inner nodes are labeled by formulae whereas terminal nodes are labeled by terms of type $\sorteureal$. The $\sorteureal$-case expression of $\xadd$ as defined in \Cref{def:switchfun_of_xadd} is the expression from \Cref{ex:eurealswitchfun}.  \hfill $\triangle$
\end{example}
%
%
%\begin{definition}
\subsection{Reducedness and Canonicity of TEDDs}
\label{sec:tedds:reducedness}
	\begin{figure}[t]
\centering
\begin{subfigure}[t]{0.48\textwidth}
\centering
\scalebox{0.75}{%
\begin{tikzpicture}[
	node distance=6mm and 3mm,
	decision/.style={ellipse, draw=black!70, line width=0.45pt, fill=white, inner sep=2.5pt, minimum width=17mm, minimum height=8mm, font=\small},
	leaf/.style={rectangle, rounded corners=1.4pt, draw=black!70, line width=0.45pt, fill=white, inner sep=1.7pt, minimum width=11.5mm, minimum height=6.8mm, font=\small},
	edgefalse/.style={-{Stealth[length=1.9mm]}, dashed, line width=0.65pt, draw=black!75},
	edgetrue/.style={-{Stealth[length=1.9mm]}, line width=0.65pt, draw=black!75}
]
	\node[decision] (l0) {$\vara \foeq \vara$};
	\node[leaf, below left=10mm and 3mm of l0] (l1) {$\terma_1$};
	\node[leaf, below right=10mm and 3mm of l0] (l2) {$\terma_2$};

	\draw[edgetrue] (l0) -- (l1);
	\draw[edgefalse] (l0) -- (l2);
\end{tikzpicture}%
}
\caption{TEDD $\xadd$.}
\label{fig:fo-canonicity-counterexample:x}
\end{subfigure}\hfill
\begin{subfigure}[t]{0.48\textwidth}
\centering
\scalebox{0.75}{%
\begin{tikzpicture}[
	node distance=6mm and 3mm,
	decision/.style={ellipse, draw=black!70, line width=0.45pt, fill=white, inner sep=2.5pt, minimum width=17mm, minimum height=8mm, font=\small},
	leaf/.style={rectangle, rounded corners=1.4pt, draw=black!70, line width=0.45pt, fill=white, inner sep=1.7pt, minimum width=11.5mm, minimum height=6.8mm, font=\small},
	edgefalse/.style={-{Stealth[length=1.9mm]}, dashed, line width=0.65pt, draw=black!75},
	edgetrue/.style={-{Stealth[length=1.9mm]}, line width=0.65pt, draw=black!75}
]
	\node[decision] (r0) {$\varb \foeq \varb$};
	\node[leaf, below left=10mm and 3mm of r0] (r1) {$\terma_1$};
	\node[leaf, below right=10mm and 3mm of r0] (r2) {$\terma_2$};

	\draw[edgetrue] (r0) -- (r1);
	\draw[edgefalse] (r0) -- (r2);
\end{tikzpicture}%
}
\caption{TEDD $\xaddb$.}
\label{fig:fo-canonicity-counterexample:y}
\end{subfigure}

\caption{Counterexample to canonicity modulo first-order logical equivalence. Let $\hastype{\vara,\varb}{\sortint}$. The labels $\vara\foeq\vara$ and $\varb\foeq\varb$ are distinct Boolean atoms, so the two reduced TEDDs are not propositionally equivalent. Under first-order semantics, however, both labels are valid in every structure, and therefore both TEDDs represent the same function from first-order interpretations to values.}
\label{fig:fo-canonicity-counterexample}
\end{figure}
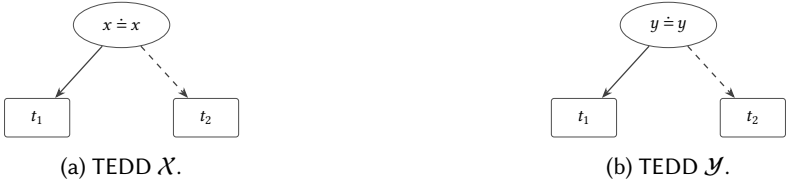
Reducedness and canonicity are important aspects of classic BDDs \cite{DBLP:journals/tc/Bryant86} and ADDs \cite{adds}.
Let us now define an appropriate notion of reducedness for TEDDs and investigate in what sense TEDDs are canonical representations of case expressions.

	Let $\sorta \in \sorts$. Given a $\sorta$-case expression $\switchfun$ and a propositional interpretation $\propinta$, denote by $\propsem{\switchfun}{\propinta}$ the unique term $\terma_i$ such that $\propsem{\forma_i}{\propinta} =1$. We call two $\sorta$-case expressions $\switchfun_1,\switchfun_2$ \emph{propositionally equivalent}, denoted $\switchfun_1 \switchfunequiv \switchfun_2$, if $\propsem{\switchfun_1}{\propinta} = \propsem{\switchfun_2}{\propinta}$ for all propositional interpretations $\propinta$.
	%
	% \[
	% 	\forall\, \text{propositional interpretations $\propinta$}\colon \propsem{\switchfun_1}{\propinta} = \propsem{\switchfun_2}{\propinta}~.
	% \]
	\begin{definition}
	We call a $\sorta$-TEDD $\xadd$ \emph{reduced}, if for all nodes $\xnodea_1,\xnodea_2\in\xnodes$, we have 
	\[
		\xswitchfunn{\xadd}{\xnodea_1} \switchfunequiv \xswitchfunn{\xadd}{\xnodea_2}
		\quad\text{implies}\quad
		\xnodea_1 = \xnodea_2~.
	\]
	\end{definition}
	Reducedness implies that $\xsuccf(\xnodea) \neq \xsucct(\xnodea)$ for all $\xnodea\in\xnnodes$. From now on \emph{we assume all TEDDs to be reduced}. With this notion of reducedness, we get the following canonicity property for TEDDs:
	\begin{theorem}
		Fix a linear order $\varord$ on $\at$. Then $\sorta$-TEDDs are canonical representations of $\sorta$-case expressions modulo propositional equivalence $\switchfunequiv$, i.e., for every $\switchfunequiv$-equivalence class $\mathcal{C}$ of $\sorta$-case expressions, there is unique (up to isomorphism) $\sorta$-TEDD $\xadd$ such that $\xswitchfuns\in\mathcal{C}$.
	\end{theorem}
	Importantly, this canonicity result is propositional and does not identify TEDDs whose labels are merely equivalent under first-order semantics, as illustrated in \Cref{fig:fo-canonicity-counterexample}: Both $\xadd$ and $\xaddb$ are reduced, yet $\xadd$ and $\xaddb$ are not isomorphic and $\sem{\xadd}{\inta}=\sem{\xaddb}{\inta}$ for all interpretations $\inta$.

TEDDs subsume classic BDDs \cite{DBLP:journals/tc/Bryant86} and ADDs \cite{adds}, as propositional variables are  formulae:
\begin{theorem} 
	%We have:
	%\begin{enumerate}
		Classic BDDs are $\sortbool$-TEDDs where all non-terminal nodes are labeled by propositional variables $\hastype{\vara}{\sortbool}$ and all terminals nodes are labeled with either $\true$ or $\false$, while % are classic BDDs.
		classic ADDs are $\sortrat$-TEDDs with the same restriction as above. % are classic ADDs.
	%\end{enumerate}
	
\end{theorem}

%\kbinline{old def start}
%Next, we define what it means for an XADD to be \emph{pruned} modulo a given theory.
%A \emph{path} in $\xadd$ is a non-empty finite sequence $\xroot,\xnodea_1,\ldots,\xnodea_n$ of non-terminal nodes , where $\xroot$ is the root and $\xnodea_{i+1}$ is some successor of $\xnodea_i$ for all $i\in\{1,\ldots,n-1\}$. Reducedness of $\xadd$ implies that $\xsuccf(\xnodea) \neq \xsucct(\xnodea)$ for all $\xnodea\in\xnnodes$. Hence, every path in $\xadd$ can be assgined a unique cube $\xpathform{\xroot,\xnodea_1,\ldots,\xnodea_n} = \xlab(\xroot) \wedge \ldots \wedge \xlab(\xnodea_n)$ obtained from conjoining the labels of the nodes along the path. This yields the following:
%%
%%
%\begin{definition}
%	Let $\theory$ be a theory. We say that a $\sorta$-XADD $\xadd$ is \emph{pruned modulo $\theory$}, if $\xpathform{\xroot,\xnodea_1,\ldots,\xnodea_n}$ is satisfiable modulo $\theory$ for all\footnote{Notice that if $\xroot$ is terminal, then there is no path in $\xadd$, which yields this condition to be vacuously satisfied.} paths $\xroot,\xnodea_1,\ldots,\xnodea_n$ in $\xadd$.
%	\hfill $\triangle$
%\end{definition}
%%
%Hence, a pruned XADD does not contains \enquote{superfluous} paths when working within the theory $\theory$. In \Cref{sec:xadds:operations}, we provide an algorithm for pruning XADDs.
%\kbinline{old def end}
\subsection{Theory-Aware Treatment of TEDDs}
Reducedness of TEDDs refers to propositional interpretations, therefore disregarding theory-aware facts such as the unsatisfiability of certain conjunctions of decisions. We now take first-order reasoning into account. Towards this end, we introduce:

\begin{definition}
	\label{def:pruned_tedd}
	Let $\xadd$ be a $\tau$-TEDD and let $\theory$ be a theory. A \emph{rooted path} in $\xadd$ is a sequence
	\[
		\patha \eeq u_0 \xrightarrow{b_1} u_1 \xrightarrow{b_2} \cdots \xrightarrow{b_k} u_k
	\]
	where $u_0=\xroot$, $b_i\in\{0,1\}$, and $u_i=\xsuccsub{b_i}(u_{i-1})$. The \emph{path condition} of $\patha$ is
	\[
		\pathcond{\xadd}{\patha}
		\eeq
		\bigwedge_{i=1}^{k}
		\begin{cases}
			\xlabn(u_{i-1}) & \text{if } b_i=1,\\
			\neg \xlabn(u_{i-1}) & \text{if } b_i=0.
		\end{cases}
	\]
	with the empty conjunction equal to $\true$. For a formula $\contexta\in\fo$, we say that $\xadd$ is \emph{pruned modulo $\theory$ under context $\contexta$} if $\contexta\wedge\pathcond{\xadd}{\patha}$ is satisfiable modulo $\theory$ for every rooted path $\patha$ ending in a terminal node. We say that $\xadd$ is \emph{pruned modulo $\theory$} if it is pruned modulo $\theory$ under context $\true$.
\end{definition}
Consider the $\sorteureal$-TEDD $\xadd$ depicted in \Cref{fig:ex_xadd1}. $\xadd$ is pruned modulo $\theoryla$ because every terminal rooted path has a path condition satisfiable modulo $\theoryla$. The TEDD $\xaddb$ depicted in \Cref{fig:ex_xadd2}, on the other hand, is \emph{not} pruned modulo $\theoryla$: The path condition $\vara\neq 0 \wedge \vara=1\wedge\vara=2$ of the longest path ending in the terminal labeled with $\varb$ is unsatisfiable modulo $\theoryla$. 
%Let $\xadd$ be a $\tau$-XADD with 
%%
%\[
%	\xswitchfuns \eeq 
%	\begin{cases}
%		\terma_1 &\switchfuncase \forma_1 \\
%		%
%		\vdots  &  \\
%		%
%		\terma_n &\switchfuncase\forma_n~,
%	\end{cases}
%\]
%

Finally, given a theory $\theory$, we say that two $\sorta$-TEDDs $\xadd_1,\xadd_2$ are \emph{equivalent modulo $\theory$}, if $\sem{\xadd_1}{\inta}= \sem{\xadd_2}{\inta}$ for all $\theory$-interpretations $\inta$. As an example, the two TEDDs depicted in \Cref{fig:ex_xadd1,fig:ex_xadd2} are equivalent modulo $\theoryla$. In the next subsection, we present an algorithm that transforms every TEDD into an equivalent pruned one modulo $\theory$, assuming decidability of satisfiability modulo $\theory$.
%
%\[
%	\forall~ \theory\text{-interpretations}~\inta \colon \quad \xswitchfunn{\xadd_1}{}(\inta) \eeq \xswitchfunn{\xadd_2}{}(\inta)~,
%\]
%

\section{Operations on Typed Extended Decision Diagrams}
\label{sec:xadds:operations}
We now provide central algorithms on TEDDs, ultimately enabling TEDD-based deductive verification. All algorithms
%for applying term operators, substituting variables by terms, and for pruning. 
produce (reduced) TEDDs, which respect the fixed order $\varord$ on formulae in $\at$. Moreover, we equip all algorithms with their formal specifications.

These algorithms use an overloaded procedure \ObtainAlg. First, given a term $\hastype{\terma}{\sorta}$, $\obtain(\terma)$ returns the $1$-node $\sorta$-TEDD labelled with $\terma$. If $\sorta=\sortbool$, we require that $\terma \in \at$, and $\obtain(\terma)$ returns the natural $3$-node $\sortbool$-TEDD with root label $\terma$ and terminal label $\true$ and $\false$. Second, given $\forma\in\at$ and $\sorta$-TEDDs $\xadd,\xaddb$ with $\forma \varord \xlabn(\xadd)$ (resp.\ $\forma \varord \xlabn(\xaddb$) if $\xadd$ (resp.\ $\xaddb$) is non-terminal, $\obtain(\forma,\xadd,\xaddb)$ returns a $\sorta$-TEDD $\xaddc$ with root label $\forma$ and $\xsucct(\xaddc) = \xadd$ and $\xsuccf(\xaddc) = \xaddb$, i.e.,
\begin{align}
	\label{eq:obtainspec}
\forall \inta\in\interprets\colon\quad
				\xswitchfun{\xaddc}(\inta) \eeq
				\begin{cases}
					\xswitchfun{\xadd}(\inta) & \text{if $\inta \models \forma$,} \\
					\xswitchfun{\xaddb}(\inta) & \text{if $\inta~\not\models~ \forma$ .} 
				\end{cases}
\end{align}
We ensure reducedness of $\xaddc$ by returning $\xadd$ in case $\xadd=\xaddb$, so that \eqref{eq:obtainspec} is trivially satisfied, and assume that repeated invocations of $\ObtainAlg(\terma)$ for the same $\terma$ always return the same node.\sj{So obtain is a function? or it is deterministic procedure? }\kb{above we say "procedure". What is the difference?}
% 
%\kb{commented out obtain for now}
% \subsubsection{Obtain}
% \kbinline{consider reducedness problem later}
% %
% \input{appendix/obtain}
%
\subsection{Applying Term Operators}
\sj{Should we capitilize captions of algorithms like this?}
\begin{algorithm}[t]
	\caption{Applying Term Operators to TEDDs} 	\label{alg:apply}
	\Fn{\apply{$\termop, \xadd_1, \ldots, \xadd_n$}}{
		\KwIn{%
			A term operator $\termop\colon\term_{\sorta_1}\times\ldots\times\term_{\sorta_n} \to \term_\sorta$.\newline
			If $\sorta=\sortbool$, then $\termop$ returns only $\true$ or $\false$ on the terminal labels reached by this call.\newline
			$\sorta_1$-TEDD $\xadd_1$, $\ldots$, $\sorta_n$-TEDD $\xadd_n$.
		}
		\KwResult{%
			$\sorta$-TEDD $\xadd$ satisfying
			$\forall \inta\in\interprets\colon
			\sem{\xadd}{\inta}\eeq
			\sem{\switchfunop{\termop}(\xswitchfun{\xadd_1},\ldots,\xswitchfun{\xadd_n})}{\inta}$.
		}
		\BlankLine
		\If{\nllabel{alg:apply:cache1} there is a cache entry $\xadd$ for $(\termop, \xadd_1, \ldots, \xadd_n)$}{%
			\KwRet $\xadd$\;
		}\nllabel{alg:apply:cache2}
		\eIf{\label{alg:apply:base_case1}$\xadd_1, \ldots, \xadd_n$ are all terminal}{%
			$\terma \leftarrow \termop(\xlabt(\xadd_1), \ldots, \xlabt(\xadd_n))$\;\label{aloapply:base_case1}
			$\xadd \leftarrow$ \obtain{$\terma$}\;\label{aloapply:base_case2}
			%Obtain MTBDD $\treea$ with terminal node $\node$ that is labeled with $\labelfunc(\node) = \selma$
		}{\nllabel{algo_apply:recur_case1}
			$\forma \leftarrow \min_{\varord} \{\xlabn(\xadd_i) ~|~ \xadd_i~\text{is non-terminal and}~i \in \{1,\ldots,n\} \}$\;\nllabel{alg:apply:topmost}
			\ForEach{\nllabel{alg:apply:traverse1}$\xadd_i$ \In $\xadd_1, \ldots, \xadd_n$}{%
				\eIf{$\xadd_i$ is non-terminal and $\xlabn(\xadd_i)=\forma$}{%
					$\xaddb_i \leftarrow \succtxadd{\xadd_{i}}$;\quad
					$\xaddc_i \leftarrow \succfxadd{\xadd_{i}}$\;\nllabel{alg:apply:traverse3}
				}{%
					$\xaddb_i \leftarrow \xadd_{i}\,$;\quad
					$\xaddc_i \leftarrow \xadd_{i}$\;\nllabel{alg:apply:traverse4}
				}
			}\nllabel{algapply:traverse2}
			$\xaddb \leftarrow \apply{$\termop, \xaddb_1, \ldots, \xaddb_n$}$;\quad
			$\xaddc \leftarrow \apply{$\termop, \xaddc_1, \ldots, \xaddc_n$}$\;\nllabel{alg:apply:recur}
			$\xadd \leftarrow$ \obtain{$\forma$, $\xaddb$, $\xaddc$}\;\nllabel{alg:apply:reduce}
		}\nllabel{algo:apply:recur_case2}
		Store $\xadd$ in the cache for $(\termop, \xadd_1, \ldots, \xadd_n)$\;\nllabel{alg:apply:cache3}
		\KwRet $\xadd$\;\nllabel{alg:apply:return}
}
\end{algorithm}
\sj{I didnt understand the sentences regarding the input in the algorithm 1}
\sj{Do we formalize what an operation cache is? should we?}
\sj{I would maybe write a cache entry exists rather than there is a cache entry, but as this appears multiple times, lets first discuss}
Let  $\termop\colon\term_{\sorta_1}\times\ldots\times\term_{\sorta_n} \to \term_\sorta$ be an operator on terms, such as addition (for $\sorta=\sorteureal$) or disjunction (for $\sorta=\sortbool$). \Cref{alg:apply} gives \ApplyAlg, a generalization of Bryant's \ApplyAlg algorithm \cite{DBLP:journals/tc/Bryant86} for BDDs to $n$-ary term operators $\termop$. Given $\sorta_1$-TEDD $\xadd_1$, \ldots, $\sorta_n$-TEDD $\xadd_n$, \apply{$\termop,\xadd_1,\ldots,\xadd_n$} operates on these TEDDs and yields a $\sorta$-TEDD $\xadd$ representing the $\sorta$-case expression obtained from applying $\termop$ to the case expressions represented by $\xadd_1,\ldots,\xadd_n$, i.e.,
\[
	\xswitchfun{\xadd} \eeq
	\switchfunop{\termop}(\xswitchfun{\xadd_1},\ldots,\xswitchfun{\xadd_n})~.
\]%\kb{make sure right equality used...}
\ApplyAlg traverses the input TEDDs in a $\varord$-preserving\sj{traversing a tree in an order-preserving manner sounds very strange. Is the point that the order is preserved, or that we follow the order when traversing?} manner (ll.\ \ref{alg:apply:traverse1}-\ref{alg:apply:reduce}) and invokes itself recursively (l.\ \ref{alg:apply:recur}) until hitting terminal nodes, in which case $\termop$ is applied to the labels of the respective leafs (l.\ \ref{aloapply:base_case1}). As with BDDs, we \sj{Should we use the operation cache terminology?} cache intermediate results from recursive calls (l.\ \ref{alg:apply:cache3}) to avoid re-computations that may arise due to the DAG-structure of TEDDs. This ensures the following:
\begin{theorem}
\label{thm:apply_soundness_main}
	\Cref{alg:apply} is sound and encounters at most
	$
		\sizeof{\xadd_1}\cdot\ldots\cdot\sizeof{\xadd_n}
	$
	cache misses.
\end{theorem}
\begin{proof}
	See \Cref{app:proof_apply}.
\end{proof}
\sj{Do we need/want a proof environment for referring to the appendix?}
\begin{example}
	Let $\xaddb$ be a $\sortbool$-TEDD and let $\xaddb_1,\xaddb_2$ be $\sorteureal$-TEDDs. Then 
	\begin{align*}
	\text{$\xadd = \apply(\switchITEsymbol, \xaddb, \xaddb_1, \xaddb_2)$}
	\quad\text{satisfies}\quad
		\forall \inta\in\interprets\colon 
		\sem{\xadd}{\inta} 
		\eeq
		\begin{cases}
			\sem{\xaddb_1}{\inta} & \text{if}~\sem{\xaddb}{\inta} = \true \\
			\sem{\xaddb_2}{\inta} & \text{if}~\sem{\xaddb}{\inta} = \false~. 
		\end{cases}
		%\tag*{$\triangle$}
	\end{align*}
   Moreover, 
   \begin{align*}
	\text{$\xadd = \apply(+, \xaddb_1, \xaddb_2)$}
	\quad\text{satisfies}\quad
		\forall \inta\in\interprets\colon 
		\sem{\xadd}{\inta} 
		\eeq
		\sem{\xaddb_1}{\inta} + \sem{\xaddb_1}{\inta}~.
		\tag*{$\triangle$}
	\end{align*}
\end{example}
%
%As with BDDs, we cache intermediate results (l.\ \ref{alg:apply:cache1} and l.\ \ref{alg:apply:cache3}) from recursive calls to ensure that the number of recursive calls is linear in the number of nodes of the input XADDs. We then traverse the input XADDs recursively: If all input XADDs are terminal, we apply $\termop$ to the terms given by the respective labels and construct\footnote{This is realized by the straightforward algorithm \obtain provided in \Cref{app:obtain}.} the $\sorta$-XADD representing the so-obtained term (l.\ \ref{aloapply:base_case2}). Otherwise, we keep traversing the input XADDs in a $\varord$-preserving manner (ll.\ \ref{alg:apply:topmost}-\ref{algapply:traverse2}), recursively invoke \apply on the so-obtained sub-XADDs (l.\ \ref{alg:apply:recur}), and construct the resulting $\sorta$-XADD (l. \ref{alg:apply:reduce}). This yields:
%
%%
%\begin{theorem}
%	Algorithm \ref{alg:apply} is sound and terminates. Moreover, \apply{$\termop, \xadd_1, \ldots, \xadd_n$} requires at most $\sizeof{\xadd_1} \cdot \ldots\cdot \sizeof{\xadd_n}$ recursive calls, where $\sizeof{\xadd_i}$ denotes the number of nodes in $\xadd_i$.
%\end{theorem}
%
%\kb{maybe give concrete example also illustrating why caching is necessary}

\subsection{Substitution}
%\kb{Since TEDDs are DAGs, different paths may reach the same sub-TEDD. We cache calls to Substitute so that each shared sub-TEDD is substituted at most once for a fixed variable and replacement term.}
\label{sec:xadds:substitution}
\begin{algorithm}[t]
	\caption{Substitution of Variables by Terms in TEDDs}
	\label{alg:substitute}
	\Fn{\substitute{$\xadd, \vara, \terma$}}{%
		\KwIn{%
			A $\sorta$-TEDD $\xadd$, a variable $\hastype{\vara}{\sorta'}$, and a term $\hastype{\terma}{\sorta'}$.
		}
		\KwResult{%
			A $\sorta$-TEDD $\xaddb$ satisfying
			$\forall \inta\in\interprets\colon
			\sem{\xaddb}{\inta}\eeq
			\sem{\xswitchfunn{\xadd}{}\switchfunsubst{\vara}{\terma}}{\inta}$.
		}
	}
	\BlankLine
	\If{\nllabel{algosubst:cache_lookup}there is a cache entry $\xaddb$ for $(\xadd, \vara, \terma)$}{%
		\KwRet $\xaddb$\;
	}
	\eIf{\nllabel{algosubst:terminal_test}$\xadd$ is terminal}
	{
		$\xaddb \leftarrow \obtain{$\xlabt(\xadd)\switchfunsubst{\vara}{\terma}$}$\nllabel{algosubst:terminal_obtain}
	}{
		$\forma \leftarrow \xlabn(\xadd)\switchfunsubst{\vara}{\terma}$\;\nllabel{algosubst:subst_guard}
		$\xaddb^+ \leftarrow \substitute(\succtxadd{\xadd}, \vara,\terma)$;\quad $\xaddb^- \leftarrow \substitute(\succfxadd{\xadd}, \vara,\terma)$\;\label{algosubst:recurse}
		$\xaddb \leftarrow \apply{$\switchITEsymbol, \obtain(\forma), \xaddb^+, \xaddb^-$}$\label{algosubst:finalobtain}
	}
	Store $\xaddb$ in the cache for $(\xadd, \vara, \terma)$\;\label{algosubst:cache}
	\KwRet $\xaddb$\;
\end{algorithm}

 While substituting variables by terms is --- in principle --- a rather straightforward operation, we must be careful: Replacing variables by terms modifies the labels \sj{of the inner nodes?}\kb{potentially all of them, no?}\sj{but the point is that the inner nodes may need reordering?}\kb{yes} so that a naive implementation might violate the order $\varord$. \Cref{alg:substitute} gives \SubstituteAlg, which remedies this:
\substitute{$\xadd, \vara, \terma$} yields a TEDD $\xaddb$ with
$
	\xswitchfunn{\xaddb}{} = \xswitchfun{\xadd}\switchfunsubst{\vara}{\terma}
$ by recursively traversing $\xadd$. Assuming $\xaddb^+$ (resp.\ $\xaddb^-$) are obtained from replacing $\vara$ by $\terma$ in $\succtxadd{\xadd}$ (resp.\ $\succfxadd{\xadd}$) (l.\ \ref{algosubst:recurse}), we obtain the sought-after TEDD by invoking 
$
	\apply(\switchITEsymbol, \obtain(\forma), \xaddb^+, \xaddb^-)%,
	%~\text{where}~
	%\underbrace{\xadd_\forma = \obtain(\forma,\obtain(\true),\obtain(\false)}_{\text{$2$-node $\sortbool$-TEDD representing $
	%\forma$}}
$ 
in (l.\ \ref{algosubst:finalobtain}).
 Since \ApplyAlg preserves $\varord$, this procedure is sound. Again, we cache intermediate results to prevent unnecessary re-computations (l.\ \ref{algosubst:cache}). %This gives us:
\begin{theorem}
\label{thm:substitute_soundness_main}
	\Cref{alg:substitute} is sound. Moreover, it performs at most
	$
		\sizeof{\xadd}
	$
	recursive calls.
\end{theorem}
\begin{proof}
	See \Cref{app:proof_substitute}.
\end{proof}
% \kb{elaborate on the obtain thingy and on caching, maybe with example (better for apply)}

%\kb{complexity?}

\subsection{Pruning}
\label{sec:xadds:pruning}
\begin{algorithm}[t]
	\caption{Pruning of TEDDs Modulo a Theory}
	\label{alg:prune}
	%\SetKwFunction{unsat}{UNSAT}
	%\setcounter{AlgoLine}{0}
	%\newcommand{\rep}{\text{result}^{+}}
	%\newcommand{\rem}{\text{result}^{-}}
	\Fn{\prune{$\theory$, $\xadd$, $\contexta$}}{%
		\KwIn{%
			A theory $\theory$, a $\sorta$-TEDD $\xadd$, 
			\newline
			and $\contexta\in\fo$ satisfiable modulo $\theory$.%\newline
			%A history $\hist$ (that is SAT in $\theo$) of the current path to $\xda$.
			%Initially $\fot$.
		}
		\KwResult{%
			A $\sorta$-TEDD $\xaddb$ equivalent to $\xadd$ under $\contexta$ and pruned modulo $\theory$ under context $\contexta$.%, and\newline
%			and if $\hist' \in \hists[\xdb]$ then
%			$\hist \wedge \hist'$ is SAT in $\theo$.
		}
		\BlankLine
		\If{$\xadd$ is terminal}{%
			\KwRet{$\xadd$};\nllabel{alg:prune:base}
		}
		\If{there is a cache entry $\xaddb$ for $(\theory, \xadd, \contexta)$}{%
			\KwRet $\xaddb$\;
		}
		$\contexta^{+} \leftarrow \contexta \wedge \xlabn(\xadd)$; %\;%\nllabel{alg:prune:hist1}
		$\contexta^{-} \leftarrow \contexta \wedge \neg \xlabn(\xadd)$\;\nllabel{alg:prune:hist2}
%			$\rep \leftarrow \solvesmt(\theo, \hist^{+})$;
%			$\rem \leftarrow \solvesmt(\theo, \hist^{-})$\;\nllabel{algo_prune:sat1}
		\uIf{$\neg\sats{\contexta^{+}}$}{%
			$\xaddb \leftarrow \prune(\theory, \succfxadd{\xadd}, \contexta^{-})$\;\nllabel{alg:prune:prune1}
		}\uElseIf{$\neg\sats{\contexta^{-}}$}{%
			$\xaddb \leftarrow \prune(\theory, \succtxadd{\xadd}, \contexta^{+})$\;\nllabel{alg:prune:prune2}
		}\Else{%
			$\xaddb^{+} \leftarrow \prune(\theory, \succtxadd{\xadd}, \contexta^{+})$;\quad 
			$\xaddb^{-} \leftarrow \prune(\theory, \succfxadd{\xadd}, \contexta^{-})$\;\nllabel{alg:prune:prune3}
			$\xaddb \leftarrow \obtain(\xlabn(\xadd), \xaddb^{+}, \xaddb^{-})$\;\nllabel{algo_prune:prune4}
		}
		Store $\xaddb$ in the cache for $(\theory, \xadd, \contexta)$\;
		\KwRet $\xaddb$;
	}
\end{algorithm}

\begin{figure}[t]
\centering
\begin{subfigure}[t]{0.48\textwidth}
\centering
\scalebox{0.75}{%
\begin{tikzpicture}[
	node distance=6mm and 3mm,
	decision/.style={ellipse, draw=black!70, line width=0.45pt, fill=white, inner sep=2.5pt, minimum width=17mm, minimum height=8mm, font=\small},
	deaddecision/.style={ellipse, draw=red!70!black, line width=0.7pt, fill=red!8, inner sep=2.5pt, minimum width=17mm, minimum height=8mm, font=\small},
	leaf/.style={rectangle, rounded corners=1.4pt, draw=black!70, line width=0.45pt, fill=white, inner sep=1.7pt, minimum width=11.5mm, minimum height=6.8mm, font=\small},
	keepleaf/.style={rectangle, rounded corners=1.4pt, draw=green!45!black, line width=0.7pt, fill=green!8, inner sep=1.7pt, minimum width=11.5mm, minimum height=6.8mm, font=\small},
	edgefalse/.style={-{Stealth[length=1.9mm]}, dashed, line width=0.65pt, draw=black!75},
	edgetrue/.style={-{Stealth[length=1.9mm]}, line width=0.65pt, draw=black!75},
	deadedge/.style={-{Stealth[length=1.9mm]}, line width=0.95pt, draw=red!70!black},
	keepedge/.style={-{Stealth[length=1.9mm]}, dashed, line width=0.95pt, draw=green!45!black},
	redirectedge/.style={-{Stealth[length=2.2mm]}, line width=1.15pt, draw=blue!65!black},
	smalllabel/.style={font=\large, inner sep=0.6pt, minimum width=7mm, align=center}
]
	\node[decision] (l0) {$\vara=0$};
	\node[decision, below right=of l0] (l1) {$\vara=1$};
	\node[deaddecision, below left=of l1] (l2dead) {$\vara=2$};
	\node[decision, below right=of l1] (l2keep) {$\vara=2$};
	\node[leaf, below left=10mm and 3mm of l2dead] (ly) {$\varb$};
	\node[keepleaf, below=10mm of l2dead] (ly1) {$\varb+1$};
	\node[leaf, below=10mm of l2keep] (ly2) {$\varb+2$};
	\node[leaf, below right=10mm and 3mm of l2keep] (ly3) {$\varb+3$};

	\draw[edgetrue] (l0) edge[bend right] (ly);
	\draw[edgefalse] (l0) -- (l1);
	\draw[redirectedge] (l1) -- (l2dead);
	\draw[edgefalse] (l1) -- (l2keep);
	\draw[deadedge] (l2dead) -- node[smalllabel, text=red!70!black, pos=0.50, left=0.2mm] {$\varphi^+$} (ly);
	\draw[keepedge] (l2dead) -- node[smalllabel, text=green!45!black, pos=0.50, right=0.2mm] {$\varphi^-$} (ly1);
	\draw[edgetrue] (l2keep) -- (ly2);
	\draw[edgefalse] (l2keep) -- (ly3);
\end{tikzpicture}%
}
\caption{$\sorteureal$-TEDD $\xadd$ not pruned modulo $\theoryla$.}
\label{fig:prune_rewire:nonpruned}
\end{subfigure}\hfill
\begin{subfigure}[t]{0.48\textwidth}
\centering
\scalebox{0.75}{%
\begin{tikzpicture}[
	node distance=6mm and 3mm,
	decision/.style={ellipse, draw=black!70, line width=0.45pt, fill=white, inner sep=2.5pt, minimum width=17mm, minimum height=8mm, font=\small},
	leaf/.style={rectangle, rounded corners=1.4pt, draw=black!70, line width=0.45pt, fill=white, inner sep=1.7pt, minimum width=11.5mm, minimum height=6.8mm, font=\small},
	keepleaf/.style={rectangle, rounded corners=1.4pt, draw=green!45!black, line width=0.7pt, fill=green!8, inner sep=1.7pt, minimum width=11.5mm, minimum height=6.8mm, font=\small},
	edgefalse/.style={-{Stealth[length=1.9mm]}, dashed, line width=0.65pt, draw=black!75},
	edgetrue/.style={-{Stealth[length=1.9mm]}, line width=0.65pt, draw=black!75},
	redirectedge/.style={-{Stealth[length=2.2mm]}, line width=1.15pt, draw=blue!65!black}
]
	\node[decision] (r0) {$\vara=0$};
	\node[leaf, below left=10mm and 3mm of r0] (ry) {$\varb$};
	\node[decision, below right=of r0] (r1) {$\vara=1$};
	\node[keepleaf, below left=10mm and 3mm of r1] (ry1) {$\varb+1$};
	\node[decision, below right=of r1] (r2) {$\vara=2$};
	\node[leaf, below left=10mm and 3mm of r2] (ry2) {$\varb+2$};
	\node[leaf, below right=10mm and 3mm of r2] (ry3) {$\varb+3$};

	\draw[edgetrue] (r0) -- (ry);
	\draw[edgefalse] (r0) -- (r1);
	\draw[redirectedge] (r1) -- (ry1);
	\draw[edgefalse] (r1) -- (r2);
	\draw[edgetrue] (r2) -- (ry2);
	\draw[edgefalse] (r2) -- (ry3);
\end{tikzpicture}%
}
\caption{Result of \PruneAlg($\theoryla,\xadd$).}
\label{fig:prune_rewire:pruned}
\end{subfigure}
\caption{Illustration of applying \PruneAlg to the TEDD from \Cref{fig:ex_xadd2}.}
\label{fig:prune_rewire}
\end{figure}
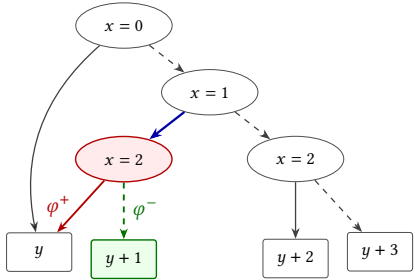
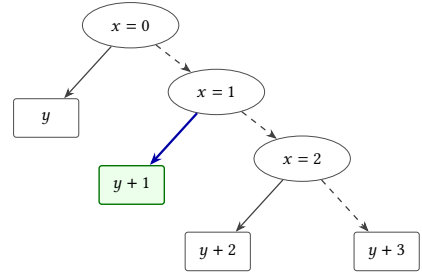
\Cref{alg:prune} implements theory-aware pruning of TEDDs in the sense of \Cref{def:pruned_tedd}. Assuming an oracle --- an SMT solver --- for checking the satisfiability of $\fo$ formulae modulo a given theory $\theory$, \prune{$\theory$, $\xadd$, $\true$} is a pruned $\sorta$-TEDD equivalent to the $\sorta$-TEDD $\xadd$ modulo $\theory$. The third argument of \PruneAlg is the path condition accumulated before the current sub-TEDD. The algorithm traverses $\xadd$ depth-first while conjoining the (negated, depending on the chosen successor) formulae encountered (l.\ \ref{alg:prune:hist2}). If the so-obtained path condition becomes unsatisfiable, the algorithm deletes the \enquote{unreachable}\sj{if the tree is a dag, then a path being unreachable doesnt mean a subtree is unreachable?  Maybe we can be a bit more specific here?}\sj{in the presence of an operations cache, what does it mean to delete a subtree.} (modulo $\theory$) sub-tree and redirects the edges accordingly. %This is demonstrated in the following example:
\begin{example}
\prune is illustrated in \Cref{fig:prune_rewire} for the TEDDs shown in \Cref{fig:ex_xadd}: Since $\vara\neq 0 \wedge \vara=1 \wedge \vara=2$ is unsatisfiable modulo $\theoryla$, the node labeled by $\vara=2$ can be deleted, and the respective edge of the predecessor labeled by $\vara=1$ can be redirected to $\varb+1$. \hfill $\triangle$
\end{example}
%
%For soundness and time complexity, we get:
%
\begin{theorem}
\label{thm:prune_soundness_main}
	Assume that $\contexta \in \fo$ is satisfiable modulo $\theory$. If \prune{$\theory$, $\xadd$, $\contexta$} (\Cref{alg:prune}) returns $\xaddb$, then $\xaddb$ is equivalent to $\xadd$ under all $\theory$-interpretations satisfying $\contexta$ and is pruned modulo $\theory$ under context $\contexta$. Moreover, \Cref{alg:prune} performs at most $2^{m+1}-1$ recursive calls, where $m$ is the maximum number of edges\sj{number of edges of a path = length of a path?}\kb{I find number of edges more precise. Otherwise I'm always wondering: number of nodes on the path? edges?} of a directed path from $\xadd$'s root to a terminal node (cf.\ \Cref{def:tedd_rank}).
\end{theorem}
\begin{proof}
	See \Cref{app:proof_prune}.
\end{proof}

%We now consider theory-aware pruning of XADDs. Algorithm \ref{alg:prune} takes as input a $\sorta$-XADD $\xadd$, a theory $\theory$, and a formula $\forma$ (initially $\true$) and returns a pruned $\sorta$-XADD $\xaddb$ equivalent to $\xadd$ modulo $\theory$. This is is done in a recursive manner and assumes an oracle for checking the satisfiability of formulae modulo $\theory$: 
%%
%\kbinline{caching?}
%\input{algorithms/prune}
%\begin{theorem}
%	Algorithm \ref{alg:apply} is sound and terminates.
%\end{theorem}

%

\section{Syntactic Weakest Pre-Expectations via TEDDs}

%\subsection{Syntactic Weakest Pre-Expectations via TEDDs}
\label{sec:xwp}
\begin{table}[t]
	\begin{center}
		\begin{tabularx}{\textwidth}{X@{\quad}l@{\quad~~}lX}
			\toprule
			\toprule
			$\boldsymbol{\cc}$ & $\mathsf{\mathbf{\overline{wp}}}\boldsymbol{\llbracket\cc \rrbracket(\xadd)}$  \\[0.5ex]
			\hline 
			%\midrule
			$\SKIP$ & $\xadd$  \rule{0pt}{3.5ex}\\[1.8ex]
			%		%
			$\ASSIGN{\vara}{\terma}$ & $\substitute(\xadd, \vara, \terma)$   \\[1.5ex]
			%		%
			$\TICK{\tickrew}$ & $\apply(+, \obtain(\tickrew), \xadd)$   \\[1.5ex]
			$\OBSERVE{\forma}$ &  $\apply(\switchITEsymbol, \xadd_\forma, \xadd, \obtain(0))$ \\[1.5ex]
			%		%
			$\COMPOSE{\cc_1}{\cc_2}$ & $\xwp{\cc_1}{\xwp{\cc_2}{\xadd}}$  \\[1.5ex]
			%		%
			$\NDCHOICE{\cc_1}{\cc_2}$ & $\minfunc(\xwp{\cc_1}{\xadd}, \xwp{\cc_2}{\xadd} )$ \\[1.5ex]
			\multirow{2}{*}{$\PCHOICE{\cc_1}{\proba}{\cc_2}$} & \multirow{2}{*}
			{\shortstack[l]{
				$\apply(+, \apply(\cdot, \obtain(\proba), \xwp{\cc_1}{\xadd}),$ \\
				 $\qquad\qquad\qquad\apply(\cdot, \obtain(1-\proba), \xwp{\cc_2}{\xadd}))$
			}
			}
			%
			%
			%$\mylambda{\vala} \proba \cdot \wp{\cc_1}{\FF}(\vala) + (1-\proba )\cdot \wp{\cc_2}{\FF}(\vala)$ 
			\\[4.2ex]
			$\ITE{\forma}{\cc_1}{\cc_2}$ & 
			$\apply(\switchITEsymbol, \xadd_\forma, \xwp{\cc_1}{\xadd}, \xwp{\cc_2}{\xadd})$   
			%$\lfp \FG.\, \iverson{\forma}\cdot \wp{\cc'}{\FG} + \iverson{\neg\forma}\cdot\FF$ 
			\\
			\bottomrule
			\bottomrule
		\end{tabularx}
	\end{center}%
	\caption{A syntactic variant of the weakest pre-expectation calculus for loop-free programs operating on $\sorteureal$-TEDDs.
		Here $\xadd_\forma$ is the $\sortbool$-TEDD obtained from the quantifier-free formula $\forma$ by recursing on the structure of $\forma$ using \ApplyAlg from \Cref{alg:apply} on the Boolean connectives. 
		% \ObtainAlg with argument $\tickrew$ (resp.\ $\proba$, $1-\proba$) returns the $1$-node $\sorteureal$-TEDD labelled with $\tickrew$ (resp.\ $\proba$, $1-\proba$) (cf.\ \Cref{alg:obtain}).
		}%
	\label{tab:xwp}%
\end{table}%

This section marries the WP calculus from \Cref{sec:wp} with the TEDD-based representation of, e.g., expectations, from \Cref{sec:xadds}. By careful construction, the integration is rather natural. 

To facilitate the marriage, we introduce $\xwpsymbol$, a syntactic variant of $\wpsymbolnostruct$ for loop-free programs that operates on $\sorteureal$-TEDDs. 
First recall that every $\sorteureal$-TEDD $\xadd$ represents a class of expectations, i.e., given a structure $\struct$, we have that
$
\mylambda{\vala} \sem{\xadd}{(\struct,\vala)} \in  \Es
$ is an expectation. This perspective enables defining $\xwpsymbol$ recursively on the structure of program $\cc$, using the previously developed operations \ApplyAlg, \ObtainAlg, and \SubstituteAlg, as well as the \MinimumAlg operation introduced below. Given an $\sorteureal$-TEDD $\xadd$, $\xwp{\cc}{\xadd}$ is an $\sorteureal$-TEDD as defined in \Cref{tab:xwp}. The definition constructs this TEDD algorithmically. By construction, the approach is sound: applying $\xwpsymbol$ to $\xadd$ is equivalent to applying $\wpsymbol$ to the expectation represented by $\xadd$ for all structures $\struct$:
%
% Before we detail $\xwpsymbol$, consider its key property:
%
\begin{theorem}
	\label{thm:xwp_sound}
	For every loop-free program $\cc$ and every $\sorteureal$-TEDD $\xadd$, we have 
	\[
	\forall~\text{structures $\struct$}\colon\quad \mylambda{\vala} \sem{\xwp{\cc}{\xadd}}{(\struct, \vala)} \eeq \wp{\cc}{\mylambda{\vala} \sem{\xadd}{(\struct,\vala)}}~.
	\]
\end{theorem}
\begin{proof}
See \Cref{app:proof_xwp_sound}.
\end{proof}
%We note that the soundness is not affected by any pruning. We also note that (only) the use of nondeterministic programs requires the minimum operator.
%
%In words: applying $\xwpsymbol$ to $\xadd$ is equivalent to applying $\wpsymbol$ to the expectation represented by $\xadd$ for all structures $\struct$. %Using the algorithms $\apply$, $\obtain$, and $\substitute$ enables defining $\xwpsymbol$ for all constructs 
%but non-deterministic choice. 
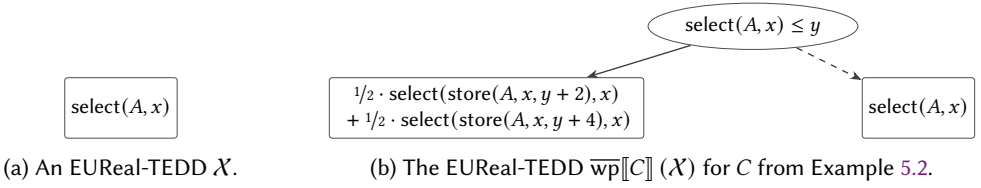
\begin{figure}[t]
\begin{center}
\vspace{-0.5em}
\begin{subfigure}[t]{.3\textwidth}
	\centering
	\scalebox{0.8}{
		\begin{tikzpicture}[node distance=6mm and 3mm]
			% Leaves
			
%			\node[bdd node] (x0) {$\funcselect(A, \vara) \leq \varb$};
%			\node[bdd leaf, below left = of x0,align=center] (y) {$\nicefrac{1}{2}\cdot \funcselect(\funcstore(A,\vara,\varb+2), \vara)$ \\
%				${}+ \nicefrac{1}{2}\cdot \funcselect(\funcstore(A,\vara,\varb+4), \vara)$};
%			\draw[bdd edge1] (x0) -- (y)  {};
			
			\node[bdd leaf,align=center, minimum height=1cm] (y2) {$\funcselect(A,\vara)$};
			%\draw[bdd edge0] (x0) -- (y2)  {};

			%	\node[bdd node, below right = of x0] (x1) {$\vara=1$};
			%	\draw[bdd edge0] (x0) -- (x1)  {};
			%	\node[bdd leaf, below left = of x1] (y1) {$\varb+1$};
			%	\draw[bdd edge1] (x1) -- (y1)  {};
			%	
			%	\node[bdd node, below right = of x1] (x2) {$\vara=2$};
			%	\draw[bdd edge0] (x1) -- (x2)  {};
			%	\node[bdd leaf, below left = of x2] (y2) {$\varb+2$};
			%	\node[bdd leaf, below right = of x2] (y3) {$\varb+3$};
			%	\draw[bdd edge1] (x2) -- (y2)  {};
			%	\draw[bdd edge0] (x2) -- (y3)  {};
			
			%\node[bdd leaf, below left=of x0] (y) {$\varb$};
			
		\end{tikzpicture}
	}
	\caption{An $\sorteureal$-TEDD $\xadd$.}
	\end{subfigure}
\hfill
\begin{subfigure}[t]{.69\textwidth}
	\centering
	\scalebox{0.8}{
\begin{tikzpicture}[node distance=6mm and 3mm]
	% Leaves
	
	\node[bdd node] (x0) {$\funcselect(A, \vara) \leq \varb$};
	\node[bdd leaf, below left =0.8cm of x0,align=center, minimum height=1cm,minimum width=5.2cm] (y) {$\nicefrac{1}{2}\cdot \funcselect(\funcstore(A,\vara,\varb+2), \vara)$ \\
	${}+ \nicefrac{1}{2}\cdot \funcselect(\funcstore(A,\vara,\varb+4), \vara)$};
	\draw[bdd edge1] (x0) -- (y)  {};
	
	\node[bdd leaf, below right =0.8cm  of x0,align=center, minimum height=1cm] (y2) {$\funcselect(A,\vara)$};
	\draw[bdd edge0] (x0) -- (y2)  {};

%	\node[bdd node, below right = of x0] (x1) {$\vara=1$};
%	\draw[bdd edge0] (x0) -- (x1)  {};
%	\node[bdd leaf, below left = of x1] (y1) {$\varb+1$};
%	\draw[bdd edge1] (x1) -- (y1)  {};
%	
%	\node[bdd node, below right = of x1] (x2) {$\vara=2$};
%	\draw[bdd edge0] (x1) -- (x2)  {};
%	\node[bdd leaf, below left = of x2] (y2) {$\varb+2$};
%	\node[bdd leaf, below right = of x2] (y3) {$\varb+3$};
%	\draw[bdd edge1] (x2) -- (y2)  {};
%	\draw[bdd edge0] (x2) -- (y3)  {};
	
	%\node[bdd leaf, below left=of x0] (y) {$\varb$};
	
\end{tikzpicture}
}
\caption{The $\sorteureal$-TEDD $\xwp{\cc}{\xadd}$ for $\cc$ from \Cref{ex:wp:array}.}
\end{subfigure}
\end{center}
\caption{Exemplification of $\xwpsymbol$ by applying it to the setting from \Cref{ex:wp:array}.}
\end{figure}

\begin{algorithm}[t]
	\caption{Computing Pointwise Minima of $\sorteureal$-TEDDs}
	\label{alg:minimum}
	%\SetKwFunction{unsat}{UNSAT}
	%\setcounter{AlgoLine}{0}
	%\newcommand{\rep}{\text{result}^{+}}
	%\newcommand{\rem}{\text{result}^{-}}
	\Fn{\nllabel{alg:min:def}\minfunc{$\xadd_1$, $\xadd_2$}}{%
		\KwIn{%
			$\sorteureal$-TEDDs $\xadd_1$ and $\xadd_2$.
		}
		\KwResult{%
			An $\sorteureal$-TEDD $\xaddb$ satisfying $\forall ~\text{structures $\struct$}\colon \mylambda{\vala} \sem{\xaddb}{(\struct,\vala)} = \mylambda{\vala} \min \{ \sem{\xadd_1}{(\struct,\vala)}, \sem{\xadd_2}{(\struct,\vala)} \}$. %, and\newline
%			and if $\hist' \in \hists[\xdb]$ then
%			$\hist \wedge \hist'$ is SAT in $\theo$.
		}
		\BlankLine
%		%
		\If{\nllabel{alg:min:cache_lookup}there is a cache entry $\xaddb$ for $(\xadd_1, \xadd_2)$}{%
			\KwRet $\xaddb$\;\nllabel{alg:min:cache_return}
		}\nllabel{alg:min:cache_end}
	    \eIf{\nllabel{alg:min:terminal_test}\label{alg:min-previous:base_case1}$\xadd_1$ and $\xadd_2$ are terminal}{%
		$\xaddb \leftarrow \obtain(\xlabt(\xadd_1) < \xlabt(\xadd_2), \xadd_1, \xadd_2)$\;\nllabel{alg:min:terminal_obtain}\label{alg:min-previous:base_case2}
		}{\nllabel{alg:min:recursive_case}\label{alg:min-previous:recur_case1}
		\uIf{\nllabel{alg:min:first_case}$\xadd_2$ is terminal or $\xlabn(\xadd_1) \varord \xlabn(\xadd_2)$}{%
			$\forma \leftarrow \xlabn(\xadd_1)$;\quad
			$\xaddb^{+} \leftarrow \minfunc(\succtxadd{\xadd_1}, {\xadd_2})$;\quad
			$\xaddb^{-} \leftarrow \minfunc(\succfxadd{\xadd_1}, {\xadd_2})$\nllabel{alg:min:first_branch}
		}\uElseIf{\nllabel{alg:min:second_case}$\xadd_1$ is terminal or $\xlabn(\xadd_2) \varord \xlabn(\xadd_1)$}{%
			$\forma \leftarrow \xlabn(\xadd_2)$;\quad
			$\xaddb^{+} \leftarrow \minfunc(\xadd_1, \succtxadd{\xadd_2})$;\quad 
			$\xaddb^{-} \leftarrow \minfunc(\xadd_1, \succfxadd{\xadd_2})$\nllabel{alg:min:second_branch}
		}\uElse{\nllabel{alg:min:equal_case}
			$\forma \leftarrow \xlabn(\xadd_1)$;\quad
			$\xaddb^{+} \leftarrow \minfunc(\succtxadd{\xadd_1}, \succtxadd{\xadd_2})$;\quad 
			$\xaddb^{-} \leftarrow \minfunc(\succfxadd{\xadd_1}, \succfxadd{\xadd_2})$\nllabel{alg:min:equal_branch}
		}
		$\xaddb \leftarrow \apply(\switchITEsymbol, \obtain(\forma), \succtxadd{\xaddb}, \succfxadd{\xaddb})$\;\nllabel{alg:min:merge}
	}\nllabel{alg:min:recursive_end}\label{alg:min-previous:recur_case2}
	Store $\xaddb$ in the cache for $(\xadd_1, \xadd_2)$\;\nllabel{alg:min:cache_store}
\KwRet $\xaddb$\;\nllabel{alg:min:return}\label{alg:min-previous:return}
	}
\end{algorithm}

The non-deterministic choice requires computing pointwise minima of expectations. Since standard arithmetical theories such as $\theoryla$ do not provide a function symbol for minima, we instead provide a dedicated algorithm for computing pointwise minima of two $\sorteureal$-TEDDs in \Cref{alg:minimum}. Given $\sorteureal$-TEDDs $\xadd_1,\xadd_2$, $\MinimumAlg(\xadd_1,\xadd_2)$ returns an $\sorteureal$-TEDD representing their point-wise minimum. Similarly to \ApplyAlg (\Cref{alg:apply}), \Cref{alg:minimum} traverses both $\xadd_1$ and $\xadd_2$ in a $\varord$-preserving manner (ll.\ \ref{alg:min-previous:recur_case1}-\ref{alg:min:recursive_end}). The base case in l.\ \ref{alg:min:terminal_obtain} creates a $3$-node TEDD representing the pointwise minimum of the respective terminal labels, using that $\min{\reala,\reala'}$ euquals $\reala$ if $\reala<\reala'$, and $\reala'$ otherwise. To ensure that the resulting TEDD respects $\varord$, we invoke \ApplyAlg in l.\ \ref{alg:min:merge}, which soundly merges the TEDDs obtained from the recursive calls. We get:
\begin{theorem}
\label{thm:minimum_soundness_main}
	\Cref{alg:minimum} is sound and encounters at most $\sizeof{\xadd_1}\cdot\sizeof{\xadd_2}$ cache misses in \Cref{alg:min:cache_lookup}.
\end{theorem}
\begin{proof}
	See \Cref{app:proof_minimum}.
\end{proof}

\section{Deductive Verification of Probabilistic Loops via TEDDs and SMT Solving}
In this section, we lift proof rules for WP-based reasoning to the TEDD-based representation. 
% We proceed in three steps: In \Cref{sec:xwp}, we provide a syntactic variant of $\wpsymbolnostruct$ for \emph{loop-free} programs operating on $\sorteureal$-XADDs. 
In \Cref{sec:wp:loops}, we recap both fixed-point iteration- and induction-based proof rules from the literature, which reduce reasoning about loops to reasoning about loop-free programs. The sought-after techniques are then obtained in \Cref{sec:proof_rules_via_tedds} by combining the results from \Cref{sec:xwp,sec:wp:loops}.

\label{sec:verification}

\subsection{Proof Rules for Reasoning about Loops}
\label{sec:wp:loops}
Whereas determining weakest pre-expectations of loops reduces to \emph{syntactic} reasoning, reasoning about loops is much more challenging. Therefore, we employ \emph{proof rules} for loops. %which, in fact, reduce reasoning about loops to reasoning about loop-free programs. 
Our proof rules target verifying or refuting whether for given loop $\cc$ and $\FF,\FG\in\Es$, we have
\begin{align*}
	\label{eqn:wp:upper_bound}
	\wp{\cc}{\FF} \eeleq \FG~,
	\tag{$\dagger$}
\end{align*}
i.e., whether or not $\FG(\vala)$ soundly upper-bounds the weakest pre-expectation of $\cc$ w.r.t.\ $\FF$ on initial state $\vala$ \emph{for all} $\vala$. We leverage two classes of proof rules for tackling this problem: \emph{fixed-point iteration}-based rules and state-of-the-art \emph{induction-based} rules, which we introduce next. 
Leveraging TEDDs for related proof-rule targeting probabilistic termination, randomized distributed algorithms, and supermartingale principles~\cite{DBLP:journals/pacmpl/MajumdarS25,DBLP:journals/pacmpl/MajumdarS24,DBLP:journals/pacmpl/EneaMMS26,DBLP:journals/corr/abs-2203-04422,DBLP:journals/corr/abs-2512-00270,10.1145/3808348}, is left future work.\kb{check}
\begin{example}
	\label{ex:wp_upper_bound_ motivate}
	Consider the post-expectation $\FF = \mylambda{\vala} 0$, the loop $\cc$ given by
	\[
	\WHILEDO{\varb \foeq 0}{ \PCHOICE{\ASSIGN{\varb}{1}}{\nicefrac{1}{2}}{\TICK{1}}}~,
	\]
	where $\hastype{\varb}{\sortnat}$,
	and let $\FG =\mylambda{\vala} 1$. The loop $\cc$ keeps flipping a fair coin. If it lands heads, the loop terminates. Otherwise, a cost of $1$ is incurred. By proving that $\wp{\cc}{\FF} \eleq \FG$, we verify that, for every initial state, the expected cost  when executing $\cc$ is upper-bounded by $1$. 
	%The structure $\struct$ is irrelevant, the  variable is of type $\sortnat$, which is always assigned its standard meaning.
	%
	\hfill $\triangle$
\end{example}

%\begin{remark}
%	\sj{I didnt understand this remark here?}
%	As is standard in SMT-based program verification, we will, in the subsequent sections, work in a more axiomatic setting: Rather than fixing a structure $\struct$, we consider a theory $\theory$ induced by a set of axioms and verify or refute that (\ref{eqn:wp:upper_bound}) holds for all $\theory$-structures $\struct$.  In this section, however, we work in a setting where $\struct$ is fixed, which will tightly integrate to the axiomatic setting. \hfill $\triangle$
%\end{remark}

\paragraph{Fixed-Point Iteration-Based Reasoning}
This technique is based on \emph{loop unrolling} and best understood by recalling that the $n$-th iterate $\wcharfuniter{\FF}{n}(\expzero)$ of a loop $\cc$'s characteristic function (w.r.t.\ post-expectation $\FF$) corresponds to the weakest pre-expectation of $\cc$ w.r.t.\ $\FF$ when aborting after at most $n$ iterations. With this in mind, consider the following: 

% given loop $\cc$, post-expectation $\FF$, and $n \in \Nats$, $\wcharfuniter{\FF}{n}(0)$ is the weakest pre-expectation of the $n$-th unrolling of $\cc$, i.e., the loop that behaves like $\cc$ but aborts after at most $n$ iterations. Based on this understanding, consider the following:
%
\begin{theorem}[\cite{kind_cav}]
	\label{thm:proof_rules:iter}
	Let $\cc = \WHILEDO{\forma}{\cc'}$, $\FF,\FG \in \Es$, and $n\in\Nats$. We have :
	\begin{enumerate}
		\item\label{thm:proof_rules:iter:verify} If $\wcharfuniter{\FF}{n+1}(\expzero) \eleq \wcharfuniter{\FF}{n}(\expzero)$ and $\wcharfuniter{\FF}{n}(\expzero) \eleq \FG$, then $\wp{\cc}{\FF} \eleq \FG$. 
		\item\label{thm:proof_rules:iter:refute} If $\wcharfuniter{\FF}{n}(\expzero) \not\eleq \FG$, then $\wp{\cc}{\FF} \not\eleq \FG$. 
	\end{enumerate}
\end{theorem}
\Cref{thm:proof_rules:iter}.\ref{thm:proof_rules:iter:verify} is an immediate consequence of the fact $\wcharfuniter{\FF}{n+1}(\expzero) \eleq \wcharfuniter{\FF}{n}(\expzero)$ implies $\wcharfuniter{\FF}{n}(\expzero)  = \lfp \wcharfun{\FF}= \wp{\cc}{\FF}$ and thus suitable for \emph{verifying} that (\ref{eqn:wp:upper_bound}) holds. \Cref{thm:proof_rules:iter}.\ref{thm:proof_rules:iter:refute}, on the other, is for \emph{refuting} (\ref{eqn:wp:upper_bound}) and an instance of what \citet{kind_cav} call \emph{latticed bounded model checking}:  Kleene's fixed-point theorem  implies that if iterating $\wcharfun{\FF}$ on $0$ exceeds $\FG$, then \mbox{$\wp{\cc}{\FF}$ exceeds $\FG$, as well.}

\paragraph{Induction-Based Reasoning} 
Induction-based reasoning is suitable for \emph{verifying} (\ref{eqn:wp:upper_bound}). We leverage a probabilistic generalization \cite{kind_cav} of the $k$-induction verification technique~\cite{DBLP:conf/fmcad/SheeranSS00,sofware_k_induction,boosting_k_induction,property_directed_k_induction,k_induction_without_unrolling}: Given expectations $\FF,\FG$ and a loop $\cc$ with  characteristic function $\wcharfun{\FF}$, the \emph{$k$-induction operator $\kindop{\FF}{\FG} \colon \Es \to \Es$ for $\cc$ w.r.t.\ $\FF$ and $\FG$} is
\(
\kindops(\FH) = \wcharfun{\FF}(\FH) \einf \FG.
\)
%
%This\sj{this what?} yields a $k$-induction principle:
%
\begin{theorem}[\cite{kind_cav}]
	\label{thm:proof_rules:kind}
	Let $\cc = \WHILEDO{\forma}{\cc'}$, let $\FF,\FG \in \Es$, and $k\in\Nats$.  We have:
	\[
	\wcharfun{\FF}\big( \kindopsiter{k}(\FG)\big) \eeleq \FG
	\qquad\text{implies}\qquad
	\wp{\cc}{\FF} \eeleq \FG.
	\]
	%	
	%	If $\wcharfun{\FF}\big( \kindopsiter{n}(\FG)\big) \eleq \FG$, then $\wp{\cc}{\FF} \eleq \FG$.
	%
\end{theorem}
If the above premise holds, then we say that $\FG$ is $(k+1)$-inductive w.r.t.\ $\cc$ and $\FF$. \Cref{thm:proof_rules:kind} is thus an infinite class of proof rules\sj{You could say the same for Thm 5.3 above, right?}\kb{which one}\sj{I think this is a very old comment of mine, I think the point was this was not the first set of infinite proof rules we introduce.} of increasing strength, i.e, we have a $k$-induction rule for every $k\in\Nats$ and for all such $k$, $k$-inductivity implies $k+1$-inductivity but not vice versa.
\begin{example}
	\label{ex:kind}
	%	Consider the post-expectation $\FF = \mylambda{\vala} 0$ and the loop $\cc$ given by
	%	%
	%	\[
	%		\WHILEDO{\varb \foeq 0}{ \PCHOICE{\ASSIGN{\varb}{1}}{\nicefrac{1}{2}}{\TICK{1}}}~.
	%	\]
	%	%
	%	Now let $\FG =\mylambda{\vala} 1$ and suppose we wish to prove that $\wp{\cc}{\FF} \eleq \FG$, i.e., the expected cost incurred when executing the \emph{unbounded} loop $\cc$ on any initial state is upper-bounded by $1$. 
	Reconsider \Cref{ex:wp_upper_bound_ motivate}.
	We employ \Cref{thm:proof_rules:kind} in an iterative fashion: $\FG$ is \emph{not} $1$-inductive since $\wcharfun{\FF}(\FG) \not\eleq \FG$. However, $\FG$ is $2$-inductive, 
	%i.e., we have $\wcharfun{\FF}(\kindop{\FF}{\FG}(\FG)) = \wcharfun{\FF}(\wcharfun{\FF}(\FG) \einf \FG) \eleq \FG$,
	 so we conclude $\wp{\cc}{\FF} \eleq \FG$.
	\hfill $\triangle$
\end{example}
Notice that both \Cref{thm:proof_rules:iter,thm:proof_rules:kind} only require us to (i) compute $\wpsymbolnostruct$'s of the loop \emph{body} and (ii) checking %\enquote{quantitative entailments} between expectations, i.e.,
 whether a given expectation is state-wise upper-bounded by another expectation. %In this paper, we realize these two tasks by marrying the power of TEDD-based reasoning with SMT solving.
\begin{remark}
	As expected, due to undecidability of $\wp{\cc}{\FF} \eleq \FG$ \cite{hardness1}, the  proof rules are incomplete: Regarding \Cref{thm:proof_rules:iter}.\ref{thm:proof_rules:iter:verify}, there are loops for which no suitable $n$ exists, which is, e.g., the case for the loop from \Cref{ex:kind}. Regarding \Cref{thm:proof_rules:kind}, there are $\cc,\FF,\FG$ such that $\FG$ is not $k$-inductive for any $k$ (cf.\ \cite[Example 2]{kind_cav}). \Cref{thm:proof_rules:iter}.\ref{thm:proof_rules:iter:refute} is, however, \emph{refutation complete} \cite[Corollary 2]{kind_cav}: If $\wp{\cc}{\FF} \not\eleq \FG$, the existence of a suitable \mbox{$n$ witnessing this is guaranteed.} \hfill $\triangle$
\end{remark}

\subsection{Implementing Proof Rules for Loops via TEDDs and SMT Solving}
\label{sec:proof_rules_via_tedds}
With the preceding concepts at hand, we tackle a variation of $\dagger$ above:
\begin{center}
	Given a theory $\theory$, a loop $\cc$, $\sorteureal$-TEDDs $\xadd,\xaddb$, decide \\
	$\forall ~\text{$\theory$-structures $\struct$}\colon \wp{\cc}{\mylambda{\vala}  \sem{\xadd}{(\struct,\vala)}} \eleq \mylambda{\vala} \sem{\xaddb}{(\struct,\vala)}$~?
	%
%	hold?
\end{center}
If the above holds, then we say that \emph{$\xaddb$ upper-bounds $\cc$ w.r.t.\ $\xadd$ modulo $\theory$}. As is standard for many \cite{kind_cav,DBLP:conf/ijcai/BaoTPH023,DBLP:conf/tacas/BatzCJKKM23,DBLP:conf/qest/GretzKM13,DBLP:conf/sas/KatoenMMM10,DBLP:conf/cav/SunFCG23} but not all \cite{DBLP:journals/pacmpl/ChatterjeeGMZ24,DBLP:conf/pldi/WangS0CG21,DBLP:conf/pldi/Wang0GCQS19,DBLP:journals/pacmpl/AvanziniMS23,DBLP:journals/pacmpl/AvanziniMS20} automated probabilistic program verification techniques, we restrict to non-nested loops. Reasoning about nested loops is currently realized by simultaneous loop unrolling~\cite[Sec.~7.2 and App.~B.2]{DBLP:journals/corr/abs-2502-19388} for fixed-point iteration-based reasoning. For induction-based reasoning, nested loops appear rarely in existing benchmark suites and are re-written to a single loop with loop-free body (for now, manually, as has been done in \cite{kind_cav,DBLP:journals/pacmpl/SchroerBKKM23}).  \kb{TODO:discuss}

% Reasoning about nested loops can be realized by simultaneous loop unrolling~\cite[Sec.~7.2 and App.~B.2]{DBLP:journals/corr/abs-2502-19388} for fixed-point iteration-based reasoning, as recalled in \Cref{app:simultaneous_unrolling}; for induction-based reasoning, nested loops are re-written to a single loop with loop-free body. This is standard for many automatic SMT-based approaches \cite{kind_cav,DBLP:conf/ijcai/BaoTPH023,DBLP:conf/tacas/BatzCJKKM23,DBLP:conf/qest/GretzKM13,DBLP:conf/sas/KatoenMMM10} KATOEN,GRETZ

% exception 
% - 2023. Automated Expected Value Analysis of Recursive Programs. Proc. ACM Program. Lang. BibTeX: `DBLP:journals/pacmpl/AvanziniMS23` (added).

%
\begin{example}
	Let $\theory = \theorylaarr$, let $k\in\Nats$, and consider program $\cc$, and TEDDs $\xadd$, $\xaddb$: % given by
	\[
		\cc \quad=\quad \WHILEDO{\vara \leq k}{\quad
			\COMPOSE{
			\PCHOICE{\ASSIGN{\varb}{\varb +\ARRAYREAD{A}{\vara}}}{\nicefrac{1}{2}}{\SKIP}
		}{~
			\ASSIGN{\vara}{\vara+1}
		}\quad
			}~,
	\]
	%
	%Now suppose $\xadd$ and $\xaddb$ are $\sorteureal$-TEDDs representing the case expressions
	%
	\[
		\xswitchfun{\xadd}\eeq
		\begin{cases}
			\varb \switchfuncase\true,
		\end{cases}
		%
		%
		%\quad\text{and}\qquad
		%
		%
		\xswitchfun{\xaddb}\eeq
		\begin{cases}
			\varb + \nicefrac{1}{2}  \cdot(\ARRAYREAD{A}{0} +\ldots +\ARRAYREAD{A}{k})   & \switchfuncase \vara \foeq 0 \wedge \vara \leq k \\
			\vdots \\
			\varb + \nicefrac{1}{2} \cdot  \ARRAYREAD{A}{k}   &\switchfuncase \vara\foneq 0 \wedge \ldots \wedge \vara\foneq k-1\wedge \vara \foeq k \wedge \vara \leq k \\
			\varb & \switchfuncase \vara > k~.
		\end{cases}
	\]
	By verifying that $\xaddb$ upper-bounds $\cc$ w.r.t.\ $\xadd$ modulo $\theorylaarr$, we prove  that we expect $\cc$ to increment $\varb$ by at most $\nicefrac{1}{2}$ times the sum of the values $\ARRAYREAD{A}{x}, \ARRAYREAD{A}{x+1},\ldots, \ARRAYREAD{A}{k}$. %whenever $\vara \leq k$ holds initially, and that $\varb$ is not modified otherwise.
	\hfill $\triangle$
\end{example}
%\subsubsection{Step 1: A Syntactic Weakest Pre-Expectation Transformer Operating on TEDDs}
%

%The above theorem straightforwardly generalizes to computing XADDs representing $\wcharfuniter{\FF}{n}(0)$ or $\wcharfun{\FF}\big( \kindopsiter{n}(\FG)\big)$ whenever both $\FF$ and $\FG$ are represented as XADDs.
%
%Moreover, for any given theory $\theory$, it is clear that we can prune $\xwp{\cc}{\xadd}$ and obtain an equivalent XADD modulo $\theory$.

%\subsubsection{Step 2: Deciding Quantitative Entailments for ${\sorteureal}$-TEDDs}
%
Now recall that, given a loop $\cc = \WHILEDO{\forma}{\cc'}$ and a post-expectation $\FF,\FG$, $\wcharfun{\FF}$ denotes the characteristic function of $\cc$ w.r.t.\ $\FF$ and  $\kindops(\FH)$ denotes the $k$-induction operator for $\cc$ w.r.t.\ $\FF$ and $\FG$. Using $\xwpsymbol$, we define algorithmic variants of these functionals operating on TEDDs\footnote{Here $\xadd_\forma$ is the atomic $\sortbool$-TEDD obtained from the quantifier-free formula $\forma$ by recursing on  $\forma$ using \ApplyAlg.}:
\begin{align*}
	\xwcharfun{\xadd}(\xaddc) 
	\eeq&
	\apply(\switchITEsymbol, \xadd_\forma, \xwp{\cc'}{\xaddc}, \xadd)  \qquad\text{and}\qquad
	\xkindops(\xaddc) \eeq  \minfunc\big( \xwcharfun{\xadd}(\xaddc), \xaddb \big)~.
\end{align*}
This also enables us to compute $n$-th unfoldings $\xwcharfuniter{\xadd}{n}(\expzero)$ (where $\expzero$ is the $1$-node $\sorteureal$-TEDD labeled by $0$) and $\xkindopsiter{n}(\xaddb)$, as required by \Cref{thm:proof_rules:iter,thm:proof_rules:kind}, respectively.

\begin{algorithm}[t]
	\caption{Checking Quantitative Entailments for $\sorteureal$-TEDDs Modulo a Theory}
	\label{alg:entails}
	%\SetKwFunction{unsat}{UNSAT}
	%\setcounter{AlgoLine}{0}
	%\newcommand{\rep}{\text{result}^{+}}
	%\newcommand{\rem}{\text{result}^{-}}
	\Fn{\nllabel{alg:entails:def}\entails{$\theory$, $\xadd_1$, $\xadd_2$, $\forma$}}{%
		\KwIn{%
			A theory $\theory$, $\sorteureal$-TEDDs $\xadd_1,\xadd_2$, 
			formula $\forma\in\fo$ (initially $\true$).
		}
		\KwResult{%
			$\true$ iff $\forall ~\text{$\theory$-structures $\struct$}\colon \mylambda{\vala} \sem{\xadd_1}{(\struct,\vala)} \eleq \mylambda{\vala} \sem{\xadd_2}{(\struct,\vala)}$ (for $\forma=\true$).%, and\newline
		}
		\BlankLine
		\If{\nllabel{alg:entails:terminal_test}$\xadd_1$ is terminal and $\xadd_2$ is terminal}{%
			$\formb \leftarrow \forma \wedge \xlabt(\xadd_1) > \xlabt(\xadd_2)$\;\nllabel{alg:entails:counterexample_formula}
			\If{\nllabel{alg:entails:cache_lookup}there is a cache entry ``$\text{result}$'' for $\formb$}{
				\KwRet{$\text{result}$}\;\nllabel{alg:entails:cache_return}
			}\nllabel{alg:entails:cache_end}
			$\text{result} \leftarrow \neg\sats{\formb}$\;\nllabel{alg:entails:smt_check}
			Store ``$\text{result}$'' in the cache for $\formb$\;\nllabel{alg:entails:cache_store}
			\KwRet{$\text{result}$}\;\nllabel{alg:entails:base}
		}\nllabel{alg:entails:terminal_end}
		$\formb \leftarrow \min_{\varord} \{\xlabn(\xadd_i) ~|~ \xadd_i~\text{is non-terminal and}~i \in \{1,2\} \}$\;\nllabel{alg:entails:top_label} 
		$\forma^+ \leftarrow \forma\wedge \formb$; \quad $\forma^- \leftarrow \forma\wedge \neg\formb$\;\nllabel{alg:entails:contexts}
		\ForEach{\nllabel{alg:entails:traverse_start}$\xadd_i$ \In $\xadd_1, \xadd_2$}{%
			\eIf{\nllabel{alg:entails:matches_top}$\xadd_i$ is non-terminal and $\xlabn(\xadd_i)=\formb$}{%
				$\xaddb_i \leftarrow \succtxadd{\xadd_{i}}$;\quad
				$\xaddc_i \leftarrow \succfxadd{\xadd_{i}}$\;\nllabel{alg:entails:take_successors}
			}{\nllabel{alg:entails:does_not_match_top}
				$\xaddb_i \leftarrow \xadd_{i}\,$;\quad
				$\xaddc_i \leftarrow \xadd_{i}$\;\nllabel{alg:entails:keep_inputs}
			}\nllabel{alg:entails:traverse_case_end}
		}\nllabel{alg:entails:traverse_end}
		\KwRet{$\entails(\theory, \xaddb_1, \xaddb_2, \forma^+) \wedge \entails(\theory, \xaddc_1, \xaddc_2, \forma^-)$};\nllabel{alg:entails:recursive_return}
	}
\end{algorithm}

Given these algorithmic operators, what remains to be done is to check \enquote{quantitative $\eleq$-entailments} between expectations, as in \Cref{alg:entails}, which assumes  an SMT solver for checking the satisfiability of $\fo$ formulae modulo $\theory$.
\begin{theorem}
\label{thm:entails_soundness_main}
Given a theory $\theory$ and $\sorteureal$-TEDDs $\xadd_1,\xadd_2$, we have
\[
    \entails{$\theory$, $\xadd_1$, $\xadd_2$} = \true
    \quad\text{iff}\quad
	\forall ~\text{$\theory$-structures $\struct$}\colon \mylambda{\vala} \sem{\xadd_1}{(\struct,\vala)} \eleq \mylambda{\vala} \sem{\xadd_2}{(\struct,\vala)}~,
\]
and \Cref{alg:entails} performs at most $2^{m_1+m_2+1}-1$ recursive invocations, where $m_1$ (resp. $m_2$) is the maximum number of edges of a directed path from $\xadd_1$'s (resp.\ $\xadd_2$'s) root to a terminal node.
\end{theorem}
\begin{proof}
See \Cref{app:proof_entails}.
\end{proof}
\Cref{alg:entails} traverses $\xadd_1$ and $\xadd_2$ in a $\varord$-preserving manner, essentially enumerating all pairs of paths from the roots to some terminal node while constructing the paired path condition $\forma$ (cf.\ \Cref{def:pruned_tedd}) in ll.\ \ref{alg:entails:top_label}-\ref{alg:entails:recursive_return}. When encountering terminal nodes (ll.\ \ref{alg:entails:counterexample_formula}-\ref{alg:entails:base}), we check whether a counter-example-to-entailment exists by checking whether $\formb \leftarrow \forma \wedge \xlabt(\xadd_1) > \xlabt(\xadd_2)$ is satisfiable modulo $\theory$ in l.\ \ref{alg:entails:smt_check}. If so, \entails returns $\false$. If such a counterexample is never encountered, then the entailment holds and \entails returns $\true$. Satisfiability results are cached.

Finally, this enables us to implement the proof rules from \Cref{thm:proof_rules:iter,thm:proof_rules:kind} as follows:	
\begin{theorem}
	\label{thm:proof_rules_xadds}
	Let $\cc$ be a loop, let $\xadd,\xaddb$ be $\sorteureal$-TEDDs, let $\theory$ be a theory, and $n\in\Nats$. Then:
	\begin{enumerate}
		\item If $\entails(\theory, \xwcharfun{\xadd}(\xkindopsiter{n}(\xaddb)), \xaddb)$, then $\xaddb$ upper-bounds $\cc$ w.r.t.\ $\xadd$ modulo $\theory$.
		\item If $\neg \entails(\theory, \xwcharfuniter{\xadd}{n}(\expzero), \xaddb)$, then $\xaddb$ does \underline{not} upper-bound $\cc$ w.r.t.\ $\xadd$ modulo $\theory$.
		\item If $\entails(\theory, \xwcharfuniter{\xadd}{n+1}(\expzero), \xwcharfuniter{\xadd}{n}(\expzero))$ and 
		$\entails(\theory,  \xwcharfuniter{\xadd}{n}(\expzero), \xaddb)$, then $\xaddb$ upper-bounds $\cc$ w.r.t.\ $\xadd$ modulo $\theory$.
	\end{enumerate}
\end{theorem}
%
%\begin{proof}
%	By (i) Thm.\ \ref{thm:xwp_sound} and soundness of $\entails$, and (ii) Thms.\ \mbox{\ref{thm:proof_rules:iter},\ref{thm:proof_rules:kind} on all $\theory$-structures $\struct$.}
%\end{proof}
%
%
\begin{remark}
	Since the checks in \Cref{thm:proof_rules_xadds} involve the theory $\theory$, we can invoke \PruneAlg from \Cref{alg:prune} on all TEDDs encountered during the involved computations while preserving the claims from \Cref{thm:proof_rules_xadds}. We discuss and evaluate our pruning strategies in \Cref{sec:implementation}.
	\hfill $\triangle$
\end{remark}

%\section{Experiments}
\section{Implementation}
\label{sec:implementation}
We have implemented a prototypical probabilistic program verifier, called \ddsolve{} (cf.\ the supplementary material), that consists of a highly efficient backend library and a flexible frontend.

\smallskip\noindent\emph{Backend library.}
We have implemented TEDDs on top of the state-of-the-art BDD library \sylvan{}~\cite{DBLP:journals/sttt/DijkP17} in C++. \sylvan{} supports multi-threaded BDD operations using a work-stealing framework called \lace{}~\cite{DBLP:conf/europar/DijkP14}. Operations on TEDDs, including pruning, are also running in parallel, relying on \lace{}. 
Pruning (\Cref{sec:xadds:pruning}) is implemented modularly and can be instantiated with different SMT solvers; the current implementation supports both Z3 and CVC5 and exploits the abilities of incremental SMT solving to alleviate the computational effort of pruning. As variable reordering is expensive in  \sylvan, we see the benefits of implementing substitution via ITEs (\Cref{sec:xadds:substitution}). 
The leafs and predicates in TEDDs can be represented in various ways. The library currently comes with two representations: 
(1)~Generic expressions using the data structures from the SMT solver CVC5.
(2)~Multivariate affine expressions with a dedicated list-like data structure and operations.  
Both data structures support rational numbers and represent them exactly. 

\smallskip\noindent\emph{Frontend.}
The frontend consists of a python-based parser for a simple probabilistic programming language, annotated with verification queries, resembling the front-end of  \caesar~\cite{DBLP:journals/pacmpl/SchroerBKKM23}. It supports multiple weakest preexpectation style calculi with minimizing and maximizing non-determinism including: weakest preexpectation (wp, \Cref{sec:wp}), weakest liberal preexpectation (wlp, \cite{mciver_morgan}), conditional weakest preexpectation (cwp,~\cite{DBLP:journals/toplas/OlmedoGJKKM18}), and expected runtime (ert,~\cite{DBLP:journals/jacm/KaminskiKMO18}).
The flexible frontend also supports parsing probabilistic programs from (a large fragment of) the \dice-language~\cite{DBLP:journals/pacmpl/HoltzenBM20}.

\smallskip\noindent\emph{Syntactical equivalence.}
We support theory-agnostic syntactical equivalence checks on TEDDs. Somewhat surprising, such theory-agnostic syntactical checks are adequate: In the experiments (Sec.~\ref{sec:experiments}), the syntactical checks never fail to identify two TEDDs that are semantically identical, while it is easy to construct syntactically different TEDDs expressing the same expectation.

\smallskip\noindent\emph{Normalization, pruning, caching, ordering.}
Data structures for terms normalize expressions, which may lead to trivial pruning (e.g., when a predicate $x > x+1$ is normalized to be \texttt{false}).  
By default, we apply SMT-based pruning after operations that (empirically) lead to potential for pruning, i.e., while sampling from uniform distributions, after every fixpoint or k-induction iteration, and after applying the k-induction operator. 
SMT-based pruning can also be disabled, see Sec.~\ref{sec:experiments:ablation}. 
While we have not discussed the use of an operation cache, we cache operations on TEDD nodes. 
Finally, the effect of variable orderings in BDDs is well-documented, but we currently use an ordering based on the occurrence in the AST of the program, which performs empirically quite well. 

\section{Experiments}
\label{sec:experiments}
We use \ddsolve{} (\Cref{sec:implementation}) to empirically evaluate the merits of TEDD-based deductive probabilistic program analysis. Therefore, we ask the following research questions:
\begin{enumerate}
    \item[\textbf{RQ1}] How does a (more explicit) TEDD-based representation of expectations compare with a (more implicit) SMT-based representation in the probabilistic program verifier \caesar~\cite{DBLP:journals/pacmpl/SchroerBKKM23}?
	\item[\textbf{RQ2}] How does an TEDD-based representation compare against an ADD-based representation in the finite-state probabilistic model checker \storm~\cite{DBLP:journals/sttt/HenselJKQV22}?
	\item[\textbf{RQ3}] What is the effect of respectively (a)~SMT-based exhaustive pruning, (b)~multi-threading support, and (c)~dedicated linear term data structures, during verification?
\end{enumerate}

\smallskip\noindent\emph{Setup.}
We run \ddsolve{} and the baselines discussed below in a container on a AMD Ryzen TRP 5965WX with 24 physical  (and identical) cores and 256 GB RAM. We sometimes limit cores by limiting the threads, as indicated below. Unless specified differently, we measure wall-clock time.

%\sjinline{Should we mention why we do not compare against explicit-state storm, dice, psi, and other tools?}

\subsection{Comparison with \caesar{} (RQ1)}
%\paragraph{Setup.}
To compare with SMT-based verification, we compare against (single-threaded) \caesar. 
 %\caesar{} always runs on a single thread. 
The verification tasks contain the two categories `fixpoint' and `induction', matching the proof rules in Section~\ref{sec:wp:loops}. The  benchmarks also cover ert- and wlp-calculi, see Sec.~\ref{sec:implementation}. %Throughout this section, we assume that $X$ and $Y$ are expectations, $C$ a program, and $n$ an index. 
\caesar{} and \ddsolve{} support both proof rules natively with similar input formats. Note that the approach to the fixpoint benchmarks is slightly different, see below. We use benchmarks from the \caesar{} verification suite as well as some new benchmarks, which we clarify below. %Below, we consider different types of programs. % Piecewise affine integer programs and with piecewise affine expectations, array programs, nonlinear integer programs and nondeterminsitic programs.

\begin{figure}
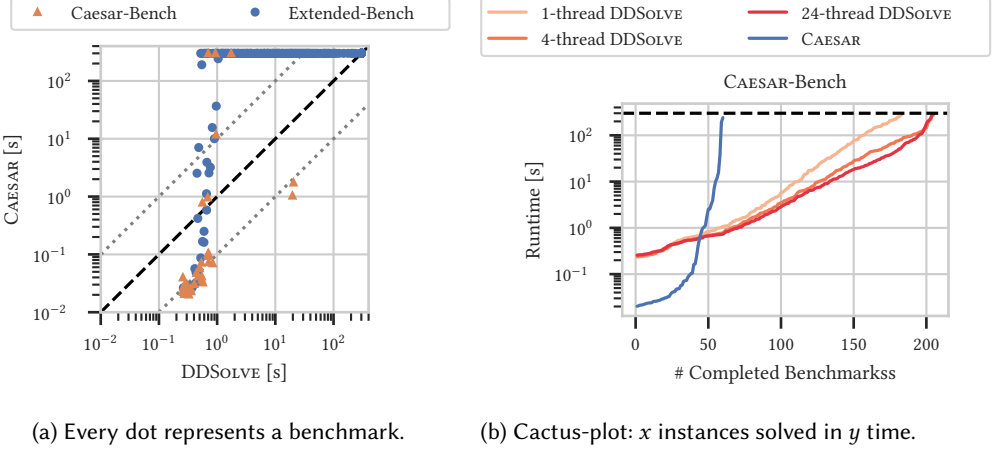

\begin{subfigure}[b]{0.45\textwidth}
	\centering
\input{figures/fig1.pgf}
\subcaption{Every dot represents a benchmark.}
\label{fig:results:linear:perf:caesar}
\end{subfigure}
\begin{subfigure}[b]{0.45\textwidth}
	\centering
\input{figures/log_cactusplot_complete.pgf}	
\subcaption{Cactus-plot: $x$ instances solved in $y$ time. }
\label{fig:results:linear:cactus:caesar}
\end{subfigure}
\vspace{-0.5em}
\caption{Performance comparison with \caesar{} on the complete piecewise affine benchmark suite.}
\label{fig:vscaesar}
\end{figure}
\subsubsection{Piecewise affine benchmarks.}
\label{sec:exp:rq1:pwa}
We study verification tasks where expectations $X, Y$ and the expressions in program $C$ are piecewise affine. These are the majority of existing benchmarks. In \ddsolve, we use the dedicated data structures for such expectations. We use benchmarks from \caesar{}~\cite{DBLP:journals/pacmpl/SchroerBKKM23} as well as parameterized benchmarks (where we vary, e.g., $n$), inspired by the benchmarks from~\cite{DBLP:journals/pacmpl/SchroerBKKM23,DBLP:conf/tacas/BatzCJKKM23}. The details of the benchmarks are given in \Cref{app:benchmark_programs}.

\smallskip\noindent\emph{Performance.}
In Fig.~\ref{fig:results:linear:perf:caesar}, we compare the run time of \caesar{} and \ddsolve{} (single-threaded). On the original benchmarks, \caesar{} solves instances within 0.1s, which is below the startup time of \ddsolve. However, on the eight remaining benchmarks, two are in favour of \caesar{}, while four are in favour of \ddsolve{}. The two instances where \caesar{} outperforms \ddsolve{} are induction verification tasks where a few pruning operations cause a massive slow down. 
On the scalable benchmarks, we see how \ddsolve{} clearly outperforms \caesar{} for larger instances. The effect is even more pronounced in the cactus plot in Fig.~\ref{fig:results:linear:cactus:caesar}. In that plot, we also see how using 4 or more threads allows \ddsolve{} to solve over twenty additional benchmarks.

\smallskip\noindent\emph{Scalability.}
To investigate scalability, we considered two benchmarks from \cite{DBLP:journals/pacmpl/SchroerBKKM23} with a scalable parameter. Rabin describes a mutual exclusion protocol with $n$ participants~\cite{DBLP:conf/podc/KushilevitzR92}, BRP describes network transmission of a file with $n$ packets~\cite{DBLP:conf/papm/DArgenioJJL01}. 
Essentially, the varying parameter influences the number of iterations that are necessary to compute the fixpoints. 
Fig.~\ref{fig:results:scale:caesar} top row shows that \caesar{} can only handle small problem instances, while \ddsolve{} scales significantly much better.

%\begin{figure}
%	\centering
%\input{figures/fig2.pgf}
%\vspace{-1em}
%\caption{Scalability comparison with \caesar{} on two benchmarks using piecewise linear expectations.}
%\label{fig:results:linear:scale:caesar}
%\end{figure}

\begin{figure}
	\centering
\input{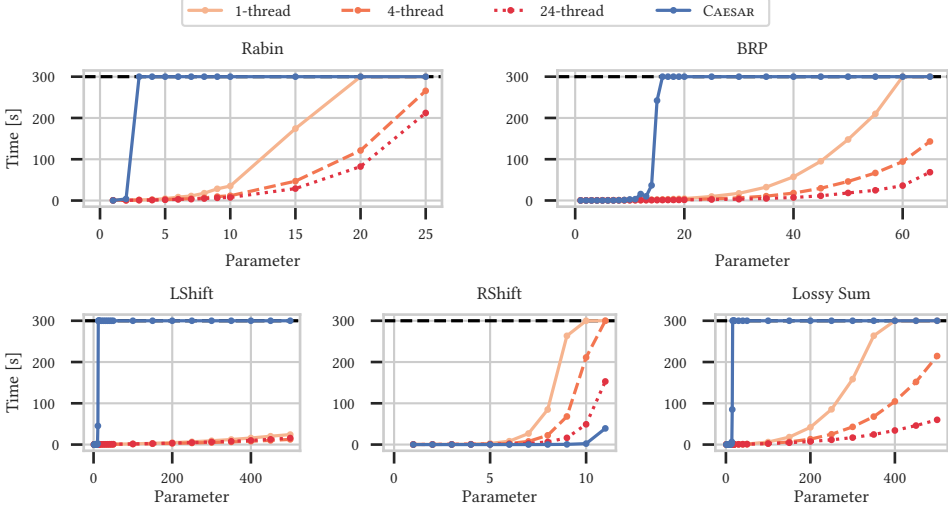}
\vspace{-1em}
\caption{Scalability comparison with \caesar{} on two piecewise affine expectations. The benchmarks on top are without arrays, the  benchmarks in the bottom row use arrays.}
\label{fig:results:scale:caesar}
\end{figure}

\smallskip\noindent\emph{Differences for fixpoints.}
 While \caesar{} \emph{requires} an $n \in \Nats$ such that  $\Phi^n(\dots)$ is a fixpoint,  and compares $\Phi^n(\dots)$  against a bound, \ddsolve{} explicitly and iteratively computes the  fixpoint  $\Phi^n(\dots)$ by computing $\Phi^m(\dots)$ for all $m \leq n$ but using a syntactical equivalence check (\Cref{sec:implementation}). This is a two-sided sword: Computing the expectations for every step is strictly more work, but has the potential to detect a fixpoint already after $m$ steps. In the benchmarks, we use the smallest $n$ such that $\Phi^n(\dots)$ is a fixpoint. This selection of $n$ is therefore slightly beneficial for \caesar{}.
 
 %As \caesar{} does not iteratively compute $\Phi^n$, changing \caesar{} to mimick the behavior of \ddsolve{} would yield a tremendous performance penalty, as witnessed by a small experiment, see Appendix~\ref{app:something}\sjinline{TODO for daniel}. 
 
\subsubsection{Array programs}
We study programs with arrays using the built-in support for array logic in SMT solvers. The benchmark families are given in \Cref{app:benchmark_programs}. We present some results in \Cref{fig:results:scale:caesar} (bottom): For most programs, \ddsolve{} clearly outperforms \caesar{}, with a notable exception for the \texttt{RShift} program, where \ddsolve{} does not manage to prune the TEDD. 

\subsubsection{Nonlinear expectations, conditioning, nondeterministic programs}
Adding nondeterminism to programs does not notably change the performance comparison, see also \Cref{fig:results:nondet:scale:caesar} in \Cref{app:additional_benchmark_results}. This is not surprising: TEDDs already supported minimization and minimization is also used to apply the k-induction operator. 
Similar observations hold for conditioning (\Cref{fig:results:cond:scale:caesar}).
The support for nonlinear expectations works, yet the cost of some pruning operations is often prohibitive. In fact, like for piecewise affine expectations, the total time by \ddsolve{} is often dominated by a few expensive pruning operations. With non-linear expectations, these operations are handled by an SMT-solver for quantifier-free nonlinear integer arithmetic (QF\_NIA), which have unpredictable timings.

\subsection{Comparison with \storm{} (RQ2)}
\begin{figure}
	\centering
\input{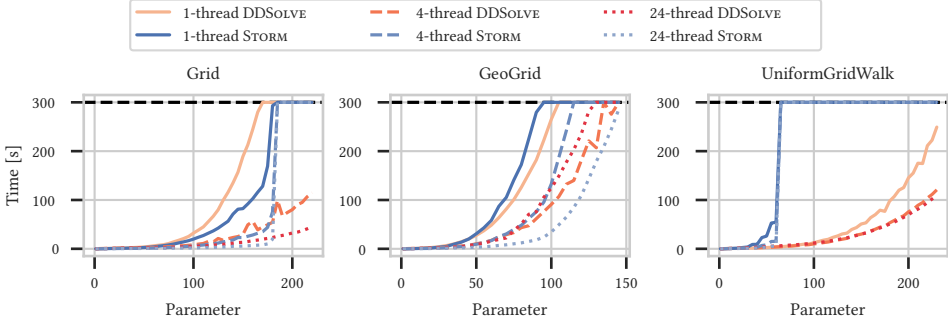}
\vspace{-1em}
\caption{Scalability comparison with \storm{} (using ADDs) on three benchmark families.}
\label{fig:results:linear:scale:storm}
\end{figure}
We compare the verification of probabilistic programs using TEDDs and ADDs. ADDs are fundamentally restricted to finite state space, but benefit from a canonical representation (for a fixed variable ordering) that is not present in TEDDs.\kb{maybe make that clear earlier} We compare against the  probabilistic model checker \storm{}. We configure \storm{} to use its \sylvan-based ADD engine, including the same support for multi-threading. We compare \ddsolve{} on programs and call  \storm{} on  manually translated PRISM files following the principled recipe as in~\cite{DBLP:conf/tacas/BatzCJKKM23}. The programs are given in \Cref{app:benchmark_programs}.

The results are given in Fig.~\ref{fig:results:linear:scale:storm}. 
We observe that on \texttt{Grid} and \texttt{UniformGridWalk}, \ddsolve{} with sufficient number of threads clearly outperforms \storm{} with the same number of threads. On \texttt{GeoGrid}, the situation is less clear, as  \ddsolve{} suffers from a performance penalty going from $4$ to $24$ threads. 
Further, the benefit of using multiple threads is completely different between both tools and over benchmarks, despite using same engine. The effect can easily be explained using the size of the different DDs.

\begin{figure}
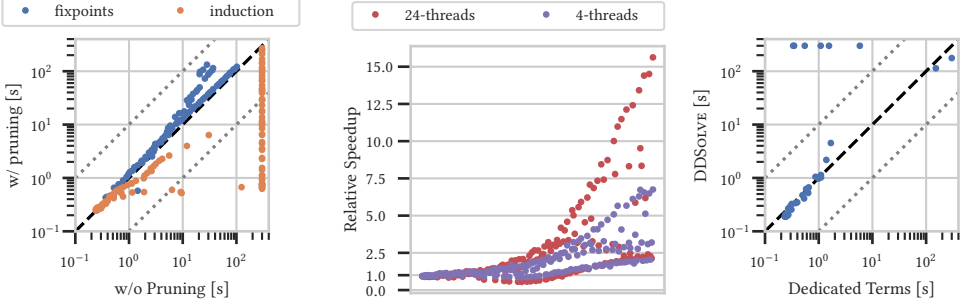

	\begin{subfigure}[b]{0.3\textwidth}
		\centering
		\input{figures/fig4.pgf}
		\subcaption{Run time with(out) pruning. Benchmarks split by \textcolor{blue}{\texttt{fixpoints}} and \textcolor{orange}{\texttt{induction}} tasks.  }
		\label{fig:results:pruning}
	\end{subfigure}\quad
	\begin{subfigure}[b]{0.3\textwidth}
		\centering
		\input{figures/fig5.pgf}
		\subcaption{ Relative speed up (y-axis) on benchmarks (x-axis, sorted by run time) with \textcolor{blue!50!red}{4} or \textcolor{red!40!black}{24} threads. }
		\label{fig:results:threading}
	\end{subfigure}\quad
	\begin{subfigure}[b]{0.3\textwidth}
		\centering
		\input{figures/fig7.pgf}
		\subcaption{Using or not using dedicated term data structures for piecewise linear expectations.}\label{fig:results:specialdatastructures}
		\label{fig:results:terms}
	\end{subfigure}
	\caption{Ablation study (RQ3): The influence of pruning, multithreading, and dedicated data structures.}
\end{figure}

\subsection{Ablation study (RQ3)}
\label{sec:experiments:ablation}
To answer this RQ, we compare configurations of \ddsolve, with short but dedicated experiments. 

\smallskip\noindent\emph{RQ3(a): Pruning.} While predicates and leafs are normalized and thereby yield trivial pruning, SMT-based pruning is optional. We study the impact of this pruning in \Cref{fig:results:pruning}. 
From this data, we observe that on our benchmarks, SMT-based pruning is essential when applying k-induction, while on fixpoint computations, the overhead often exceeds the benefits somewhat. We conjecture that for k-induction, the applications of meets provides a lot of opportunity for (essential) pruning.
%
%\begin{figure}
%	\centering
%\input{figures/fig4.pgf}
%\caption{Every dot represents a benchmark from our complete benchmark suite. We plot the run time with pruning enabled versus pruning disabled. We split the benchmarks in two categories: \texttt{fixpoints} and \texttt{induction} based on the verification task.  }
%\label{fig:results:pruning}	
%\end{figure}

\smallskip\noindent\emph{RQ3(b): Multithreading.}
\ddsolve{} comes with support for multithreading. We compare a fixed configuration with a varying number of threads (\Cref{fig:results:threading}). Beyond the inherent overhead, the work-stealing framework interferes with the incrementality of the SMT solvers. Nevertheless, we see that for 4 threads, we can achieve superlinear speed-ups on the harder instances, and that with 24 threads, we often reach speed-ups between 5x and 16x.

\smallskip\noindent\emph{RQ3(c): Tailored linear terms.}
For piecewise affine benchmarks (as in \Cref{sec:exp:rq1:pwa}), we can use specialized data structures for affine terms and predicates \emph{or} use the generic SMT-based expressions. We study the impact of the specialized data structures on the \caesar{}-benchmarks, using one thread. The results are given in \Cref{fig:results:specialdatastructures}: Many benchmarks are sufficiently small that it does not matter, but for larger benchmarks, it yields a difference of orders of magnitude. The exception are two instances where the SMT-based pruning dominates the run time.

%- ⁠How do XAADDs compare with ADDs, which are used primarily in (finite-state) probabilistic model checking? To answer this question, we compare DdSolve with the state-of-the-art ADD-based model checker Storm, using probabilistic model checking benchmarks. 

\section{Conclusion}
In this paper, we have shown that the suitable representation of expectations is key for scalable deductive verification of probabilistic programs: Our new tool \ddsolve{} significantly outperforms the state of the art tool \caesar.  We consider the best of both worlds: DDs to exploit regularity and SMT solving for pruning and entailment. It is a significant step forward. Nevertheless, (engineering) work regarding pruning and variable orderings, especially for expensive SMT theories, remains open, as is the question to characterize where theory-unaware equivalence checks are sufficient. 
Beyond the proof rules in this paper, a lot of work focusses on inductive synthesis~\cite{DBLP:conf/cav/ChenHWZ15,DBLP:conf/tacas/BatzCJKKM23,DBLP:conf/aaai/ZikelicLHC23} or abstraction-refinement~\cite{DBLP:journals/corr/abs-2508-12344}  and we believe that TEDDs can be integrated in those. It would also be interesting to apply TEDDs to higher-order probabilistic WP, as in \cite{watanabe2025denotationalproductconstructiontemporal}, and for inference tasks beyond verification, as in tools like PSI~\cite{DBLP:conf/cav/GehrMV16}, Dice~\cite{DBLP:journals/pacmpl/HoltzenBM20} or Roulette~\cite{DBLP:journals/pacmpl/MoyCLMH25}.

%\newpage
%\printbibliography
\bibliographystyle{plainnat}
\bibliography{literature}

@article{DBLP:journals/pacmpl/SchroerBKKM23,
  author       = {Philipp Schr{\"{o}}er and
                  Kevin Batz and
                  Benjamin Lucien Kaminski and
                  Joost{-}Pieter Katoen and
                  Christoph Matheja},
  title        = {A Deductive Verification Infrastructure for Probabilistic Programs},
  journal      = {Proc. {ACM} Program. Lang.},
  volume       = {7},
  number       = {{OOPSLA2}},
  pages        = {2052--2082},
  year         = {2023}
}

@article{DBLP:journals/corr/abs-2502-19388,
  author       = {Kevin Batz and
                  Joost{-}Pieter Katoen and
                  Francesca Randone and
                  Tobias Winkler},
  title        = {Foundations for Deductive Verification of Continuous Probabilistic
                  Programs: From Lebesgue to Riemann and Back},
  journal      = {Proc. {ACM} Program. Lang.},
  volume       = {9},
  number       = {{OOPSLA1}},
  pages        = {421--448},
  year         = {2025}
}

@incollection{DBLP:series/faia/BarrettSST21,
  author       = {Clark W. Barrett and
                  Roberto Sebastiani and
                  Sanjit A. Seshia and
                  Cesare Tinelli},
  title        = {Satisfiability Modulo Theories},
  booktitle    = {Handbook of Satisfiability},
  series       = {Frontiers in Artificial Intelligence and Applications},
  pages        = {1267--1329},
  publisher    = {{IOS} Press},
  year         = {2021}
}

@inproceedings{DBLP:conf/qest/GretzKM13,
  author       = {Friedrich Gretz and
                  Joost{-}Pieter Katoen and
                  Annabelle McIver},
  editor       = {Kaustubh R. Joshi and
                  Markus Siegle and
                  Mari{\"{e}}lle Stoelinga and
                  Pedro R. D'Argenio},
  title        = {Prinsys - On a Quest for Probabilistic Loop Invariants},
  booktitle    = {Quantitative Evaluation of Systems - 10th International Conference,
                  {QEST} 2013, Buenos Aires, Argentina, August 27-30, 2013. Proceedings},
  series       = {Lecture Notes in Computer Science},
  pages        = {193--208},
  publisher    = {Springer},
  year         = {2013},
  url          = {https://doi.org/10.1007/978-3-642-40196-1\_17},
  doi          = {10.1007/978-3-642-40196-1\_17},
  bibsource    = {dblp computer science bibliography, https://dblp.org}
}

@inproceedings{DBLP:conf/sas/KatoenMMM10,
  author       = {Joost{-}Pieter Katoen and
                  Annabelle McIver and
                  Larissa Meinicke and
                  Carroll C. Morgan},
  editor       = {Radhia Cousot and
                  Matthieu Martel},
  title        = {Linear-Invariant Generation for Probabilistic Programs: - Automated
                  Support for Proof-Based Methods},
  booktitle    = {Static Analysis - 17th International Symposium, {SAS} 2010, Perpignan,
                  France, September 14-16, 2010. Proceedings},
  series       = {Lecture Notes in Computer Science},
  pages        = {390--406},
  publisher    = {Springer},
  year         = {2010},
  url          = {https://doi.org/10.1007/978-3-642-15769-1\_24},
  doi          = {10.1007/978-3-642-15769-1\_24},
  bibsource    = {dblp computer science bibliography, https://dblp.org}
}

@article{DBLP:journals/jlp/GrooteT03,
  author       = {Jan Friso Groote and
                  Olga Tveretina},
  title        = {Binary decision diagrams for first-order predicate logic},
  journal      = {J. Log. Algebraic Methods Program.},
  volume       = {57},
  number       = {1-2},
  pages        = {1--22},
  year         = {2003}
}

@inproceedings{DBLP:conf/cav/ChatterjeeGMZ22,
  author       = {Krishnendu Chatterjee and
                  Amir Kafshdar Goharshady and
                  Tobias Meggendorfer and
                  Dorde Zikelic},
  title        = {Sound and Complete Certificates for Quantitative Termination Analysis
                  of Probabilistic Programs},
  booktitle    = {{CAV} {(1)}},
  series       = {Lecture Notes in Computer Science},
  volume       = {13371},
  pages        = {55--78},
  publisher    = {Springer},
  year         = {2022}
}

@article{DBLP:journals/jacm/KaminskiKMO18,
  author       = {Benjamin Lucien Kaminski and
                  Joost{-}Pieter Katoen and
                  Christoph Matheja and
                  Federico Olmedo},
  title        = {Weakest Precondition Reasoning for Expected Runtimes of Randomized
                  Algorithms},
  journal      = {J. {ACM}},
  volume       = {65},
  number       = {5},
  pages        = {30:1--30:68},
  year         = {2018}
}

@article{DBLP:journals/pacmpl/GregersenAHTB24,
  author       = {Simon Oddershede Gregersen and
                  Alejandro Aguirre and
                  Philipp G. Haselwarter and
                  Joseph Tassarotti and
                  Lars Birkedal},
  title        = {Asynchronous Probabilistic Couplings in Higher-Order Separation Logic},
  journal      = {Proc. {ACM} Program. Lang.},
  volume       = {8},
  number       = {{POPL}},
  pages        = {753--784},
  year         = {2024}
}

@article{DBLP:journals/pacmpl/BartheHL20,
  author       = {Gilles Barthe and
                  Justin Hsu and
                  Kevin Liao},
  title        = {A probabilistic separation logic},
  journal      = {Proc. {ACM} Program. Lang.},
  volume       = {4},
  number       = {{POPL}},
  pages        = {55:1--55:30},
  year         = {2020}
}

@inproceedings{DBLP:conf/aaai/ZikelicLHC23,
  author       = {Dorde Zikelic and
                  Mathias Lechner and
                  Thomas A. Henzinger and
                  Krishnendu Chatterjee},
  title        = {Learning Control Policies for Stochastic Systems with Reach-Avoid
                  Guarantees},
  booktitle    = {{AAAI}},
  pages        = {11926--11935},
  publisher    = {{AAAI} Press},
  year         = {2023}
}

@inproceedings{DBLP:conf/tacas/BatzCJKKM23,
  author       = {Kevin Batz and
                  Mingshuai Chen and
                  Sebastian Junges and
                  Benjamin Lucien Kaminski and
                  Joost{-}Pieter Katoen and
                  Christoph Matheja},
  title        = {Probabilistic Program Verification via Inductive Synthesis of Inductive
                  Invariants},
  booktitle    = {{TACAS} {(2)}},
  series       = {Lecture Notes in Computer Science},
  volume       = {13994},
  pages        = {410--429},
  publisher    = {Springer},
  year         = {2023}
}

@article{DBLP:journals/ai/SannerB09,
  author       = {Scott Sanner and
                  Craig Boutilier},
  title        = {Practical solution techniques for first-order MDPs},
  journal      = {Artif. Intell.},
  volume       = {173},
  number       = {5-6},
  pages        = {748--788},
  year         = {2009}
}

@inproceedings{DBLP:conf/uai/SannerDB11,
  author       = {Scott Sanner and
                  Karina Valdivia Delgado and
                  Leliane Nunes de Barros},
  title        = {Symbolic Dynamic Programming for Discrete and Continuous State MDPs},
  booktitle    = {{UAI}},
  pages        = {643--652},
  publisher    = {{AUAI} Press},
  year         = {2011}
}

@inproceedings{Barrett2010TheSS,
	title={The SMT-LIB Standard Version 2.0},
	author={Clark W. Barrett and Aaron Stump and Cesare Tinelli},
	year={2010},
	url={https://api.semanticscholar.org/CorpusID:7943149}
}

@inproceedings{z3,
	author       = {Leonardo Mendon{\c{c}}a de Moura and
	Nikolaj S. Bj{\o}rner},
	title        = {{Z3:} An Efficient {SMT} Solver},
	booktitle    = {{TACAS}},
	series       = {Lecture Notes in Computer Science},
	volume       = {4963},
	pages        = {337--340},
	publisher    = {Springer},
	year         = {2008}
}

@inproceedings{cvc5,
	author       = {Haniel Barbosa and
	Clark W. Barrett and
	Martin Brain and
	Gereon Kremer and
	Hanna Lachnitt and
	Makai Mann and
	Abdalrhman Mohamed and
	Mudathir Mohamed and
	Aina Niemetz and
	Andres N{\"{o}}tzli and
	Alex Ozdemir and
	Mathias Preiner and
	Andrew Reynolds and
	Ying Sheng and
	Cesare Tinelli and
	Yoni Zohar},
	title        = {cvc5: {A} Versatile and Industrial-Strength {SMT} Solver},
	booktitle    = {{TACAS} {(1)}},
	series       = {Lecture Notes in Computer Science},
	volume       = {13243},
	pages        = {415--442},
	publisher    = {Springer},
	year         = {2022}
}

@book{mciver_morgan,
	author       = {Annabelle McIver and
	Carroll Morgan},
	title        = {Abstraction, Refinement and Proof for Probabilistic Systems},
	series       = {Monographs in Computer Science},
	publisher    = {Springer},
	year         = {2005}
}

@article{aert,
	author       = {Kevin Batz and
	Benjamin Lucien Kaminski and
	Joost-Pieter Katoen and
	Christoph Matheja and
	Lena Verscht},
	title        = {A Calculus for Amortized Expected Runtimes},
	journal      = {Proc. {ACM} Program. Lang.},
	volume       = {7},
	number       = {{POPL}},
	pages        = {1957--1986},
	year         = {2023}
}

@inproceedings{kind_cav,
	author       = {Kevin Batz and
	Mingshuai Chen and
	Benjamin Lucien Kaminski and
	Joost-Pieter Katoen and
	Christoph Matheja and
	Philipp Schr{\"{o}}er},
	title        = {Latticed k-Induction with an Application to Probabilistic Programs},
	booktitle    = {{CAV} {(2)}},
	series       = {Lecture Notes in Computer Science},
	volume       = {12760},
	pages        = {524--549},
	publisher    = {Springer},
	year         = {2021}
}

@inproceedings{DBLP:conf/fmcad/SheeranSS00,
	author       = {Mary Sheeran and
	Satnam Singh and
	Gunnar St{\aa}lmarck},
	title        = {Checking Safety Properties Using Induction and a SAT-Solver},
	booktitle    = {{FMCAD}},
	series       = {Lecture Notes in Computer Science},
	volume       = {1954},
	pages        = {108--125},
	publisher    = {Springer},
	year         = {2000}
}

@inproceedings{sofware_k_induction,
	author       = {Alastair F. Donaldson and
	Leopold Haller and
	Daniel Kroening and
	Philipp R{\"{u}}mmer},
	title        = {Software Verification Using k-Induction},
	booktitle    = {{SAS}},
	series       = {Lecture Notes in Computer Science},
	volume       = {6887},
	pages        = {351--368},
	publisher    = {Springer},
	year         = {2011}
}

@inproceedings{boosting_k_induction,
	author       = {Dirk Beyer and
	Matthias Dangl and
	Philipp Wendler},
	title        = {Boosting k-Induction with Continuously-Refined Invariants},
	booktitle    = {{CAV} {(1)}},
	series       = {Lecture Notes in Computer Science},
	volume       = {9206},
	pages        = {622--640},
	publisher    = {Springer},
	year         = {2015}
}

@inproceedings{property_directed_k_induction,
	author       = {Dejan Jovanovic and
	Bruno Dutertre},
	title        = {Property-directed k-induction},
	booktitle    = {{FMCAD}},
	pages        = {85--92},
	publisher    = {{IEEE}},
	year         = {2016}
}

@inproceedings{k_induction_without_unrolling,
	author       = {Arie Gurfinkel and
	Alexander Ivrii},
	title        = {K-induction without unrolling},
	booktitle    = {{FMCAD}},
	pages        = {148--155},
	publisher    = {{IEEE}},
	year         = {2017}
}

@article{DBLP:journals/iandc/BurchCMDH92,
  author       = {Jerry R. Burch and
                  Edmund M. Clarke and
                  Kenneth L. McMillan and
                  David L. Dill and
                  L. J. Hwang},
  title        = {Symbolic Model Checking: 10{\^{}}20 States and Beyond},
  journal      = {Inf. Comput.},
  volume       = {98},
  number       = {2},
  pages        = {142--170},
  year         = {1992}
}

@article{DBLP:journals/tc/Bryant86,
	author       = {Randal E. Bryant},
	title        = {Graph-Based Algorithms for Boolean Function Manipulation},
	journal      = {{IEEE} Trans. Computers},
	volume       = {35},
	number       = {8},
	pages        = {677--691},
	year         = {1986}
}

@book{dijkstra_discipline,
	author       = {Edsger W. Dijkstra},
	title        = {A Discipline of Programming},
	publisher    = {Prentice-Hall},
	year         = {1976}
}

@MISC{BarFT-SMTLIB,
	author =	 {Clark Barrett and Pascal Fontaine and Cesare Tinelli},
	title =	 {{The Satisfiability Modulo Theories Library (SMT-LIB)}},
	howpublished = {{\tt www.SMT-LIB.org}},
	year =	 2016,
}

@phdthesis{kaminski_diss,
	author       = {Benjamin Lucien Kaminski},
	title        = {Advanced Weakest Precondition Calculi for Probabilistic Programs},
	school       = {{RWTH} Aachen University, Germany},
	year         = {2019}
}

@article{hardness1,
	author       = {Benjamin Lucien Kaminski and
	Joost-Pieter Katoen and
	Christoph Matheja},
	title        = {On the hardness of analyzing probabilistic programs},
	journal      = {Acta Informatica},
	volume       = {56},
	number       = {3},
	pages        = {255--285},
	year         = {2019}
}

@inproceedings{DBLP:conf/icse/GordonHNR14,
  author       = {Andrew D. Gordon and
                  Thomas A. Henzinger and
                  Aditya V. Nori and
                  Sriram K. Rajamani},
  title        = {Probabilistic programming},
  booktitle    = {{FOSE}},
  pages        = {167--181},
  publisher    = {{ACM}},
  year         = {2014}
}

@inproceedings{DBLP:conf/fmcad/0001BT14,
	author       = {Tim King and
	Clark W. Barrett and
	Cesare Tinelli},
	title        = {Leveraging linear and mixed integer programming for {SMT}},
	booktitle    = {{FMCAD}},
	pages        = {139--146},
	publisher    = {{IEEE}},
	year         = {2014}
}

@inproceedings{DBLP:conf/podc/KushilevitzR92,
  author       = {Eyal Kushilevitz and
                  Michael O. Rabin},
  title        = {Randomized Mutual Exclusion Algorithms Revisited},
  booktitle    = {{PODC}},
  pages        = {275--283},
  publisher    = {{ACM}},
  year         = {1992}
}

@article{DBLP:journals/pacmpl/HoltzenBM20,
  author       = {Steven Holtzen and
                  Guy Van den Broeck and
                  Todd D. Millstein},
  title        = {Scaling exact inference for discrete probabilistic programs},
  journal      = {Proc. {ACM} Program. Lang.},
  volume       = {4},
  number       = {{OOPSLA}},
  pages        = {140:1--140:31},
  year         = {2020}
}

@article{DBLP:journals/sttt/HenselJKQV22,
  author       = {Christian Hensel and
                  Sebastian Junges and
                  Joost{-}Pieter Katoen and
                  Tim Quatmann and
                  Matthias Volk},
  title        = {The probabilistic model checker Storm},
  journal      = {Int. J. Softw. Tools Technol. Transf.},
  volume       = {24},
  number       = {4},
  pages        = {589--610},
  year         = {2022}
}

@inproceedings{DBLP:conf/europar/DijkP14,
  author       = {Tom van Dijk and
                  Jaco C. van de Pol},
  title        = {Lace: Non-blocking Split Deque for Work-Stealing},
  booktitle    = {Euro-Par Workshops {(2)}},
  series       = {Lecture Notes in Computer Science},
  volume       = {8806},
  pages        = {206--217},
  publisher    = {Springer},
  year         = {2014}
}

@article{DBLP:journals/sttt/DijkP17,
  author       = {Tom van Dijk and
                  Jaco van de Pol},
  title        = {Sylvan: multi-core framework for decision diagrams},
  journal      = {Int. J. Softw. Tools Technol. Transf.},
  volume       = {19},
  number       = {6},
  pages        = {675--696},
  year         = {2017}
}

@inproceedings{DBLP:conf/papm/DArgenioJJL01,
  author       = {Pedro R. D'Argenio and
                  Bertrand Jeannet and
                  Henrik Ejersbo Jensen and
                  Kim Guldstrand Larsen},
  title        = {Reachability Analysis of Probabilistic Systems by Successive Refinements},
  booktitle    = {{PAPM-PROBMIV}},
  series       = {Lecture Notes in Computer Science},
  volume       = {2165},
  pages        = {39--56},
  publisher    = {Springer},
  year         = {2001}
}

@inproceedings{DBLP:conf/icalp/BaierCHKR97,
  author       = {Christel Baier and
                  Edmund M. Clarke and
                  Vasiliki Hartonas{-}Garmhausen and
                  Marta Z. Kwiatkowska and
                  Mark Ryan},
  title        = {Symbolic Model Checking for Probabilistic Processes},
  booktitle    = {{ICALP}},
  series       = {Lecture Notes in Computer Science},
  volume       = {1256},
  pages        = {430--440},
  publisher    = {Springer},
  year         = {1997}
}

@book{DBLP:books/daglib/0020348,
  author       = {Christel Baier and
                  Joost{-}Pieter Katoen},
  title        = {Principles of model checking},
  publisher    = {{MIT} Press},
  year         = {2008}
}

@book{DBLP:books/daglib/0068003,
	author       = {Jean H. Gallier},
	title        = {Logic for Computer Science: Foundations of Automatic Theorem Proving},
	publisher    = {Wiley},
	year         = {1987}
}

@book{Tao2011MeasureTheory,
	author    = {Terence Tao},
	title     = {An Introduction to Measure Theory},
	series    = {Graduate Studies in Mathematics},
	volume    = {126},
	publisher = {American Mathematical Society},
	year      = {2011}
}

@article{watanabe2025denotationalproductconstructiontemporal,
      title={A Denotational Product Construction for Temporal Verification of Effectful Higher-Order Programs}, 
      author={Kazuki Watanabe and Mayuko Kori and Taro Sekiyama and Satoshi Kura and Hiroshi Unno},
      year={2025},
      volume={abs/2510.11320},
      journal      = {CoRR},
}

@inproceedings{DBLP:conf/cav/GehrMV16,
  author       = {Timon Gehr and
                  Sasa Misailovic and
                  Martin T. Vechev},
  title        = {{PSI:} Exact Symbolic Inference for Probabilistic Programs},
  booktitle    = {{CAV} {(1)}},
  series       = {Lecture Notes in Computer Science},
  volume       = {9779},
  pages        = {62--83},
  publisher    = {Springer},
  year         = {2016}
}

@article{DBLP:journals/pacmpl/MoyCLMH25,
  author       = {Cameron Moy and
                  Jack Czenszak and
                  John M. Li and
                  Brianna Marshall and
                  Steven Holtzen},
  title        = {Roulette: {A} Language for Expressive, Exact, and Efficient Discrete
                  Probabilistic Programming},
  journal      = {Proc. {ACM} Program. Lang.},
  volume       = {9},
  number       = {{PLDI}},
  pages        = {2081--2105},
  year         = {2025}
}

@inproceedings{DBLP:conf/pldi/GehrMTVWV18,
  author       = {Timon Gehr and
                  Sasa Misailovic and
                  Petar Tsankov and
                  Laurent Vanbever and
                  Pascal Wiesmann and
                  Martin T. Vechev},
  title        = {Bayonet: probabilistic inference for networks},
  booktitle    = {{PLDI}},
  pages        = {586--602},
  publisher    = {{ACM}},
  year         = {2018}
}

@inproceedings{DBLP:conf/cav/ChenHWZ15,
  author       = {Yu{-}Fang Chen and
                  Chih{-}Duo Hong and
                  Bow{-}Yaw Wang and
                  Lijun Zhang},
  title        = {Counterexample-Guided Polynomial Loop Invariant Generation by Lagrange
                  Interpolation},
  booktitle    = {{CAV} {(1)}},
  series       = {Lecture Notes in Computer Science},
  volume       = {9206},
  pages        = {658--674},
  publisher    = {Springer},
  year         = {2015}
}

@article{DBLP:journals/corr/abs-2504-04132,
  author       = {Satoshi Kura and
                  Hiroshi Unno and
                  Takeshi Tsukada},
  title        = {Ranking and Invariants for Lower-Bound Inference in Quantitative Verification
                  of Probabilistic Programs},
  journal      = {CoRR},
  volume       = {abs/2504.04132},
  year         = {2025}
}

@article{DBLP:journals/corr/abs-2508-12344,
  author       = {Guanyan Li and
                  Juanen Li and
                  Zhilei Han and
                  Peixin Wang and
                  Hongfei Fu and
                  Fei He},
  title        = {Structural Abstraction and Refinement for Probabilistic Programs},
  journal      = {CoRR},
  volume       = {abs/2508.12344},
  year         = {2025}
}

@article{DBLP:journals/toplas/OlmedoGJKKM18,
	author       = {Federico Olmedo and
	Friedrich Gretz and
	Nils Jansen and
	Benjamin Lucien Kaminski and
	Joost{-}Pieter Katoen and
	Annabelle McIver},
	title        = {Conditioning in Probabilistic Programming},
	journal      = {{ACM} Trans. Program. Lang. Syst.},
	volume       = {40},
	number       = {1},
	pages        = {4:1--4:50},
	year         = {2018}
}

@inproceedings{adds,
	author={Bahar, R.I. and Frohm, E.A. and Gaona, C.M. and Hachtel, G.D. and Macii, E. and Pardo, A. and Somenzi, F.},
	booktitle={Proceedings of 1993 International Conference on Computer Aided Design (ICCAD)}, 
	title={Algebraic decision diagrams and their applications}, 
	year={1993},
	volume={},
	number={},
	pages={188-191}}

@inproceedings{DBLP:conf/birthday/GrooteV17,
  author       = {Jan Friso Groote and
                  Erik P. de Vink},
  editor       = {Joost{-}Pieter Katoen and
                  Rom Langerak and
                  Arend Rensink},
  title        = {Problem Solving Using Process Algebra Considered Insightful},
  booktitle    = {ModelEd, TestEd, TrustEd - Essays Dedicated to Ed Brinksma on the
                  Occasion of His 60th Birthday},
  series       = {Lecture Notes in Computer Science},
  volume       = {10500},
  pages        = {48--63},
  publisher    = {Springer},
  year         = {2017},
  url          = {https://doi.org/10.1007/978-3-319-68270-9\_3},
  doi          = {10.1007/978-3-319-68270-9\_3},
}

@thesis{spelMonotonicityMarkovChains2018,
  type = {mathesis},
  title = {Monotonicity in {{Markov Chains}}},
  author = {Spel, Jip},
  date = {2018},
  institution = {University of Twente},
  location = {Enschede \& Aachen},
  url = {https://fmt.ewi.utwente.nl/media/Thesis-FINAL_Jip_Spel_yIy3I5R.pdf},
  urldate = {2025-03-15},
  pagetotal = {103}
}

@article{DBLP:journals/pacmpl/GiannarakisSW21,
  author       = {Nick Giannarakis and
                  Alexandra Silva and
                  David Walker},
  title        = {ProbNV: probabilistic verification of network control planes},
  journal      = {Proc. {ACM} Program. Lang.},
  volume       = {5},
  number       = {{ICFP}},
  pages        = {1--30},
  year         = {2021},
  url          = {https://doi.org/10.1145/3473595},
  doi          = {10.1145/3473595},
  bibsource    = {dblp computer science bibliography, https://dblp.org}
}

@inproceedings{DBLP:conf/pldi/SmolkaKKFHK019,
  author       = {Steffen Smolka and
                  Praveen Kumar and
                  David M. Kahn and
                  Nate Foster and
                  Justin Hsu and
                  Dexter Kozen and
                  Alexandra Silva},
  editor       = {Kathryn S. McKinley and
                  Kathleen Fisher},
  title        = {Scalable verification of probabilistic networks},
  booktitle    = {Proceedings of the 40th {ACM} {SIGPLAN} Conference on Programming
                  Language Design and Implementation, {PLDI} 2019, Phoenix, AZ, USA,
                  June 22-26, 2019},
  pages        = {190--203},
  publisher    = {{ACM}},
  year         = {2019},
  url          = {https://doi.org/10.1145/3314221.3314639},
  doi          = {10.1145/3314221.3314639},
  bibsource    = {dblp computer science bibliography, https://dblp.org}
}

@inproceedings{DBLP:conf/popl/SmolkaKFK017,
  author       = {Steffen Smolka and
                  Praveen Kumar and
                  Nate Foster and
                  Dexter Kozen and
                  Alexandra Silva},
  editor       = {Giuseppe Castagna and
                  Andrew D. Gordon},
  title        = {Cantor meets scott: semantic foundations for probabilistic networks},
  booktitle    = {Proceedings of the 44th {ACM} {SIGPLAN} Symposium on Principles of
                  Programming Languages, {POPL} 2017, Paris, France, January 18-20,
                  2017},
  pages        = {557--571},
  publisher    = {{ACM}},
  year         = {2017},
  url          = {https://doi.org/10.1145/3009837.3009843},
  doi          = {10.1145/3009837.3009843},
  bibsource    = {dblp computer science bibliography, https://dblp.org}
}

@article{DBLP:journals/corr/SmolkaKKFK017,
  author       = {Steffen Smolka and
                  David M. Kahn and
                  Praveen Kumar and
                  Nate Foster and
                  Dexter Kozen and
                  Alexandra Silva},
  title        = {Deciding Probabilistic Program Equivalence in NetKAT},
  journal      = {CoRR},
  volume       = {abs/1707.02772},
  year         = {2017},
  url          = {http://arxiv.org/abs/1707.02772},
  eprinttype   = {arXiv},
  eprint       = {1707.02772},
  bibsource    = {dblp computer science bibliography, https://dblp.org}
}

@inproceedings{DBLP:conf/esop/FosterKMR016,
  author       = {Nate Foster and
                  Dexter Kozen and
                  Konstantinos Mamouras and
                  Mark Reitblatt and
                  Alexandra Silva},
  editor       = {Peter Thiemann},
  title        = {Probabilistic NetKAT},
  booktitle    = {Programming Languages and Systems - 25th European Symposium on Programming,
                  {ESOP} 2016, Held as Part of the European Joint Conferences on Theory
                  and Practice of Software, {ETAPS} 2016, Eindhoven, The Netherlands,
                  April 2-8, 2016, Proceedings},
  series       = {Lecture Notes in Computer Science},
  pages        = {282--309},
  publisher    = {Springer},
  year         = {2016},
  url          = {https://doi.org/10.1007/978-3-662-49498-1\_12},
  doi          = {10.1007/978-3-662-49498-1\_12},
  bibsource    = {dblp computer science bibliography, https://dblp.org}
}

@article{DBLP:journals/pacmpl/BaoDF25,
  author  = {Jialu Bao and
                  Emanuele D'Osualdo and
                  Azadeh Farzan},
  title   = {Bluebell: An Alliance of Relational Lifting and Independence for Probabilistic Reasoning},
  journal = {Proc. ACM Program. Lang.},
  volume  = {9},
  number  = {POPL},
  pages   = {1719-1749},
  year    = {2025},
  doi     = {10.1145/3704894},
  url     = {https://dblp.org/rec/journals/pacmpl/BaoDF25}
}

@article{DBLP:journals/pacmpl/HoWR26,
  author  = {Shing Hin Ho and
                  Nicolas Wu and
                  Azalea Raad},
  title   = {Bayesian Separation Logic: A Logical Foundation and Axiomatic Semantics for Probabilistic Programming},
  journal = {Proc. ACM Program. Lang.},
  volume  = {10},
  number  = {POPL},
  pages   = {1557-1585},
  year    = {2026},
  doi     = {10.1145/3776696},
  url     = {https://dblp.org/rec/journals/pacmpl/HoWR26}
}

@article{DBLP:journals/pacmpl/ChatterjeeGMZ24,
  author  = {Krishnendu Chatterjee and
                  Amir Kafshdar Goharshady and
                  Tobias Meggendorfer and
                  Dorde Zikelic},
  title   = {Quantitative Bounds on Resource Usage of Probabilistic Programs},
  journal = {Proc. ACM Program. Lang.},
  volume  = {8},
  number  = {OOPSLA1},
  pages   = {362-391},
  year    = {2024},
  doi     = {10.1145/3649824},
  url     = {https://dblp.org/rec/journals/pacmpl/ChatterjeeGMZ24}
}

@inproceedings{DBLP:conf/cav/SunFCG23,
  author    = {Yican Sun and
                  Hongfei Fu and
                  Krishnendu Chatterjee and
                  Amir Kafshdar Goharshady},
  title     = {Automated Tail Bound Analysis for Probabilistic Recurrence Relations},
  booktitle = {CAV (3)},
  pages     = {16-39},
  year      = {2023},
  doi       = {10.1007/978-3-031-37709-9\_2},
  url       = {https://dblp.org/rec/conf/cav/SunFCG23}
}

@inproceedings{DBLP:conf/pldi/WangS0CG21,
  author    = {Jinyi Wang and
                  Yican Sun and
                  Hongfei Fu and
                  Krishnendu Chatterjee and
                  Amir Kafshdar Goharshady},
  title     = {Quantitative analysis of assertion violations in probabilistic programs},
  booktitle = {PLDI},
  pages     = {1171-1186},
  year      = {2021},
  doi       = {10.1145/3453483.3454102},
  url       = {https://dblp.org/rec/conf/pldi/WangS0CG21}
}

@article{DBLP:journals/pacmpl/Huang0CG19,
  author  = {Mingzhang Huang and
                  Hongfei Fu and
                  Krishnendu Chatterjee and
                  Amir Kafshdar Goharshady},
  title   = {Modular verification for almost-sure termination of probabilistic programs},
  journal = {Proc. ACM Program. Lang.},
  volume  = {3},
  number  = {OOPSLA},
  pages   = {129:1-129:29},
  year    = {2019},
  doi     = {10.1145/3360555},
  url     = {https://dblp.org/rec/journals/pacmpl/Huang0CG19}
}

@inproceedings{DBLP:conf/pldi/Wang0GCQS19,
  author    = {Peixin Wang and
                  Hongfei Fu and
                  Amir Kafshdar Goharshady and
                  Krishnendu Chatterjee and
                  Xudong Qin and
                  Wenjun Shi},
  title     = {Cost analysis of nondeterministic probabilistic programs},
  booktitle = {PLDI},
  pages     = {204-220},
  year      = {2019},
  doi       = {10.1145/3314221.3314581},
  url       = {https://dblp.org/rec/conf/pldi/Wang0GCQS19}
}

@inproceedings{DBLP:conf/cav/ChatterjeeFG16,
  author    = {Krishnendu Chatterjee and
                  Hongfei Fu and
                  Amir Kafshdar Goharshady},
  title     = {Termination Analysis of Probabilistic Programs Through Positivstellensatz's},
  booktitle = {CAV (1)},
  pages     = {3-22},
  year      = {2016},
  doi       = {10.1007/978-3-319-41528-4\_1},
  url       = {https://dblp.org/rec/conf/cav/ChatterjeeFG16}
}

@article{DBLP:journals/pacmpl/LiLHW0025,
  author  = {Guanyan Li and
                  Juanen Li and
                  Zhilei Han and
                  Peixin Wang and
                  Hongfei Fu and
                  Fei He},
  title   = {Structural Abstraction and Refinement for Probabilistic Programs},
  journal = {Proc. ACM Program. Lang.},
  volume  = {9},
  number  = {OOPSLA2},
  pages   = {1809-1836},
  year    = {2025},
  doi     = {10.1145/3763115},
  url     = {https://dblp.org/rec/journals/pacmpl/LiLHW0025}
}

@article{DBLP:journals/corr/abs-2203-04422,
  author  = {Guanyan Li and
                  Zhilei Han and
                  Fei He},
  title   = {ProbTA: A sound and complete proof rule for probabilistic verification},
  journal = {CoRR},
  volume  = {abs/2203.04422},
  year    = {2022},
  doi     = {10.48550/arXiv.2203.04422},
  url     = {https://dblp.org/rec/journals/corr/abs-2203-04422}
}

@inproceedings{DBLP:conf/pldi/ChenH20,
  author    = {Jianhui Chen and
                  Fei He},
  title     = {Proving almost-sure termination by omega-regular decomposition},
  booktitle = {PLDI},
  pages     = {869-882},
  year      = {2020},
  doi       = {10.1145/3385412.3386002},
  url       = {https://dblp.org/rec/conf/pldi/ChenH20}
}

@article{DBLP:journals/pacmpl/AlbarghouthiH18,
  author  = {Aws Albarghouthi and
                  Justin Hsu},
  title   = {Synthesizing coupling proofs of differential privacy},
  journal = {Proc. ACM Program. Lang.},
  volume  = {2},
  number  = {POPL},
  pages   = {58:1-58:30},
  year    = {2018},
  doi     = {10.1145/3158146},
  url     = {https://dblp.org/rec/journals/pacmpl/AlbarghouthiH18}
}

@inproceedings{DBLP:conf/cav/AlbarghouthiH18,
  author    = {Aws Albarghouthi and
                  Justin Hsu},
  title     = {Constraint-Based Synthesis of Coupling Proofs},
  booktitle = {CAV (1)},
  pages     = {327-346},
  year      = {2018},
  doi       = {10.1007/978-3-319-96145-3\_18},
  url       = {https://dblp.org/rec/conf/cav/AlbarghouthiH18}
}

@inproceedings{DBLP:conf/popl/BartheGHS17,
  author    = {Gilles Barthe and
                  Benjamin Grégoire and
                  Justin Hsu and
                  Pierre-Yves Strub},
  title     = {Coupling proofs are probabilistic product programs},
  booktitle = {POPL},
  pages     = {161-174},
  year      = {2017},
  doi       = {10.1145/3009837.3009896},
  url       = {https://dblp.org/rec/conf/popl/BartheGHS17}
}

@inproceedings{DBLP:conf/cav/BartheEFH16,
  author    = {Gilles Barthe and
                  Thomas Espitau and
                  Luis María Ferrer Fioriti and
                  Justin Hsu},
  title     = {Synthesizing Probabilistic Invariants via Doob's Decomposition},
  booktitle = {CAV (1)},
  pages     = {43-61},
  year      = {2016},
  doi       = {10.1007/978-3-319-41528-4\_3},
  url       = {https://dblp.org/rec/conf/cav/BartheEFH16}
}

@inproceedings{DBLP:conf/popl/BartheKOB12,
  author    = {Gilles Barthe and
                  Boris Köpf and
                  Federico Olmedo and
                  Santiago Zanella-Béguelin},
  title     = {Probabilistic relational reasoning for differential privacy},
  booktitle = {POPL},
  pages     = {97-110},
  year      = {2012},
  doi       = {10.1145/2103656.2103670},
  url       = {https://dblp.org/rec/conf/popl/BartheKOB12}
}

@article{DBLP:journals/pacmpl/0001024,
  author  = {Satoshi Kura and
                  Hiroshi Unno},
  title   = {Automated Verification of Higher-Order Probabilistic Programs via a Dependent Refinement Type System},
  journal = {Proc. ACM Program. Lang.},
  volume  = {8},
  number  = {ICFP},
  pages   = {973-1002},
  year    = {2024},
  doi     = {10.1145/3674662},
  url     = {https://dblp.org/rec/journals/pacmpl/0001024}
}

@article{DBLP:journals/corr/abs-2512-00270,
author = {Kura, Satoshi and Unno, Hiroshi},
title = {A Hierarchy of Supermartingales for {$\omega$}-Regular Verification},
year = {2026},
issue_date = {June 2026},
publisher = {Association for Computing Machinery},
address = {New York, NY, USA},
volume = {10},
number = {PLDI},
doi = {10.1145/3808257},
month = jun,
articleno = {179},
numpages = {24}
}

@article{DBLP:journals/fmsd/BaoTPHR25,
  author  = {Jialu Bao and
                  Nitesh Trivedi and
                  Drashti Pathak and
                  Justin Hsu and
                  Subhajit Roy},
  title   = {Data-driven invariant learning for probabilistic programs},
  journal = {Formal Methods Syst. Des.},
  volume  = {66},
  number  = {2},
  pages   = {278-306},
  year    = {2025},
  doi     = {10.1007/s10703-024-00466-x},
  url     = {https://dblp.org/rec/journals/fmsd/BaoTPHR25}
}

@inproceedings{DBLP:conf/ijcai/BaoTPH023,
  author    = {Jialu Bao and
                  Nitesh Trivedi and
                  Drashti Pathak and
                  Justin Hsu and
                  Subhajit Roy},
  title     = {Data-Driven Invariant Learning for Probabilistic Programs (Extended Abstract)},
  booktitle = {IJCAI},
  pages     = {6415-6419},
  year      = {2023},
  doi       = {10.24963/ijcai.2023/712},
  url       = {https://dblp.org/rec/conf/ijcai/BaoTPH023}
}

@article{DBLP:journals/pacmpl/SusagLHR22,
  author  = {Zachary Susag and
                  Sumit Lahiri and
                  Justin Hsu and
                  Subhajit Roy},
  title   = {Symbolic execution for randomized programs},
  journal = {Proc. ACM Program. Lang.},
  volume  = {6},
  number  = {OOPSLA2},
  pages   = {1583-1612},
  year    = {2022},
  doi     = {10.1145/3563344},
  url     = {https://dblp.org/rec/journals/pacmpl/SusagLHR22}
}

@inproceedings{DBLP:conf/cav/BaoTPHR22,
  author    = {Jialu Bao and
                  Nitesh Trivedi and
                  Drashti Pathak and
                  Justin Hsu and
                  Subhajit Roy},
  title     = {Data-Driven Invariant Learning for Probabilistic Programs},
  booktitle = {CAV (1)},
  pages     = {33-54},
  year      = {2022},
  doi       = {10.1007/978-3-031-13185-1\_3},
  url       = {https://dblp.org/rec/conf/cav/BaoTPHR22}
}

@article{DBLP:journals/pacmpl/SmithHA19,
  author  = {Calvin Smith and
                  Justin Hsu and
                  Aws Albarghouthi},
  title   = {Trace abstraction modulo probability},
  journal = {Proc. ACM Program. Lang.},
  volume  = {3},
  number  = {POPL},
  pages   = {39:1-39:31},
  year    = {2019},
  doi     = {10.1145/3290352},
  url     = {https://dblp.org/rec/journals/pacmpl/SmithHA19}
}

@inproceedings{DBLP:conf/cpp/MarionneauB0B26,
  author       = {Virgil Marionneau and
                  F{\'{e}}lix Sassus Bourda and
                  Alejandro Aguirre and
                  Lars Birkedal},
  editor       = {Kathrin Stark and
                  Yannick Zakowski and
                  Nikhil Swamy and
                  Nicolas Tabareau},
  title        = {Modular Specifications and Implementations of Random Samplers in Higher-Order
                  Separation Logic},
  booktitle    = {Proceedings of the 15th {ACM} {SIGPLAN} International Conference on
                  Certified Programs and Proofs, {CPP} 2026, Rennes, France, January
                  12-13, 2026},
  pages        = {368--382},
  publisher    = {{ACM}},
  year         = {2026},
  url          = {https://doi.org/10.1145/3779031.3779109},
  doi          = {10.1145/3779031.3779109},
  bibsource    = {dblp computer science bibliography, https://dblp.org}
}

@article{DBLP:journals/corr/abs-2604-12713,
  author       = {Philipp G. Haselwarter and
                  Alejandro Aguirre and
                  Simon Oddershede Gregersen and
                  Kwing Hei Li and
                  Joseph Tassarotti and
                  Lars Birkedal},
  title        = {Modular Verification of Differential Privacy in Probabilistic Higher-Order
                  Separation Logic (Extended Version)},
  journal      = {CoRR},
  volume       = {abs/2604.12713},
  year         = {2026},
  url          = {https://doi.org/10.48550/arXiv.2604.12713},
  doi          = {10.48550/ARXIV.2604.12713},
  eprinttype   = {arXiv},
  eprint       = {2604.12713},
  bibsource    = {dblp computer science bibliography, https://dblp.org}
}

@article{DBLP:journals/pacmpl/Li0GHTB25,
  author       = {Kwing Hei Li and
                  Alejandro Aguirre and
                  Simon Oddershede Gregersen and
                  Philipp G. Haselwarter and
                  Joseph Tassarotti and
                  Lars Birkedal},
  title        = {Modular Reasoning about Error Bounds for Concurrent Probabilistic
                  Programs},
  journal      = {Proc. {ACM} Program. Lang.},
  volume       = {9},
  number       = {{ICFP}},
  pages        = {276--305},
  year         = {2025},
  url          = {https://doi.org/10.1145/3747514},
  doi          = {10.1145/3747514},
  bibsource    = {dblp computer science bibliography, https://dblp.org}
}

@article{DBLP:journals/pacmpl/HaselwarterLAGTB25,
  author       = {Philipp G. Haselwarter and
                  Kwing Hei Li and
                  Alejandro Aguirre and
                  Simon Oddershede Gregersen and
                  Joseph Tassarotti and
                  Lars Birkedal},
  title        = {Approximate Relational Reasoning for Higher-Order Probabilistic Programs},
  journal      = {Proc. {ACM} Program. Lang.},
  volume       = {9},
  number       = {{POPL}},
  pages        = {1196--1226},
  year         = {2025},
  url          = {https://doi.org/10.1145/3704877},
  doi          = {10.1145/3704877},
  bibsource    = {dblp computer science bibliography, https://dblp.org}
}

@article{DBLP:journals/pacmpl/StassenMZAB25,
  author       = {Philipp Stassen and
                  Rasmus Ejlers M{\o}gelberg and
                  Maaike Zwart and
                  Alejandro Aguirre and
                  Lars Birkedal},
  title        = {Modelling Recursion and Probabilistic Choice in Guarded Type Theory},
  journal      = {Proc. {ACM} Program. Lang.},
  volume       = {9},
  number       = {{POPL}},
  pages        = {1417--1445},
  year         = {2025},
  url          = {https://doi.org/10.1145/3704884},
  doi          = {10.1145/3704884},
  bibsource    = {dblp computer science bibliography, https://dblp.org}
}

@article{DBLP:journals/corr/abs-2511-10135,
  author       = {Kwing Hei Li and
                  Alejandro Aguirre and
                  Joseph Tassarotti and
                  Lars Birkedal},
  title        = {Contextual Refinement of Higher-Order Concurrent Probabilistic Programs},
  journal      = {CoRR},
  volume       = {abs/2511.10135},
  year         = {2025},
  url          = {https://doi.org/10.48550/arXiv.2511.10135},
  doi          = {10.48550/ARXIV.2511.10135},
  eprinttype   = {arXiv},
  eprint       = {2511.10135},
  bibsource    = {dblp computer science bibliography, https://dblp.org}
}

@article{DBLP:journals/pacmpl/Gregersen0HTB24,
  author       = {Simon Oddershede Gregersen and
                  Alejandro Aguirre and
                  Philipp G. Haselwarter and
                  Joseph Tassarotti and
                  Lars Birkedal},
  title        = {Almost-Sure Termination by Guarded Refinement},
  journal      = {Proc. {ACM} Program. Lang.},
  volume       = {8},
  number       = {{ICFP}},
  pages        = {203--233},
  year         = {2024},
  url          = {https://doi.org/10.1145/3674632},
  doi          = {10.1145/3674632},
  bibsource    = {dblp computer science bibliography, https://dblp.org}
}

@article{DBLP:journals/pacmpl/0001HMLGTB24,
  author       = {Alejandro Aguirre and
                  Philipp G. Haselwarter and
                  Markus de Medeiros and
                  Kwing Hei Li and
                  Simon Oddershede Gregersen and
                  Joseph Tassarotti and
                  Lars Birkedal},
  title        = {Error Credits: Resourceful Reasoning about Error Bounds for Higher-Order
                  Probabilistic Programs},
  journal      = {Proc. {ACM} Program. Lang.},
  volume       = {8},
  number       = {{ICFP}},
  pages        = {284--316},
  year         = {2024},
  url          = {https://doi.org/10.1145/3674635},
  doi          = {10.1145/3674635},
  bibsource    = {dblp computer science bibliography, https://dblp.org}
}

@article{DBLP:journals/pacmpl/HaselwarterLMG024,
  author       = {Philipp G. Haselwarter and
                  Kwing Hei Li and
                  Markus de Medeiros and
                  Simon Oddershede Gregersen and
                  Alejandro Aguirre and
                  Joseph Tassarotti and
                  Lars Birkedal},
  title        = {Tachis: Higher-Order Separation Logic with Credits for Expected Costs},
  journal      = {Proc. {ACM} Program. Lang.},
  volume       = {8},
  number       = {{OOPSLA2}},
  pages        = {1189--1218},
  year         = {2024},
  url          = {https://doi.org/10.1145/3689753},
  doi          = {10.1145/3689753},
  bibsource    = {dblp computer science bibliography, https://dblp.org}
}

@article{DBLP:journals/pacmpl/AguirreB23,
  author       = {Alejandro Aguirre and
                  Lars Birkedal},
  title        = {Step-Indexed Logical Relations for Countable Nondeterminism and Probabilistic
                  Choice},
  journal      = {Proc. {ACM} Program. Lang.},
  volume       = {7},
  number       = {{POPL}},
  pages        = {33--60},
  year         = {2023},
  url          = {https://doi.org/10.1145/3571195},
  doi          = {10.1145/3571195},
  bibsource    = {dblp computer science bibliography, https://dblp.org}
}

@article{DBLP:journals/pacmpl/AvanziniMS23,
  author  = {Martin Avanzini and
                  Georg Moser and
                  Michael Schaper},
  title   = {Automated Expected Value Analysis of Recursive Programs},
  journal = {Proc. ACM Program. Lang.},
  volume  = {7},
  number  = {PLDI},
  pages   = {1050-1072},
  year    = {2023},
  doi     = {10.1145/3591263},
  url     = {https://dblp.org/rec/journals/pacmpl/AvanziniMS23}
}

@article{DBLP:journals/pacmpl/AvanziniMS20,
  author  = {Martin Avanzini and
                  Georg Moser and
                  Michael Schaper},
  title   = {A modular cost analysis for probabilistic programs},
  journal = {Proc. ACM Program. Lang.},
  volume  = {4},
  number  = {OOPSLA},
  pages   = {172:1-172:30},
  year    = {2020},
  doi     = {10.1145/3428240},
  url     = {https://dblp.org/rec/journals/pacmpl/AvanziniMS20}
}

@article{DBLP:journals/corr/abs-1908-11343,
  author  = {Martin Avanzini and
                  Michael Schaper and
                  Georg Moser},
  title   = {Modular Runtime Complexity Analysis of Probabilistic While Programs},
  journal = {CoRR},
  volume  = {abs/1908.11343},
  year    = {2019},
  url     = {https://dblp.org/rec/journals/corr/abs-1908-11343}
}

@article{DBLP:journals/acta/ChistikovDM17,
  author  = {Dmitry Chistikov and
                  Rayna Dimitrova and
                  Rupak Majumdar},
  title   = {Approximate counting in SMT and value estimation for probabilistic programs},
  journal = {Acta Informatica},
  volume  = {54},
  number  = {8},
  pages   = {729-764},
  year    = {2017},
  doi     = {10.1007/s00236-017-0297-2},
  url     = {https://dblp.org/rec/journals/acta/ChistikovDM17}
}

@inproceedings{DBLP:conf/tacas/ChistikovDM15,
  author    = {Dmitry Chistikov and
                  Rayna Dimitrova and
                  Rupak Majumdar},
  title     = {Approximate Counting in SMT and Value Estimation for Probabilistic Programs},
  booktitle = {TACAS},
  pages     = {320-334},
  year      = {2015},
  doi       = {10.1007/978-3-662-46681-0\_26},
  url       = {https://dblp.org/rec/conf/tacas/ChistikovDM15}
}

@article{DBLP:journals/pacmpl/EneaMMS26,
  author  = {Constantin Enea and
                  Rupak Majumdar and
                  Harshit Jitendra Motwani and
                  V. R. Sathiyanarayana},
  title   = {Verifying Almost-Sure Termination for Randomized Distributed Algorithms},
  journal = {Proc. ACM Program. Lang.},
  volume  = {10},
  number  = {POPL},
  pages   = {1412-1441},
  year    = {2026},
  doi     = {10.1145/3776691},
  url     = {https://dblp.org/rec/journals/pacmpl/EneaMMS26}
}

@article{DBLP:journals/pacmpl/MajumdarS25,
  author  = {Rupak Majumdar and
                  V. R. Sathiyanarayana},
  title   = {Sound and Complete Proof Rules for Probabilistic Termination},
  journal = {Proc. ACM Program. Lang.},
  volume  = {9},
  number  = {POPL},
  pages   = {1871-1902},
  year    = {2025},
  doi     = {10.1145/3704899},
  url     = {https://dblp.org/rec/journals/pacmpl/MajumdarS25}
}

@article{DBLP:journals/pacmpl/MajumdarS24,
  author  = {Rupak Majumdar and
                  V. R. Sathiyanarayana},
  title   = {Positive Almost-Sure Termination: Complexity and Proof Rules},
  journal = {Proc. ACM Program. Lang.},
  volume  = {8},
  number  = {POPL},
  pages   = {1089-1117},
  year    = {2024},
  doi     = {10.1145/3632879},
  url     = {https://dblp.org/rec/journals/pacmpl/MajumdarS24}
}

@article{10.1145/3808348,
author = {Kura, Satoshi and Unno, Hiroshi and Tsukada, Takeshi},
title = {Supermartingales for Unique Fixed Points: A Unified Approach to Lower Bound Verification},
year = {2026},
issue_date = {June 2026},
publisher = {Association for Computing Machinery},
address = {New York, NY, USA},
volume = {10},
number = {PLDI},
doi = {10.1145/3808348},
journal = {Proc. ACM Program. Lang.},
month = jun,
articleno = {270},
numpages = {24}
}
\newpage
\appendix
\crefalias{section}{appendix}
\crefalias{subsection}{appendix}
\crefalias{subsubsection}{appendix}

% \section{Omitted Algorithms}
%  We denote by $\xlabt(\xadd)$  the label $\xlabt(\xroot)$ of the terminal root $\xroot$ (when $\xadd$ has one node) and by $\xlabn(\xadd)$ the label of non-terminal root $\xroot$.
 
% \input{appendix/algorithms}
% \clearpage
% \newpage

\section{Additional Proofs}
\label{app:additional_proofs}
Recall that $\sizeof{\xadd}$ denotes the number of nodes of $\xadd$.
The time bounds below distinguish procedure invocations from cache misses. A \emph{cache miss} is an invocation that is not returned immediately and is therefore fully evaluated. Unless stated otherwise, cache lookups, cache insertions, root-label access, and successor access are counted as unit-cost operations; calls to \ObtainAlg, SMT queries, term/formula substitution, and nested calls to other algorithms are accounted for explicitly.

\begin{definition}[Rank]
\label{def:tedd_rank}
	Let $\xadd$ be a TEDD and let $v$ be a node of $\xadd$. The \emph{rank} of $v$ in $\xadd$, denoted $\mathrm{rank}_{\xadd}(v)$, is the maximum number of edges of a directed path from $v$ to a terminal node. Thus terminal nodes have rank $0$. For a tuple of nodes $(v_1,\ldots,v_n)$ from TEDDs $\xadd_1,\ldots,\xadd_n$, its rank is $\sum_{i=1}^n \mathrm{rank}_{\xadd_i}(v_i)$.
\end{definition}

\subsection[Correctness and Cache Misses of \ApplyAlg]{Correctness and Cache Misses of \ApplyAlg (Algorithm~\ref{alg:apply})}
\label{app:proof_apply}

The following theorem is equivalent to Theorem~\ref{thm:apply_soundness_main} since we inline the specification of Algorithm~\ref{alg:apply}.

\begin{theorem}
\label{thm:apply_correctness_complexity}
\label{app:proof_apply_complexity}
	Let $\termop\colon\term_{\sorta_1}\times\ldots\times\term_{\sorta_n}\to\term_\sorta$ be a term operator, and let $\xadd_1,\ldots,\xadd_n$ be TEDDs of types $\sorta_1,\ldots,\sorta_n$, respectively. 
	%If $\sorta=\sortbool$, assume that $\termop$ returns only $\true$ or $\false$ on all tuples of terminal labels reached by this call. 
	If \apply{$\termop,\xadd_1,\ldots,\xadd_n$} returns $\xadd$, then $\xadd$ is a $\sorta$-TEDD and
	\[
		\forall \inta\in\interprets\colon\quad
		\sem{\xadd}{\inta}
		\eeq
		\sem{\switchfunop{\termop}(\xswitchfun{\xadd_1},\ldots,\xswitchfun{\xadd_n})}{\inta}~.
	\]
	Moreover, \apply{$\termop,\xadd_1,\ldots,\xadd_n$} has at most
	\[
		\sizeof{\xadd_1}\cdot\ldots\cdot\sizeof{\xadd_n}
	\]
	cache misses.
\end{theorem}

\begin{proof}
	We first prove correctness and then establish the cache-miss bound.
	The key idea is the usual product traversal: \ApplyAlg synchronously splits all current TEDDs on the $\varord$-least available root formula and applies $\termop$ only at tuples of terminal nodes.  We maintain the cache invariant that each stored result is a well-formed $\sorta$-TEDD with the claimed denotation, making cache hits correct.

	For correctness, the proof strategy is induction on tuple rank.  Rank $0$ is the terminal case.  For positive rank, let $\forma$ be the chosen least root formula.  In both recursive calls, at least one component advances to a proper successor and no component increases in rank, so the induction hypothesis applies.  The recursive results therefore denote the two branches of the desired operator application; since all remaining roots are $\varord$-larger than $\forma$, the final \ObtainAlg call is well formed.  Splitting any interpretation according to whether it satisfies $\forma$ gives the required denotation.

	For the bound, each cache miss is determined by one tuple of nodes $(v_1,\ldots,v_n)$ from the input TEDDs and the fixed operator $\termop$.  After that tuple is fully evaluated it is cached, so it cannot cause another miss.  There are at most $\sizeof{\xadd_1}\cdot\ldots\cdot\sizeof{\xadd_n}$ such tuples.
\end{proof}

\subsection[Correctness and Cache Misses of \SubstituteAlg]{Correctness and Cache Misses of \SubstituteAlg (Algorithm~\ref{alg:substitute})}
\label{app:proof_substitute}

The following theorem is equivalent to Theorem~\ref{thm:substitute_soundness_main} since we inline the specification of Algorithm~\ref{alg:substitute}.

\begin{theorem}
\label{thm:substitute_correctness_complexity}
	Let $\xadd$ be a $\sorta$-TEDD, let $\hastype{\vara}{\sorta'}$, and let $\hastype{\terma}{\sorta'}$. If \substitute{$\xadd,\vara,\terma$} returns $\xaddb$, then $\xaddb$ is a $\sorta$-TEDD and
	\[
		\forall \inta\in\interprets\colon\quad
		\sem{\xaddb}{\inta}
		\eeq
		\sem{\xswitchfun{\xadd}\switchfunsubst{\vara}{\terma}}{\inta}~.
	\]
	For fixed $\vara$ and $\terma$, \SubstituteAlg has at most $\sizeof{\xadd}$ cache misses. This count is for \SubstituteAlg itself and excludes the nested \ApplyAlg calls in line~\ref{algosubst:finalobtain}.
\end{theorem}

\begin{proof}
	We first prove correctness and then establish the cache-miss bound.
	The key idea is to substitute recursively and let \ApplyAlg rebuild the root case distinction, because changing a guard may violate the fixed order.  We maintain the cache invariant that every stored triple $(\xadd',\vara,\terma)$ denotes $\xswitchfun{\xadd'}\switchfunsubst{\vara}{\terma}$, so cache hits are correct.

	For correctness, the proof strategy is induction on the rank of the current node.  At a terminal node, \ObtainAlg is called on the substituted terminal label.  At a non-terminal node with guard $\forma$, the algorithm substitutes the guard to obtain $\forma'$, recursively substitutes both successors, and then applies $\switchITEsymbol$ to the guard TEDD for $\forma'$ and the two substituted children.  The recursive calls have smaller rank, and \Cref{thm:apply_correctness_complexity} gives a well-formed TEDD denoting exactly the substituted case distinction.

	For the bound, $\vara$ and $\terma$ are fixed, and each \SubstituteAlg cache miss is keyed by one input sub-TEDD.  Hence there is at most one such miss per node of $\xadd$, i.e. at most $\sizeof{\xadd}$ misses; the nested \ApplyAlg calls are counted separately.
\end{proof}

\subsection[Correctness and Invocations of \PruneAlg]{Correctness and Invocations of \PruneAlg (Algorithm~\ref{alg:prune})}
\label{app:proof_prune}
The following theorem is equivalent to Theorem~\ref{thm:prune_soundness_main} since we inline the specification of Algorithm~\ref{alg:prune}.

\begin{theorem}[\PruneAlg correctness and invocations]
\label{thm:prune_correctness_complexity}
	Assume that the initial context $\contexta$ is satisfiable modulo $\theory$. If \prune{$\theory,\xadd,\contexta$} returns $\xaddb$, then $\xaddb$ is a $\sorta$-TEDD and, for every $\theory$-interpretation $\inta$ satisfying $\contexta$,
	\[
		\sem{\xaddb}{\inta}\eeq\sem{\xadd}{\inta}.
	\]
	Moreover, $\xaddb$ is pruned modulo $\theory$ under context $\contexta$ in the sense of \Cref{def:pruned_tedd}. In particular, the call with $\contexta=\true$ returns a $\sorta$-TEDD equivalent to $\xadd$ modulo $\theory$ and pruned modulo $\theory$.

	Let $m=\mathrm{rank}_{\xadd}(\xroot)$ be the rank from \Cref{def:tedd_rank}. Then \PruneAlg makes at most $2^{m+1}-1$ invocations, counting the initial invocation and cache-hit invocations.
\end{theorem}

\begin{proof}
	We first prove correctness and prunedness and then establish the invocation bound.
	The key idea is to traverse the TEDD under the accumulated context and delete exactly those branches whose extended context is unsatisfiable modulo $\theory$.  The cache invariant is that every stored triple $(\theory,\xadd',\contexta')$ is well formed, equivalent to $\xadd'$ under $\contexta'$, and pruned under $\contexta'$.

	For correctness and prunedness, the proof strategy is induction on the rank of the current root, assuming the current context is satisfiable.  Terminal nodes are already pruned.  At a non-terminal node with guard $\formb$, the algorithm considers $\contexta'\wedge\formb$ and $\contexta'\wedge\neg\formb$.  If one is unsatisfiable, all models of $\contexta'$ take the other branch, so returning the recursively pruned reachable successor preserves equivalence and prunedness.  If both are satisfiable, the recursively pruned successors are equivalent and pruned under their respective contexts, and \ObtainAlg rebuilds a well-formed node; equivalence follows by splitting on $\formb$, and prunedness follows from the satisfiable path witnesses in the children.  The case $\contexta=\true$ is exactly equivalence modulo $\theory$.

	For the invocation bound, view the run as a recursion tree.  Each non-leaf has at most two children, and every recursive edge moves to a proper successor, decreasing rank.  Thus the height is at most $m=\mathrm{rank}_{\xadd}(\xroot)$, giving at most $2^{m+1}-1$ invocations.
\end{proof}

\subsection[Correctness and Cache Misses of \MinimumAlg]{Correctness and Cache Misses of \MinimumAlg (Algorithm~\ref{alg:minimum})}
\label{app:proof_minimum}

The following theorem is equivalent to Theorem~\ref{thm:minimum_soundness_main} since we inline the specification of Algorithm~\ref{alg:minimum}.

\begin{theorem}[\MinimumAlg correctness and cache misses]
\label{thm:minimum_correctness_complexity}
	Let $\xadd_1$ and $\xadd_2$ be $\sorteureal$-TEDDs. If \minfunc{$\xadd_1,\xadd_2$} returns $\xaddb$, then $\xaddb$ is an $\sorteureal$-TEDD and, for every interpretation $\inta$,
	\[
		\sem{\xaddb}{\inta}
		\eeq
		\min\{\sem{\xadd_1}{\inta},\sem{\xadd_2}{\inta}\}.
	\]
	Moreover, \MinimumAlg has at most $\sizeof{\xadd_1}\cdot\sizeof{\xadd_2}$ cache misses. This count excludes the nested \ApplyAlg calls in the final merge line immediately before line~\ref{alg:min-previous:recur_case2}, which are accounted for by \Cref{thm:apply_correctness_complexity}.
\end{theorem}

\begin{proof}
	We first prove correctness and then establish the cache-miss bound.
	The key idea is again a product traversal: split the current inputs on the relevant least root label, and compute the actual minimum only when both inputs are terminal.  The cache invariant is that every stored pair denotes the pointwise minimum of that pair.

	For correctness, the proof strategy is induction on tuple rank.  At rank $0$, the algorithm returns the TEDD that selects $\xadd_1'$ iff its terminal label is strictly smaller than that of $\xadd_2'$, which is the pointwise minimum.  At positive rank, the algorithm chooses the appropriate split label $\forma$, recurses on the plus and minus inputs, and each recursive call has smaller tuple rank.  The recursive results denote the two branchwise minima; the final \ApplyAlg call with $\switchITEsymbol$ merges them into the minimum of the original pair by splitting interpretations according to $\forma$.

	For the bound, a \MinimumAlg cache miss is determined by a pair of nodes from $\xadd_1$ and $\xadd_2$, and each pair is cached after its first full evaluation.  Hence there are at most $\sizeof{\xadd_1}\cdot\sizeof{\xadd_2}$ misses, excluding the nested \ApplyAlg calls.
\end{proof}

\subsection[Correctness and Invocations of \EntailsAlg]{Correctness and Invocations of \EntailsAlg (Algorithm~\ref{alg:entails})}
\label{app:proof_entails}

The following theorem is equivalent to Theorem~\ref{thm:entails_soundness_main} since we inline the specification of Algorithm~\ref{alg:entails}.

\begin{theorem}[\EntailsAlg correctness and invocations]
\label{thm:entails_correctness_complexity}
	For every call \entails{$\theory,\xadd_1,\xadd_2,\forma$}, the returned value is $\true$ if and only if every $\theory$-interpretation $\inta$ satisfying $\forma$ satisfies
	\[
		\sem{\xadd_1}{\inta}\leq\sem{\xadd_2}{\inta}.
	\]
	Consequently, the initial call with $\forma=\true$ satisfies the specification in \Cref{alg:entails}.

	Let $r_i$ be the root of $\xadd_i$, and let $m_i=\mathrm{rank}_{\xadd_i}(r_i)$ be the root rank from \Cref{def:tedd_rank}. Then \EntailsAlg makes at most $2^{m_1+m_2+1}-1$ invocations, counting the initial invocation. This bound is an invocation bound, not a cache-miss bound for recursive subproblems, because \EntailsAlg caches terminal SMT queries but does not cache recursive triples $(\theory,\xadd_1,\xadd_2,\forma)$.
\end{theorem}

\begin{proof}
	We first prove correctness and then establish the invocation bound.
	The key idea is to enumerate paired paths while carrying the current path condition.  At terminal pairs, the SMT query checks whether $\forma'\wedge\xlabt(\xadd_1')>\xlabt(\xadd_2')$ is satisfiable, i.e. whether a counterexample to the entailment exists.  The terminal SMT cache is sound because it stores exactly these query results.

	Correctness is by induction on the tuple rank of the current TEDD pair.  At rank $0$, the negated satisfiability result is true exactly when no interpretation satisfying the accumulated context violates the inequality.  At positive rank, the algorithm splits on the $\varord$-least current root label $\formb$, forming contexts $\forma'\wedge\formb$ and $\forma'\wedge\neg\formb$ and the corresponding successor pairs.  These two contexts partition the interpretations satisfying $\forma'$, and both recursive calls have smaller tuple rank.  Therefore their conjunction is true exactly when the inequality holds under the original context.

	For the invocation bound, the recursion tree is binary, and along each edge at least one TEDD advances to a proper successor.  A branch therefore has length at most $m_1+m_2$, so the tree has at most $2^{m_1+m_2+1}-1$ invocations.
\end{proof}

\subsection[Soundness of Syntactic Weakest Pre-Expectations]{Soundness of Syntactic Weakest Pre-Expectations (Theorem~\ref{thm:xwp_sound})}
\label{app:proof_xwp_sound}

\begin{proof}[Proof of \Cref{thm:xwp_sound}]
	Fix a structure $\struct$. For a $\sorteureal$-TEDD $\xadd$, write
	\[
		\mathcal{E}_{\xadd}^{\struct}
		\eeq
		\mylambda{\vala} \sem{\xadd}{(\struct,\vala)} .
	\]
	We prove the following statement by structural induction on the loop-free program $\cc$: for every $\sorteureal$-TEDD $\xadd$ and every valuation $\vala$,
	\[
		\sem{\xwp{\cc}{\xadd}}{(\struct,\vala)}
		\eeq
		\wp{\cc}{\mathcal{E}_{\xadd}^{\struct}}(\vala).
	\]
	For guards, we use the following fact. The Boolean TEDD $\xadd_\forma$ from \Cref{tab:xwp} satisfies
	\begin{align*}
		&\sem{\xadd_\forma}{(\struct,\vala)}=\true
		\quad \text{iff}\quad
		(\struct,\vala)\models\forma .
	\tag{\text{\Cref{tab:xwp}, Thm.~\ref{thm:apply_correctness_complexity}}}
	\end{align*}
	This follows by induction on $\forma$ from the construction of $\xadd_\forma$ using \ApplyAlg on Boolean connectives.

	\emph{The case} $\cc=\SKIP$.
	\begin{align*}
		&\sem{\xwp{\SKIP}{\xadd}}{(\struct,\vala)}\\
		\eeq & \sem{\xadd}{(\struct,\vala)}
		\tag{\text{\Cref{tab:xwp}}}\\
		\eeq & \mathcal{E}_{\xadd}^{\struct}(\vala)
		\tag{\text{def. }E}\\
		\eeq & \wp{\SKIP}{\mathcal{E}_{\xadd}^{\struct}}(\vala).
		\tag{\text{\Cref{tab:wp}}}
	\end{align*}

	\emph{The case} $\cc=\ASSIGN{\vara}{\terma}$.
	\begin{align*}
		&\sem{\xwp{\ASSIGN{\vara}{\terma}}{\xadd}}{(\struct,\vala)}\\
		\eeq & \sem{\substitute(\xadd,\vara,\terma)}{(\struct,\vala)}
		\tag{\text{\Cref{tab:xwp}}}\\
		\eeq & \sem{\xswitchfun{\xadd}\switchfunsubst{\vara}{\terma}}{(\struct,\vala)}
		\tag{\text{Thm.~\ref{thm:substitute_correctness_complexity}}}\\
		\eeq & \sem{\xadd}{(\struct,\vala\valsubst{\vara}{\sem{\terma}{(\struct,\vala)}})}
		\tag{\text{substitution semantics, \Cref{sec:caseexpr}}}\\
		\eeq & \mathcal{E}_{\xadd}^{\struct}\bigl(\vala\valsubst{\vara}{\sem{\terma}{(\struct,\vala)}}\bigr)
		\tag{\text{def. }E}\\
		\eeq & \wp{\ASSIGN{\vara}{\terma}}{\mathcal{E}_{\xadd}^{\struct}}(\vala).
		\tag{\text{\Cref{tab:wp}}}
	\end{align*}

	\emph{The case} $\cc=\TICK{\tickrew}$.
	\begin{align*}
		&\sem{\xwp{\TICK{\tickrew}}{\xadd}}{(\struct,\vala)}\\
		\eeq & \sem{\apply(+, \obtain(\tickrew),\xadd)}{(\struct,\vala)}
		\tag{\text{\Cref{tab:xwp}}}\\
		\eeq & \sem{\obtain(\tickrew)}{(\struct,\vala)}+\sem{\xadd}{(\struct,\vala)}
		\tag{\text{Thm.~\ref{thm:apply_correctness_complexity}}}\\
		\eeq & \tickrew+\mathcal{E}_{\xadd}^{\struct}(\vala)
		\tag{\text{\ObtainAlg, def. }E}\\
		\eeq & \wp{\TICK{\tickrew}}{\mathcal{E}_{\xadd}^{\struct}}(\vala).
		\tag{\text{\Cref{tab:wp}}}
	\end{align*}

	\emph{The case} $\cc=\OBSERVE{\forma}$.
	\begin{align*}
		&\sem{\xwp{\OBSERVE{\forma}}{\xadd}}{(\struct,\vala)}\\
		\eeq & \sem{\apply(\switchITEsymbol,\xadd_\forma,\xadd,\obtain(0))}{(\struct,\vala)}
		\tag{\text{\Cref{tab:xwp}}}\\
		\eeq &
		\begin{cases}
			\sem{\xadd}{(\struct,\vala)} & \text{if }(\struct,\vala)\models\forma,\\
			\sem{\obtain(0)}{(\struct,\vala)} & \text{if }(\struct,\vala)\models\neg\forma
		\end{cases}
		\tag{\text{guard, Thm.~\ref{thm:apply_correctness_complexity}}}\\
		\eeq &
		\begin{cases}
			\mathcal{E}_{\xadd}^{\struct}(\vala) & \text{if }(\struct,\vala)\models\forma,\\
			0 & \text{if }(\struct,\vala)\models\neg\forma
		\end{cases}
		\tag{\text{\ObtainAlg, def. }E}\\
		\eeq & \wp{\OBSERVE{\forma}}{\mathcal{E}_{\xadd}^{\struct}}(\vala).
		\tag{\text{\Cref{tab:wp}}}
	\end{align*}

	\emph{The case} $\cc=\COMPOSE{\cc_1}{\cc_2}$.
	Assume the induction hypothesis for $\cc_1$ and $\cc_2$.
	\begin{align*}
		&\sem{\xwp{\COMPOSE{\cc_1}{\cc_2}}{\xadd}}{(\struct,\vala)}\\
		\eeq & \sem{\xwp{\cc_1}{\xwp{\cc_2}{\xadd}}}{(\struct,\vala)}
		\tag{\text{\Cref{tab:xwp}}}\\
		\eeq & \wp{\cc_1}{\mathcal{E}_{\xwp{\cc_2}{\xadd}}^{\struct}}(\vala)
		\tag{\text{I.H. on C1}}\\
		\eeq & \wp{\cc_1}{\wp{\cc_2}{\mathcal{E}_{\xadd}^{\struct}}}(\vala)
		\tag{\text{I.H. on C2}}\\
		\eeq & \wp{\COMPOSE{\cc_1}{\cc_2}}{\mathcal{E}_{\xadd}^{\struct}}(\vala).
		\tag{\text{\Cref{tab:wp}}}
	\end{align*}

	\emph{The case} $\cc=\NDCHOICE{\cc_1}{\cc_2}$.
	Assume the induction hypothesis for $\cc_1$ and $\cc_2$.
	\begin{align*}
		&\sem{\xwp{\NDCHOICE{\cc_1}{\cc_2}}{\xadd}}{(\struct,\vala)}\\
		\eeq & \sem{\minfunc(\xwp{\cc_1}{\xadd},\xwp{\cc_2}{\xadd})}{(\struct,\vala)}
		\tag{\text{\Cref{tab:xwp}}}\\
		\eeq & \min\{\sem{\xwp{\cc_1}{\xadd}}{(\struct,\vala)},\sem{\xwp{\cc_2}{\xadd}}{(\struct,\vala)}\}
		\tag{\text{Thm.~\ref{thm:minimum_correctness_complexity}}}\\
		\eeq & \min\{\wp{\cc_1}{\mathcal{E}_{\xadd}^{\struct}}(\vala),\wp{\cc_2}{\mathcal{E}_{\xadd}^{\struct}}(\vala)\}
		\tag{\text{I.H. on C1,C2}}\\
		\eeq & \bigl(\wp{\cc_1}{\mathcal{E}_{\xadd}^{\struct}}\einf\wp{\cc_2}{\mathcal{E}_{\xadd}^{\struct}}\bigr)(\vala)
		\tag{\text{def. inf}}\\
		\eeq & \wp{\NDCHOICE{\cc_1}{\cc_2}}{\mathcal{E}_{\xadd}^{\struct}}(\vala).
		\tag{\text{\Cref{tab:wp}}}
	\end{align*}

	\emph{The case} $\cc=\PCHOICE{\cc_1}{\proba}{\cc_2}$.
	Assume the induction hypothesis for $\cc_1$ and $\cc_2$, and let $\xadd_1'=\xwp{\cc_1}{\xadd}$ and $\xadd_2'=\xwp{\cc_2}{\xadd}$.
	\begin{align*}
		&\sem{\xwp{\cc}{\xadd}}{(\struct,\vala)}\\
		\eeq &
		\sem{\apply(+, \apply(\cdot,\obtain(\proba),\xadd_1'),\apply(\cdot,\obtain(1-\proba),\xadd_2'))}{(\struct,\vala)}
		\tag{\text{\Cref{tab:xwp}}}\\
		\eeq &
		\proba\cdot\sem{\xadd_1'}{(\struct,\vala)}
		+(1-\proba)\cdot\sem{\xadd_2'}{(\struct,\vala)}
		\tag{\text{\ObtainAlg, Thm.~\ref{thm:apply_correctness_complexity}}}\\
		\eeq &
		\proba\cdot\wp{\cc_1}{\mathcal{E}_{\xadd}^{\struct}}(\vala)
		+(1-\proba)\cdot\wp{\cc_2}{\mathcal{E}_{\xadd}^{\struct}}(\vala)
		\tag{\text{I.H. on C1,C2}}\\
		\eeq & \wp{\cc}{\mathcal{E}_{\xadd}^{\struct}}(\vala).
		\tag{\text{\Cref{tab:wp}}}
	\end{align*}

	\emph{The case} $\cc=\ITE{\forma}{\cc_1}{\cc_2}$.
	Assume the hypothesis statement for $\cc_1$ and $\cc_2$.
	\begin{align*}
		&\sem{\xwp{\ITE{\forma}{\cc_1}{\cc_2}}{\xadd}}{(\struct,\vala)}\\
		\eeq & \sem{\apply(\switchITEsymbol,\xadd_\forma,\xwp{\cc_1}{\xadd},\xwp{\cc_2}{\xadd})}{(\struct,\vala)}
		\tag{\text{\Cref{tab:xwp}}}\\
		\eeq &
		\begin{cases}
			\sem{\xwp{\cc_1}{\xadd}}{(\struct,\vala)} & \text{if }(\struct,\vala)\models\forma,\\
			\sem{\xwp{\cc_2}{\xadd}}{(\struct,\vala)} & \text{if }(\struct,\vala)\models\neg\forma
		\end{cases}
		\tag{\text{guard, Thm.~\ref{thm:apply_correctness_complexity}}}\\
		\eeq &
		\begin{cases}
			\wp{\cc_1}{\mathcal{E}_{\xadd}^{\struct}}(\vala) & \text{if }(\struct,\vala)\models\forma,\\
			\wp{\cc_2}{\mathcal{E}_{\xadd}^{\struct}}(\vala) & \text{if }(\struct,\vala)\models\neg\forma
		\end{cases}
		\tag{\text{I.H. on C1,C2}}\\
		\eeq & \wp{\ITE{\forma}{\cc_1}{\cc_2}}{\mathcal{E}_{\xadd}^{\struct}}(\vala).
		\tag{\text{\Cref{tab:wp}}}
	\end{align*}
\end{proof}

\section{Simultaneous Loop Unrolling}
\label{app:simultaneous_unrolling}
We briefly recall the simultaneous loop-unrolling construction used by
\citet[Sec.~7.2 and App.~B.2]{DBLP:journals/corr/abs-2502-19388}.  Write
$\ABORT$ as shorthand for $\OBSERVE{\false}$, so $\wp{\ABORT}{\FF}=\expzero$.
For $\ell\in\Nats$, the $\ell$-fold simultaneous unrolling
$\cc^{\langle \ell\rangle}$ is defined by structural recursion.  Atomic
commands are left unchanged:
\begin{align*}
	\SKIP^{\langle \ell\rangle}=\SKIP,\qquad
	(\ASSIGN{\vara}{\terma})^{\langle \ell\rangle}
	&=\ASSIGN{\vara}{\terma},
	\\
	(\TICK{\tickrew})^{\langle \ell\rangle}=\TICK{\tickrew},\qquad
	(\OBSERVE{\forma})^{\langle \ell\rangle}
	&=\OBSERVE{\forma}.
\end{align*}
The non-looping compound constructs commute with unrolling:
\begin{align*}
	(\COMPOSE{\cc_1}{\cc_2})^{\langle \ell\rangle}
	&= \COMPOSE{\cc_1^{\langle \ell\rangle}}{\cc_2^{\langle \ell\rangle}},
	\\
	(\PCHOICE{\cc_1}{\proba}{\cc_2})^{\langle \ell\rangle}
	&= \PCHOICE{\cc_1^{\langle \ell\rangle}}{\proba}{\cc_2^{\langle \ell\rangle}},
	\\
	(\NDCHOICE{\cc_1}{\cc_2})^{\langle \ell\rangle}
	&= \NDCHOICE{\cc_1^{\langle \ell\rangle}}{\cc_2^{\langle \ell\rangle}},
	\\
	(\ITE{\forma}{\cc_1}{\cc_2})^{\langle \ell\rangle}
	&= \ITE{\forma}{\cc_1^{\langle \ell\rangle}}{\cc_2^{\langle \ell\rangle}}.
\end{align*}
For loops, the unrolling budget is consumed globally:
\begin{align*}
	(\WHILEDO{\forma}{\cc})^{\langle 0\rangle}
	&=\ABORT,
	\\
	(\WHILEDO{\forma}{\cc})^{\langle \ell+1\rangle}
	&=
	\IFSYMBOL(\forma)\left\{
		\COMPOSE{\cc^{\langle \ell\rangle}}
		        {(\WHILEDO{\forma}{\cc})^{\langle \ell\rangle}}
	\right\}
	\\[-0.5ex]
	&\quad
	\ELSESYMBOL\left\{\SKIP\right\}.
\end{align*}

Since the weakest pre-expectation transformer is $\omega$-continuous,
these loop-free approximants converge to the semantics of the original program:
\[
	\wp{\cc}{\FF}
	=
	\bigesup_{\ell\in\Nats}
	\wp{\cc^{\langle \ell\rangle}}{\FF}.
\]
For a single loop with loop-free body, this is exactly the characteristic-function
iteration from \Cref{sec:wp:loops}.  For nested loops, the same fixed-point
iteration can therefore be implemented by applying the loop-free
$\xwpsymbol$ transformer to the simultaneously unrolled program
$\cc^{\langle \ell\rangle}$.

\clearpage
\newpage

\section{Additional Benchmark Results}
\label{app:additional_benchmark_results}
We present additional benchmark results in \Cref{fig:results:coupon:scale:caesar}, \Cref{fig:results:nondet:scale:caesar}, and \Cref{fig:results:cond:scale:caesar}.

\begin{figure}
	\centering
\input{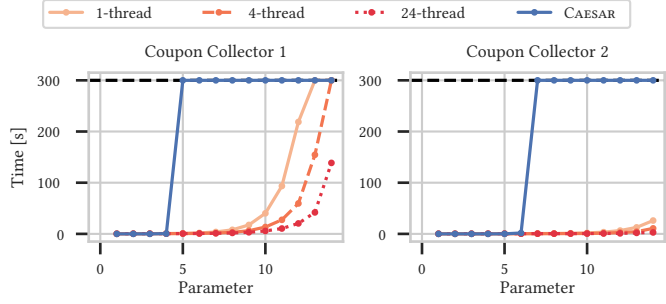}
\vspace{-1em}
\caption{Scalability comparison with \caesar{} on two more benchmarks that use arrays.}
\label{fig:results:coupon:scale:caesar}
\end{figure}

\begin{figure}
	\centering
\input{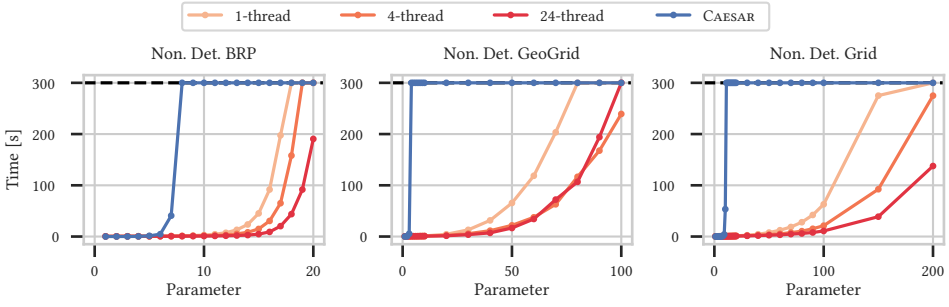}
\vspace{-1em}
\caption{Scalability comparison with \caesar{} on programs utilizing non-determinism.}
\label{fig:results:nondet:scale:caesar}
\end{figure}

\begin{figure}
	\centering
\input{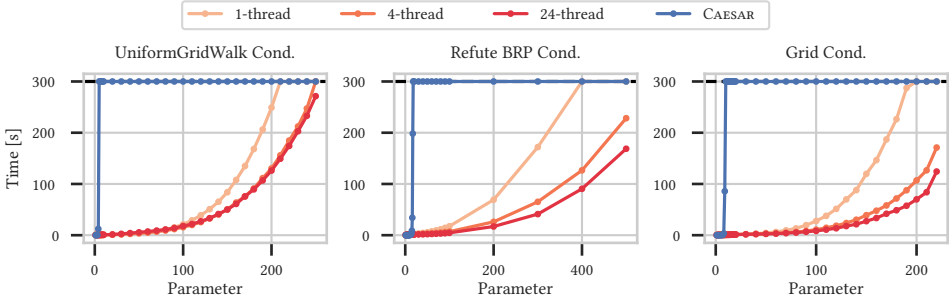}
\vspace{-1em}
\caption{Scalability comparison with \caesar{} on programs utilizing conditioning.}
\label{fig:results:cond:scale:caesar}
\end{figure}

\clearpage
\newpage

\section{Benchmark Programs}
\label{app:benchmark_programs}
\lstset{language=pgcl,mathescape=true, numbers=left, frame=lines, xleftmargin=4.8mm, framerule=0.300mm, framexleftmargin=4.84mm}
\toolddsolve uses two different dialects of PGCL: an untyped variant where every variable is an unsigned integer
(\cref{pgcl:tbd,pgcl:brp,pgcl:grid,pgcl:geo_grid,pgcl:rabin,pgcl:nd:brp,pgcl:nd:geo_grid,pgcl:nd:grid,pgcl:grid-cond,pgcl:refute-brp-cond,pgcl:tbd-cond})
and a typed variant that supports more features of the underlying SMT solver
(\cref{pgcl:coupon1,pgcl:lossy_sum,pgcl:lshift,pgcl:rshift,pgcl:square,pgcl:nd:grow,pgcl:plane}).
Both variants use annotations (\lstinline{@fixpoint}, \lstinline{@unroll(bound)}, \lstinline{@kinduction(invariant)}) on loops to specify the proof rule that should be used. The C++-style comments in the first and last line give the bound on the weakest pre-expectation to verify and the post-expectation, respectively.
\Cref{pgcl:tbd,pgcl:coupon1,pgcl:tbd-cond} sample variables from uniform distributions through the statement \lstinline{x ~= uniform(a,b)} where x is uniformly sampled from the half-open interval $[a,b)$.
All listings but \cref{pgcl:square} define families of programs based on the parameter \textit{bnd}. \Cref{pgcl:coupon1} has an additional parameter \textit{unrbnd}. \cite{DBLP:conf/birthday/GrooteV17}

Many of the programs are hand translations of programs from the literature.
We list the benchmarks and the source where we got the untranslated version, as well as any changes that we made to that version:
\begin{itemize}
  \item We created \texttt{UniformGridWalk} (\cref{pgcl:tbd}) as a generalization of the \texttt{Grid} program (listed below).
  \item We took \texttt{BRP} (\cref{pgcl:brp}) from \cite{DBLP:journals/pacmpl/SchroerBKKM23} and added the parameter \textit{bnd}. This formulation of \texttt{BRP} is originally from \cite[p. 79]{spelMonotonicityMarkovChains2018}.
  \item We took \texttt{Grid} (\cref{pgcl:grid}) from \cite{DBLP:conf/tacas/BatzCJKKM23} and added the parameter \textit{bnd}.
  \item We created \texttt{GeoGrid} (\cref{pgcl:geo_grid}) as a generalization of the \texttt{Grid} program (listed above).
  \item We took \texttt{Rabin} (\cref{pgcl:rabin}) from \cite{DBLP:journals/pacmpl/SchroerBKKM23} and added the parameter \textit{bnd}.
  \item We took \texttt{Coupon Collector} (\cref{pgcl:coupon1}) from \cite{DBLP:journals/jacm/KaminskiKMO18} and added the parameter \textit{bnd} and \textit{unrbnd}.
  \item We took \texttt{Plane} (\cref{pgcl:plane}) from \cite{DBLP:conf/birthday/GrooteV17}.
  \item We created \texttt{Grid Cond.} (\cref{pgcl:grid-cond}), \texttt{Refute BRP Cond.} (\cref{pgcl:refute-brp-cond}), and \texttt{UniformGridWalk Cond.} (\cref{pgcl:tbd-cond}) by modifying \texttt{Grid}, \texttt{BRP}, and \texttt{UniformGridWalk} (listed above), respectively.
  \item We created \texttt{Non. Det. BRP} (\cref{pgcl:nd:brp}), \texttt{Non. Det. GeoGrid} (\cref{pgcl:nd:geo_grid}), and \texttt{Non. Det. Grid} (\cref{pgcl:nd:grid}) by modifying \texttt{BRP}, \texttt{GeoGrid} and\texttt{Grid} (listed above), respectively.
\end{itemize}
\vspace{0pt plus 10em}

\subsection{Piecewise Linear Programs}
\noindent%
\begin{minipage}{\linewidth}
\begin{lstlisting}[caption={\texttt{UniformGridWalk}}, label={pgcl:tbd}]
// wp < $\mathit{bnd}$
a:=0; b:=0; k:=0;
@fixpoint
while(a < $\mathit{bnd}$ & b < $\mathit{bnd}$ & k < $\mathit{bnd}$) {
  {a ~= uniform(0,$\mathit{bnd}+1$);}[1/2]{b ~= uniform(0,$\mathit{bnd}+1$);};
  k:=k+1;
}
// a
\end{lstlisting}
\end{minipage}

\noindent%
\begin{minipage}{\linewidth}
\begin{lstlisting}[caption={\texttt{BRP}}, label={pgcl:brp}]
// wp <= ite(toSend <= $\mathit{bnd}$, totalFailed+10, $\infty$)
@kinduction(ite(toSend <= $\mathit{bnd}$, totalFailed+10, $\infty$))
while(failed < maxFailed & sent < toSend){
  {
    failed:=0; sent:=sent+1;
  }[9/10]{
    failed:=failed+1; totalFailed:=totalFailed+1;
  }
}
// totalFailed
\end{lstlisting}
\end{minipage}

\noindent%
\begin{minipage}{\linewidth}
\begin{lstlisting}[caption={\texttt{Grid}}, label={pgcl:grid}]
// wp <= $\mathit{bnd}$
a:=0; b:=0;
@fixpoint
while(a < $\mathit{bnd}$ & b < $\mathit{bnd}$) {
  {a:=a+1;}[1/2]{b:=b+1;}
}
// a
\end{lstlisting}
\end{minipage}

\noindent%
\begin{minipage}{\linewidth}
\begin{lstlisting}[caption={\texttt{GeoGrid}}, label={pgcl:geo_grid}]
// wp <= $\mathit{bnd}$
run:=1; a:=0; b:=0;
@kinduction(a+$\mathit{bnd}$)
while(run > 0) {
  {a:=a+1;}[1/2]{b:=b+1;}
  {}[1/2]{
    if(!(a < $\mathit{bnd}$ & b < $\mathit{bnd}$)) {
      run:=0;
    } else {skip;}
  }
}
// a
\end{lstlisting}
\end{minipage}

\noindent%
\begin{minipage}{\linewidth}
\begin{lstlisting}[caption={\texttt{Rabin}}, label={pgcl:rabin}]
// wp <= ite(1<i & i<4 & phase == 0, (5/7)+(1/$10^{\mathit{bnd}}$), 1)
@kinduction(ite(1<i & i<4 & phase == 0, (5/7)+(1/$10^{\mathit{bnd}}$), 1))
while(1<i | phase == 1){
  if(phase == 0){
    n:=i; phase:=1;
  } else {
    if(0<n){
      {d:=0;}[1/2]{d:=1;}
      i:=i-d; n:=n-1;
    } else {
      phase:=0;
    }
  }
}
// ite(i == 1, 1, 0)
\end{lstlisting}
\end{minipage}

\subsection{Array Programs}
\noindent%
\begin{minipage}{\linewidth}
\begin{lstlisting}[caption={\texttt{Coupon Collector}}, label={pgcl:coupon1}]
// wp <= ($\mathit{bnd}$ * $\mathit{bnd}$) * (1/2)
int[] a; uint k; uint z; uint i;
a[0]:=0; $\cdots$ ; a[$\mathit{bnd}$]:=0;
i:=0; z:=0;
@fixpoint
while(z < $\mathit{bnd}$) {
  @unroll($\mathit{unrbnd}$)
  while(a[i] != 0) {
    i ~= uniform(0, $\mathit{bnd}$);
    cost(1);
  }
  a[i]:=1; z:=z+1;
}
// 0
\end{lstlisting}
\end{minipage}

\noindent%
\begin{minipage}{\linewidth}
\begin{lstlisting}[caption={\texttt{Lossy Sum}}, label={pgcl:lossy_sum}]
// wp <= 1/2*(A[0]+$\ldots$+A[$\mathit{bnd - 1}$])
int[] A; uint j; uint x;
x:=0; j:=0;
@unroll($\mathit{bnd}+1$)
while(j < bnd) {
  {x:=x+A[j];}[1/2]{}
  j:=j+1;
}
// x
\end{lstlisting}
\end{minipage}

\noindent%
\begin{minipage}{\linewidth}
\begin{lstlisting}[caption={\texttt{LShift}}, label={pgcl:lshift}]
// wp <= ite(a[0] == a[$\mathit{bnd}$],
             ite(a[1] == a[$\mathit{bnd}$],1,1/2),
             ite(a[1] == a[$\mathit{bnd}$],1/2,0))
int[] a; uint i;
i:=0;
@fixpoint
while (i < $\mathit{bnd}$) {
  {a[i]:=a[i+1];}[1/2]{}
  i:=i+1;
}
// ite(a[0] == a[$\mathit{bnd}$],1,0)
\end{lstlisting}
\end{minipage}

\noindent%
\begin{minipage}{\linewidth}
\begin{lstlisting}[caption={\texttt{RShift}}, label={pgcl:rshift}]
// wp <= ite(a[0] == a[$\mathit{bnd}$],1,1/2)
int[] a; uint i;
i:=0;
@fixpoint
while (i < $\mathit{bnd}$) {
  {a[i+1]:=a[i];}[1/2]{}
  i:=i+1;
}
// ite(a[0] == a[$\mathit{bnd}$],1,0)
\end{lstlisting}
\end{minipage}

\subsection{Non-Linear Programs}

\noindent%
\begin{minipage}{\linewidth}
\begin{lstlisting}[caption={\texttt{Multiply}}, label={pgcl:square}]
// wp <= m * n
uint r; uint i; uint n; uint m;
r:=0; i:=0;
@kinduction(ite(i < m, r+(m-i)*n, r))
while(i < m) {
  r:=r+n; i:=i+1;
}
// r
\end{lstlisting}
\end{minipage}

\dbinline{Maybe remove \cref{pgcl:square}?}
\noindent%
\begin{minipage}{\linewidth}
\begin{lstlisting}[caption={\texttt{Non. Det. Grow}}, label={pgcl:nd:grow}]
// wp <= r
uint i; uint a; uint r;
i:=0;
@fixpoint
while(i<$\mathit{bnd}$) {
  {a:=a+1;}[]{a:=a*a;}
  i:=i+1;
}
if (r > a) {
  r:=r-a;
} else {skip;}
// r
\end{lstlisting}
\end{minipage}

\noindent%
\begin{minipage}{\linewidth}
\begin{lstlisting}[caption={\texttt{Plane}}, label={pgcl:plane}]
// wp <= 1/2
uint everybody_own_seat; uint i;
{everybody_own_seat:=1;}[1/$\mathit{bnd}$]{everybody_own_seat:=0;}
i:=1;
@fixpoint
while(i < $\mathit{bnd}-1$) {
  if (everybody_own_seat == 0) {
    {everybody_own_seat:=1;}[(1/($\mathit{bnd}$-i))*(1/($\mathit{bnd}$-i))]{}
  } else {skip;}
  i:=i+1;
}
// everybody_own_seat
\end{lstlisting}
\end{minipage}

\dbinline{We need to cite some of these programs as these are just modifications.}

\subsection{Conditioned Programs}

\noindent%
\begin{minipage}{\linewidth}
\begin{lstlisting}[caption={\texttt{Grid Cond.}}, label={pgcl:grid-cond}]
// cwp <= ($\mathit{bnd}$/2,1/$\mathit{bnd}$)
a:=0; b:=0;
@fixpoint
while (a < $\mathit{bnd}$ & b < $\mathit{bnd}$) {
  {a:=a+1;}[1/2]{b:=b+1;}
  observe(a <= b);
}
// (a,1)
\end{lstlisting}
\end{minipage}

\noindent%
\begin{minipage}{\linewidth}
\begin{lstlisting}[caption={\texttt{Refute BRP Cond.}}, label={pgcl:refute-brp-cond}]
// cwp <= ($\infty$, 1/2)
@unroll(bnd)
while(failed < maxFailed & sent < toSend){
  {
    failed:=0; sent:=sent+1
  }[9/10]{
    failed:=failed+1; totalFailed:=totalFailed+1
  }
  observe(failed < 0);
}
// (totalFailed, 1)
\end{lstlisting}
\end{minipage}

\noindent%
\begin{minipage}{\linewidth}
\begin{lstlisting}[caption={\texttt{UniformGridWalk Cond.}}, label={pgcl:tbd-cond}]
// cwp < ($\mathit{bnd}$, 1/2)
a:=0; b:=0; k:=0;
@fixpoint
while(a < $\mathit{bnd}$ & b < $\mathit{bnd}$ & k < $\mathit{bnd}$) {
  {
    a ~= uniform(0,$\mathit{bnd} + 1$);
  }[1/2]{
    b ~= uniform(0,$\mathit{bnd} + 1$);
  };
  k:=k + 1;
  {observe(a < $\mathit{bnd}$)}[1/2]{observe(b < $\mathit{bnd}$)}
}
// (a, 1)
\end{lstlisting}
\end{minipage}

\subsection{Non-Deterministc Programs}

\noindent%
\begin{minipage}{\linewidth}
\begin{lstlisting}[caption={\texttt{Non. Det. BRP}}, label={pgcl:nd:brp}]
// wp <= ite(toSend <= $\mathit{bnd}$, totalFailed+10, $\infty$)
@kinduction(ite(toSend <= $\mathit{bnd}$, totalFailed+10, $\infty$))
while(failed < maxFailed & sent < toSend) {
  {
    {
      failed:=0; sent:=sent+1;
    }[9/10]{
      failed:=failed+1; totalFailed:=totalFailed+1
    }
  }[]{
    {
      failed:=0; sent:=sent+1;
    }[91/100]{
      failed:=failed+1; totalFailed:=totalFailed+1;
    }
  }
}
// totalFailed
\end{lstlisting}
\end{minipage}

\noindent%
\begin{minipage}{\linewidth}
\begin{lstlisting}[caption={\texttt{Non. Det. GeoGrid}}, label={pgcl:nd:geo_grid}]
// wp <= $\mathit{bnd}$
run:=1; a:=0; b:=0;
@kinduction(a + $\mathit{bnd}$)
while (run > 0) {
  {{a:=a+1;}[1/2]{b:=b+1}}[]{a:=a+1;}
  {}[1/2]{
    if (!(a < $\mathit{bnd}$ & b < $\mathit{bnd}$)) {
      run:=0;
    } else {skip;}
  }
}
// a
\end{lstlisting}
\end{minipage}

\noindent%
\begin{minipage}{\linewidth}
\begin{lstlisting}[caption={\texttt{Non. Det. Grid}}, label={pgcl:nd:grid}]
// wp <= $\mathit{bnd} + \mathit{bnd} - 2$
a:=0; b:=0;
@fixpoint
while (a < $\mathit{bnd}$ & b < $\mathit{bnd}$) {
  {a:=a+1}[]{b := b+1}
  cost(1);
}
// 0
\end{lstlisting}
\end{minipage}

\end{document}